\documentclass[envcountsame,runningheads,a4paper]{llncs}

\usepackage{microtype}
\usepackage{mathtools,amsthm}
\usepackage{enumerate}
\usepackage[hidelinks,bookmarks,bookmarksopen,bookmarksdepth=2]{hyperref}
\usepackage[utf8]{inputenc}

\usepackage{amssymb}
\usepackage{stmaryrd}
\usepackage{bussproofs}
\usepackage{mathtools}
\usepackage{bbold}

\usepackage{ifthen}
\usepackage{xspace}
\usepackage{fancybox}
\usepackage{hhline}

\usepackage{xcolor}

\newcommand{\ignore}[1]{}

\newcommand{\myinput}[1]{\ifthenelse{\boolean{withimages}}{\input{#1}}{}}

\newcommand{\reflemma}[1]{Lemma~\ref{l:#1}}
\newcommand{\reflemmap}[2]{Lemma~\ref{l:#1}.\ref{p:#1-#2}}

\newcommand{\reflemmasps}[3]{Lemmas~\ref{l:#1}.\ref{p:#1-#2}-\ref{p:#1-#3}} %Giulio
\newcommand{\reflemmapp}[3]{Lemma~\ref{l:#1}.\ref{p:#1-#2}.\ref{p:#1-#2-#3}}

\newcommand{\refpoint}[1]{Point~\ref{p:#1}}

\newcommand{\refeqpoint}[1]{Point~\eqref{p:#1}} %Giulio

\newcommand{\refthm}[1]{Thm.~\ref{thm:#1}}

\newcommand{\refprop}[1]{Prop.~\ref{prop:#1}}
\newcommand{\refpropp}[2]{Prop.~\ref{prop:#1}.\ref{p:#1-#2}} %Giulio
\newcommand{\refsect}[1]{Sect.~\ref{sect:#1}}
\newcommand{\refsects}[2]{Sect.~\ref{sect:#1}-\ref{sect:#2}}

\renewcommand{\refeq}[1]{(\ref{eq:#1})} %renewcommand to avoid conflict with package mathtools
\newcommand{\reffig}[1]{Fig.~\ref{fig:#1}}
\newcommand{\reffigs}[2]{Fig.~\ref{fig:#1} and \ref{fig:#2}} %Giulio
\newcommand{\refcoro}[1]{Cor.~\ref{coro:#1}}
\newcommand{\refcor}[1]{Cor.\,\ref{coro:#1}}
\newcommand{\refdef}[1]{Definition~\ref{def:#1}}

\newcommand{\refex}[1]{Ex.~\ref{ex:#1}}

\newcommand{\ie}{\textit{i.e.}\xspace}

\newcommand{\ih}{\textit{i.h.}\xspace}

\newcommand{\defeq}{\coloneqq} %Giulio
\newcommand{\eqdef}{\eqqcolon} %Giulio
\newcommand{\grameq}{\Coloneqq} %Giulio

\newcommand{\nat}{\mathbb{N}}
\newcommand{\size}[1]{|#1|}

\newcommand{\vsym}{\mathsf{v}}

\newcommand{\admsym}{{\mathtt c}}

\renewcommand{\l}{\lambda}
\newcommand{\isub}[2]{\{#1/#2\}}
\renewcommand{\isub}[2]{\{#1{\shortleftarrow}#2\}}
\newcommand{\esub}[2]{[#1/#2]}
\renewcommand{\esub}[2]{[#1{\shortleftarrow}#2]}
\newcommand{\fv}[1]{{\tt fv}(#1)}

\newcommand{\nf}{\mathsf{nf}} %Giulio
\newcommand{\nfo}[1]{\nf_\osym(#1)} %Giulio
\newcommand{\rootRew}[1]{\mapsto_{#1}}
\newcommand{\Rew}[1]{\rightarrow_{#1}}

\newcommand{\lRew}[1]{\; \mbox{}_{#1}{\leftarrow}\ }
\newcommand{\lRewn}[1]{\; \mbox{}^{*}_{#1}{\leftarrow}\ }

\newcommand{\rtobabs}{\rootRew{\betaabs}} 
 
\newcommand{\rtoin}{\rootRew{\inert}}

\newcommand{\betav}{{\beta_v}} %Giulio
\newcommand{\abssym}{\lambda} 
 
\newcommand{\betaabs}{{\beta_{\!\abssym}}} 

\newcommand{\tobv}{\Rew{\betav}} %Giulio
\newcommand{\tobabs}{\Rew{\betaabs}} 
 
\newcommand{\betain}{\beta_{\isym}} %Giulio
\newcommand{\inert}{\betain} %Giulio
\newcommand{\toin}{\Rew{\inert}}

\newcommand{\isym}{i}
\newcommand{\rsym}{{\mathtt r}}

\newcommand{\fsym}{f}

\newcommand{\subsym}{{\mathsf{sub}}}

\newcommand{\ssym}{{\mathtt s}}
\newcommand{\shufeqext}{\shufeqext} %Giulio
\newcommand{\tm}{t}
\newcommand{\tmtwo}{u}
\newcommand{\tmthree}{s}
\newcommand{\tmfour}{r}
\newcommand{\tmfive}{q}
\newcommand{\tmsix}{p}

\newcommand{\tmtwop}{\tmtwo'}
\newcommand{\tmthreep}{\tmthree'}

\newcommand{\var}{x}
\newcommand{\vartwo}{y}
\newcommand{\varthree}{z}

\newcommand{\ctxholep}[1]{\langle #1\rangle}
\newcommand{\ctxhole}{\ctxholep{\cdot}}

\newcommand{\ctx}{C}

\newcommand{\arbctxp}[1]{\arbctxp{#1}}
\newcommand{\arbctxtwop}[1]{\arbctxtwop{#1}}

\newcommand{\genevctx}{E}

\newcommand{\evctx}{\genevctx}
\newcommand{\evctxtwo}{\evctx'}

\newcommand{\evctxp}[1]{\evctx\ctxholep{#1}}

\newcommand{\revctx}{R}

\newcommand{\tomachhole}[1]{\leadsto_{#1}}
\newcommand{\tomach}{\tomachhole{}}
\newcommand{\tomachm}{\tomachhole{\mulsym}}

\newcommand{\tomachmone}{\tomachhole{\mulsym_1}}
\newcommand{\tomachmtwo}{\tomachhole{\mulsym_2}}

\newcommand{\tomacho}{\tomachhole{\osym}}
\newcommand{\tomachb}{\tomachhole{\beta}} %Giulio

\newcommand{\tomachc}{\tomachhole{\admsym}}
\newcommand{\tomachcp}[1]{\tomachhole{\admsym#1}}
\newcommand{\tomachcone}{\tomachhole{\admsym_1}}
\newcommand{\tomachctwo}{\tomachhole{\admsym_2}}
\newcommand{\tomachcthree}{\tomachhole{\admsym_3}}

\newcommand{\tomachsm}{\tomachhole{\msym}}

\newcommand{\tomachse}{\tomachhole{\esym}}

\newcommand{\code}{\overline{\tm}}
\newcommand{\codetwo}{\overline{\tmtwo}}

\newcommand{\genv}{E}
\newcommand{\genvtwo}{E'}
\newcommand{\genvthree}{E''}

\newcommand{\stctx}[1]{\ctx_{#1}} %Giulio
\newcommand{\stempty}{\epsilon}
\newcommand{\cons}{:}

\newcommand{\stack}{\pi}
\newcommand{\stacktwo}{\pi'}
\newcommand{\stackthree}{\pi''}

\newcommand{\decstack}{\decode{\stack}}

\newcommand{\decodestack}[2]{\ctxholep{#1}\decode{#2}}
\newcommand{\decstackp}[1]{\decodestack{#1}{\stack}}

\newcommand{\decodeinv}[2]{\ctxholep{#2}\decode{#1}}

\newcommand{\state}{s}
\newcommand{\statetwo}{s'}
\newcommand{\statethree}{s''}

\newcommand{\rename}[1]{\renamenop{(#1)}}
\newcommand{\renamenop}[1]{#1^\alpha}

\newcommand{\exec}{\rho}
\newcommand{\execp}{\rho'}
\newcommand{\execpp}{\rho''}
\newcommand{\exectwo}{\sigma}

\newcommand{\decode}[1]{\llbracket #1\rrbracket}
\renewcommand{\decode}[1]{\underline{#1}}
\newcommand{\sizebabs}[1]{\size{#1}_{\betaabs}} 
\newcommand{\sizein}[1]{\size{#1}_{\betain}}
\newcommand{\sizef}[1]{\size{#1}_{\betaf}}

\newcommand{\deriv}{d}
\newcommand{\derivp}{d'} %Giulio
\newcommand{\derivpp}{d''} %Giulio
\newcommand{\derivtwo}{e}

\newcommand{\sizehole}[2]{|#2|_{#1}}

\newcommand{\sizee}[1]{\sizehole{\expo}{#1}} %Giulio

\newcommand{\sizem}[1]{\sizehole{\mult}{#1}} %Giulio
\newcommand{\sizefree}[1]{\sizehole{\mathsf{free}}{#1}} %Giulio
\newcommand{\mach}{{\tt M}}

\renewcommand{\dump}{D}

\newcommand{\decdump}{\decode \dump}
\newcommand{\decdumpp}[1]{\decdump \ctxholep{#1}}

 \newcommand{\gconst}{i} 
\newcommand{\gconsttwo}{\gconst'}
\newcommand{\gconstthree}{\gconst''}

\newcommand{\fire}{f}
\newcommand{\firetwo}{\fire'}
\newcommand{\firethree}{\fire''}

\newcommand{\betaf}{\beta_{\!\fsym}} %Giulio
\newcommand{\rtof}{\rootRew{\betaf}}
\newcommand{\tof}{\Rew{\betaf}}
\newcommand{\torf}{\Rew{\rsym\betaf}}

\newcommand{\vsub}{{\vsym\subsym}} %Giulio
\newcommand{\unfsym}{\rotatebox[origin=c]{-90}{$\rightarrow$}}
\newcommand{\unf}[1]{#1\unfsym\,}
\newcommand{\relunf}[2]{\unf{#1}_{#2}}

\newcommand{\fireball}{fireball}

\newcommand{\dentry}[2]{#1\Diamond#2}

\newcommand{\csym}{{\mathtt c}}

\newcommand{\osym}{{\mathtt o}}

\newcommand{\la}[1]{\lambda #1.}

\newcommand{\glamst}[4]{(#1,#2,#3,#4)}
\newcommand{\glamsttab}[4]{#1 &#2 &#3 &#4} %Giulio
\newcommand{\eglamst}[4]{(#1,#2,#3,#4)}

\newcommand{\sizecom}[1]{|#1|_c}

\newcommand{\stackitem}{\phi}
\newcommand{\stackitemtwo}{\stackitem'}

\newcommand{\pair}[2]{#1@#2}

\newcommand{\decodep}[2]{\decode{#1}\ctxholep{#2}}

\newcommand{\myproof}[1]{
\ifthenelse{\boolean{omitproofs}}{\begin{IEEEproof} Proof available but omitted for readability. \end{IEEEproof}}{#1}}

\newcommand{\glamour}{GLA\-MOUr\xspace}
\newcommand{\eglamour}{Easy \glamour}
\newcommand{\fglamour}{Fast \glamour}
\newcommand{\uglamour}{Unchaining \glamour}

\newcommand{\gregoire}{Gr{\'{e}}goire\xspace}

\newcommand{\withproofs}[1]{\ifthenelse{\boolean{withproofs}}{#1}{}}

\newcommand{\withoutproofs}[1]{\ifthenelse{\boolean{withproofs}}{}{#1}}

\newcommand{\NoteProof}[1]{\withproofs{\marginpar{\scriptsize \ \ Proof p.\,{\pageref{#1}}}}} %Giulio
\newcommand{\NoteState}[1]{\withproofs{\marginpar{\scriptsize \ \ See p.~{\pageref{#1}}}}} %Giulio

\newcommand{\undef}{\bot}

\newcommand{\vsubcalc}{\lambda_\vsub}

\newcommand{\firecalc}{\lambda_\mathsf{fire}}

\newcommand{\Quiet}{Inert\xspace}

\newcommand{\doubt}[1]{}

\newenvironment{varitemize}
{
\begin{list}{\labelitemi}
{
\setlength{\itemsep}{.7pt}
 \setlength{\topsep}{.7pt}
 \setlength{\parsep}{.7pt}
 \setlength{\partopsep}{.7pt}
 \setlength{\leftmargin}{15pt}
 \setlength{\rightmargin}{0pt}
 \setlength{\itemindent}{0pt}
 \setlength{\labelsep}{5pt}
 \setlength{\labelwidth}{15pt}
}}
{
 \end{list} 
}

\newcounter{numberone}
\newcounter{numberoneroman}

\newenvironment{varenumerate}
{
\begin{list}{\arabic{numberone}.}
{
 \usecounter{numberone}
 \setlength{\itemsep}{.7pt}
 \setlength{\topsep}{.7pt}
 \setlength{\parsep}{.7pt}
 \setlength{\partopsep}{.7pt}
 \setlength{\leftmargin}{15pt}
 \setlength{\rightmargin}{0pt}
 \setlength{\itemindent}{0pt}
 \setlength{\labelsep}{5pt}
 \setlength{\labelwidth}{15pt}
}}
{
\end{list} 
}

\newenvironment{varenumerateroman}
{
\begin{list}{\roman{numberoneroman}.}
{
 \usecounter{numberoneroman}
 \setlength{\itemsep}{.7pt}
 \setlength{\topsep}{.7pt}
 \setlength{\parsep}{.7pt}
 \setlength{\partopsep}{.7pt}
 \setlength{\leftmargin}{15pt}
 \setlength{\rightmargin}{0pt}
 \setlength{\itemindent}{0pt}
 \setlength{\labelsep}{5pt}
 \setlength{\labelwidth}{15pt}
}}
{
\end{list} 
}

\newcommand{\cbv}{CbV\xspace}
\newcommand{\ocbv}{Open \cbv}
\newcommand{\ccbv}{Closed \cbv}
\newcommand{\scbv}{Strong \cbv}

\newcommand{\compil}[1]{#1^\circ}
\newcommand{\sizebeta}[1]{\size{#1}_\beta}
\newcommand{\tostrat}{\rightarrow}
\newcommand{\tomachine}{\tomachhole\mach}

\renewcommand{\tomachse}{\tomachhole{\ssym}}
\renewcommand{\tomachm}{\tomachhole{\beta}}
\renewcommand{\tomachsm}{\tomachm}
\renewcommand{\sizem}[1]{\sizebeta{#1}}
\renewcommand{\sizee}[1]{\sizehole{\ssym}{#1}}
\renewcommand{\sizecom}[1]{\sizehole{\csym}{#1}}

\renewcommand{\tomachmone}{\tomachhole{\beta_1}}
\renewcommand{\tomachmtwo}{\tomachhole{\beta_2}}

\renewcommand{\reffigs}[2]{Fig.~\ref{fig:#1}-\ref{fig:#2}} %Giulio
\renewcommand{\refdef}[1]{Def.~\ref{def:#1}}

\newboolean{withproofs}
\setboolean{withproofs}{true}

\newtheorem{lemmaAppendix}{Lemma} %Lemma in appendix
\newtheorem{theoremAppendix}{Theorem} %Theorem in appendix
\newtheorem{propositionAppendix}{Proposition} %Proposition in appendix
\newtheorem{corollaryAppendix}{Corollary} %Corollary in appendix

\theoremstyle{remark}
\begin{document}

% \special{papersize=8.5in,11in}
% \setlength{\pdfpageheight}{\paperheight}
% \setlength{\pdfpagewidth}{\paperwidth}
% 
% \conferenceinfo{PPDP'16}{Month d--d, 20yy, City, ST, Country}
% \copyrightyear{2016}
% \copyrightdata{978-1-nnnn-nnnn-n/yy/mm}
% \copyrightdoi{nnnnnnn.nnnnnnn}
% 
% % Uncomment the publication rights you want to use.
% %\publicationrights{transferred}
% %\publicationrights{licensed}     % this is the default
% %\publicationrights{author-pays}
% 
% \titlebanner{(Submitted to PPDP 2016)}        % These are ignored unless
% \preprintfooter{Accattoli, Guerrieri - Open Call-by-Value (Submitted to PPDP 2016)}   % 'preprint' option specified.
% 
% \title{Open Call-by-Value}
% % \subtitle{Subtitle Text, if any}
% 
% \authorinfo{Beniamino Accattoli}
%            {INRIA, UMR 7161, LIX, \'Ecole Polytechnique}
%            {\href{mailto:beniamino.accattoli@inria.fr}{beniamino.accattoli@inria.fr}}
% \authorinfo{Giulio Guerrieri}
%            {Aix-Marseille Université, %Centrale 
%            Marseille, I2M UMR 7373}
%            {\href{mailto:giulio.guerrieri@univ-amu.fr}{giulio.guerrieri@univ-amu.fr}, \href{mailto:gguerrieri@uniroma3.it}{gguerrieri@uniroma3.it}}

\mainmatter  % start of an individual contribution

\title{Implementing Open Call-by-Value (Extended Version)}

\titlerunning{Implementing Open Call-by-Value (Extended Version)}

\author{Beniamino Accattoli\inst{1}\and Giulio Guerrieri\inst{2}}
\authorrunning{B.~Accattoli\and G.~Guerrieri}


\institute{INRIA, UMR 7161, LIX, \'Ecole Polytechnique, \email{\href{mailto:beniamino.accattoli@inria.fr}{beniamino.accattoli@inria.fr}} \and
% Aix Marseille Université, CNRS, Centrale Marseille, 
% I2M UMR 7373, \\ F-13453 Marseille, France, 
% Aix Marseille Univ, CNRS, Centrale Marseille, I2M, Marseille, France,
% \email{\href{mailto:giulio.guerrieri@univ-amu.fr}{giulio.guerrieri@univ-amu.fr}}
University of Oxford, Department of Computer Science, Oxford, United Kingdom,
\email{\href{mailto:giulio.guerrieri@cs.ox.ac.uk}{giulio.guerrieri@cs.ox.ac.uk}}
}


\toctitle{Lecture Notes in Computer Science}
\tocauthor{Authors' Instructions}
\maketitle

% \begin{abstract}
% This is the text of the abstract.
% \end{abstract}

%\category{CR-number}{subcategory}{third-level}

% general terms are not compulsory anymore,
% you may leave them out
%\terms
%term1, term2

% \keywords
% $\l$-calculus, call-by-value, operational semantics, abstract machines, cost models

 % !TEX root = main.tex
\begin{abstract}
The theory of the call-by-value $\lambda$-calculus relies on weak evaluation and closed terms, that are natural hypotheses in the study of programming languages. To model proof assistants, however, strong evaluation and open terms are required. Open call-by-value is the intermediate setting of weak evaluation with open terms, on top of which \gregoire and Leroy designed the abstract machine of Coq. This paper provides a theory of abstract machines for open call-by-value. The literature contains machines that are either simple but inefficient, as they have an exponential overhead, or efficient but heavy, as they rely on a labelling of environments and a technical optimization. We introduce a machine that is simple and efficient: it does not use labels and it implements open call-by-value within a bilinear overhead. Moreover, we provide a new fine understanding of how different optimizations \mbox{impact on the complexity of the overhead.}
%  This work is part of a wider research effort, the COCA HOLA project \cite{COCAHOLA}. 
%  This work has been partially funded by the ANR JCJC grant COCA HOLA (ANR-16-CE40-004-01).
\end{abstract}

\begin{center}
This work is part of a wider research effort, the COCA HOLA project\\ \url{https://sites.google.com/site/beniaminoaccattoli/coca-hola}.
\end{center}

%  \cite{COCAHOLA}. 
%  This work has been partially funded by the ANR JCJC grant COCA HOLA (ANR-16-CE40-004-01).
 % !TEX root = main.tex
\section{Introduction}
\label{sect:intro}
The $\l$-calculus is the computational model behind functional programming languages and  proof assistants. A charming feature is that its definition is based on just one \emph{macro-step} computational rule, \emph{$\beta$-reduction}, and does not rest on any notion of machine or automaton. 
Compilers and proof assistants however are concrete tools that have to implement the $\l$-calculus in some way---a problem clearly arises. 
There is a huge gap between the abstract mathematical setting of the calculus and the technical intricacies of an actual implementation. This is why the issue is studied via intermediate \emph{abstract machines}, that are implementation schemes with \emph{micro-step} operations and without too many concrete details.

\paragraph{Closed and Strong $\l$-Calculus.} Functional programming languages are based on a simplified form of $\l$-calculus, that we like to call \emph{closed $\l$-calculus}, with two important restrictions. First, evaluation is \emph{weak}, \ie it does not evaluate function bodies. Second, terms are \emph{closed}, that is, they have no free variables. The theory of the closed $\l$-calculus is much simpler than the general one. 

Proof assistants based on the $\l$-calculus usually require the power of the full theory. 
Evaluation is then \emph{strong}, \ie unrestricted, and the distinction between open and closed terms no longer makes sense, because evaluation has to deal with the issues of open terms even if terms are closed, when it enters function bodies.
We refer to this setting as the \emph{strong $\l$-calculus}. 

Historically, the study of strong and closed $\l$-calculi have followed orthogonal approaches. Theoretical studies rather dealt with the strong $\l$-calculus, and it is only since the seminal work of Abramsky and Ong \cite{DBLP:journals/iandc/AbramskyO93} that theoreticians started to take the closed case seriously. Dually, practical studies mostly ignored  strong evaluation, with the notable exception of Cr\'{e}gut \cite{DBLP:conf/lfp/Cregut90} (1990) and some very recent works \cite{DBLP:conf/ppdp/Garcia-PerezNM13,DBLP:conf/aplas/AccattoliBM15,DBLP:conf/wollic/Accattoli16}. Strong evaluation is nonetheless essential in the implementation of proof assistants or higher-order logic programming, typically for type-checking with dependent types as in the Edinburgh Logical Framework or the Calculus of Constructions, as well as for unification in simply typed frameworks like $\l$-prolog. 

\paragraph{Open Call-by-Value.} In a very recent work \cite{DBLP:conf/aplas/AccattoliG16}, we advocated the relevance of the \emph{open $\l$-calculus}, a framework in between the closed and the strong ones, where evaluation is \emph{weak} but terms may be \emph{open}. Its key property is that the strong case can be described as the iteration of the open one into function bodies. The same cannot be done with the closed $\l$-calculus because---as already pointed out---entering into function bodies requires to deal with (locally) open terms.

The open $\l$-calculus did not emerge before because most theoretical studies focus on the \emph{call-by-name} strong $\l$-calculus, and in call-by-name the distinction open/closed does not play an important role. Such a distinction, instead, is delicate for call-by-value evaluation, where Plotkin's original operational semantics \cite{DBLP:journals/tcs/Plotkin75} is not adequate for open terms. This issue is discussed at length in \cite{DBLP:conf/aplas/AccattoliG16}, where four extensions of Plotkin's semantics to open terms are compared and shown to be equivalent. That paper then introduces the expression \emph{Open Call-by-Value} (shortened \emph{\ocbv}) to refer to them as a whole, as well as \emph{\ccbv} and \emph{\scbv} to concisely refer to the closed and strong call-by-value $\l$-calculus.

\paragraph{The Fireball Calculus.} The simplest presentation of \ocbv is the \emph{fireball calculus} $\firecalc$, obtained from the \cbv $\l$-calculus by generalizing values into \emph{fireballs}. Dynamically, $\beta$-redexes are allowed to fire only when the argument is a fireball (\emph{fireball} is a pun on \emph{fire-able}). The fireball calculus was introduced without a name by Paolini and Ronchi Della Rocca \cite{DBLP:journals/ita/PaoliniR99,parametricBook}, then rediscovered independently first by Leroy and \gregoire \cite{DBLP:conf/icfp/GregoireL02}, and then by Accattoli and Sacerdoti Coen \cite{fireballs}. 
% A nice property of $\firecalc$ is that it is a conservative extension of \ccbv, in the sense that 
Notably, on closed terms, $\firecalc$ %behaves \emph{exactly}
\emph{coincides} with Plotkin's (Closed) \cbv $\l$-calculus.

\paragraph{Coq by Levels.} In \cite{DBLP:conf/icfp/GregoireL02} (2002) Leroy and \gregoire used the fireball calculus to improve the implementation of the Coq proof assistant. In fact, Coq rests on \scbv, but Leroy and \gregoire design an abstract machine for the fireball calculus (\ie \ocbv) and then use it to evaluate \scbv \emph{by levels}: the machine is first executed at top level (that is, out of all abstractions), and then re-launched recursively under abstractions. Their study is itself formalized in Coq, but it lacks an estimation of the efficiency of the machine.\bigskip

In order to continue our story some basic facts about cost models and abstract machines have to be recalled (see \cite{accattoliWPTE} for a gentle tutorial).

\paragraph{Interlude 1: Size Explosion.} It is well-known that $\l$-calculi suffer from a degeneracy called \emph{size explosion}: there are families of terms whose size is linear in $n$, that evaluate in $n$ $\beta$-steps, and whose result has size exponential in $n$. The problem is that the number of $\beta$-steps, the natural candidate as a time cost model, then seems not to be a reasonable cost model, because it does not even account for the time to write down the result of a computation---the \emph{macro-step} character of $\beta$-reduction seems to forbid to count 1 for each step. This is a problem that affects all $\l$-calculi and all evaluation strategies. 

\paragraph{Interlude 2: Reasonable Cost Models and Abstract Machines.} Despite size explosion, surprisingly, the number of $\beta$-steps often \emph{is} a reasonable cost model---so one can indeed count 1 for each $\beta$-step. There are no paradoxes: $\l$-calculi can be simulated in alternative formalisms employing some form of sharing, such as abstract machines. These settings manage compact representations of terms via \emph{micro-step} operations and produce compact representations of the result, avoiding size explosion. 
Showing that a certain $\l$-calculus is reasonable usually is done by %exhibiting 
simulating it with a \emph{reasonable} abstract machine, 
% \ie a machine whose overhead for implementing it on a Random Access Machine (RAM) is polynomial in the number of $\beta$-steps in the calculus. 
\ie a machine implementable with overhead polynomial in the number of $\beta$-steps in the calculus. 
The design of a reasonable abstract machine depends very much on the kind of $\l$-calculus to be implemented, as different calculi admit different forms of size explosion and/or require more sophisticated forms of sharing. 
For strategies in the closed $\l$-calculus it is enough to use the ordinary technology for abstract machines, as first shown by Blelloch and Greiner \cite{DBLP:conf/icfp/BlellochG96}% (1995)
, and then by Sands, Gustavsson, and Moran \cite{DBLP:conf/birthday/SandsGM02}% (2002)
, and, with other techniques, by combining the results in Dal Lago and Martini's \cite{DBLP:conf/icalp/LagoM09} and \cite{DBLP:conf/fopara/LagoM09}% (2006)
. The case of the strong $\l$-calculus is  subtler, and a more sophisticated form of sharing is necessary, as first shown by Accattoli and Dal Lago \cite{DBLP:conf/csl/AccattoliL14}% (2014)
. The topic of this paper is the study of reasonable machines for the intermediate case of \ocbv.

\paragraph{Fireballs are Reasonable.} In \cite{fireballs} Accattoli and Sacerdoti Coen study \ocbv from the point of view of cost models. Their work provides 3 contributions:
\begin{varenumerate}
  \item \emph{Open Size Explosion}: they show that \ocbv is subtler than \ccbv by exhibiting a form of size explosions that is  not possible in \ccbv, making \ocbv closer to \scbv rather than to \ccbv;
  \item \emph{Fireballs are Reasonable}: they show that the number of $\beta$-steps in the fireball calculus is nonetheless a reasonable cost model by exhibiting a reasonable abstract machine, called \glamour, improving over Leroy and \gregoire's machine in \cite{DBLP:conf/icfp/GregoireL02} (see the conclusions for more on their machine);
  \item \emph{And Even Efficient}: they optimize the \glamour into the \uglamour, whose overhead is bilinear (\ie linear in the number of $\beta$-steps \emph{and} the size of the initial term), that is the best possible overhead.
\end{varenumerate}

\paragraph{This Paper.} Here we present two machines, the \eglamour and the \fglamour, that are proved to be correct implementations of \ocbv and to have a polynomial and bilinear overhead, respectively. Their study refines the results of \cite{fireballs} along three axes:
\begin{varenumerate}
  \item \emph{Simpler Machines}: both the \glamour and the \uglamour of \cite{fireballs} are sophisticated machines resting on a labeling of terms. The unchaining optimizations of the second machine is also quite heavy. Both the \eglamour and the \fglamour, instead, do not need labels and the \fglamour is bilinear with no need of the unchaining optimization.

  \item \emph{Simpler Analyses}: the correctness and complexity analyses of the (Unchaining) \glamour are developed in \cite{fireballs} via an informative but complex decomposition via explicit substitutions, by means of the distillation methodology \cite{DBLP:conf/icfp/AccattoliBM14}. Here, instead, we decode the Easy and Fast \glamour  directly to the fireball calculus, that turns out to be much simpler. Moreover, the complexity analysis of the \fglamour, surprisingly, turns out to be straightforward.
  
  \item \emph{Modular Decomposition of the Overhead}: we provide a fine analysis of how different optimizations impact on the complexity of the overhead of abstract machines for \ocbv. In particular, it turns out that one of the optimizations considered essential in \cite{fireballs}, namely \emph{substituting abstractions on-demand}, is not mandatory for reasonable machines---the \eglamour does not implement it and yet it is reasonable. We show, however, that this is true only as long as one stays \emph{inside} \ocbv because the optimization is instead mandatory for \scbv (seen by \gregoire and Leroy as \ocbv \emph{by levels}). To our knowledge substituting abstractions on-demand is an optimization introduced in \cite{DBLP:conf/csl/AccattoliL14} and currently no proof assistant implements it. Said differently, our work shows that the technology currently in use in proof assistants is, at least theoretically, unreasonable.
  
%   both the \uglamour and the \fglamour have overhead $O((\size{\tm_0}+1)\cdot n)$, \ie linear in the number $n$ of steps and in the size $\size{\tm_0}$ of the initial term---these are the two fundamental parameters for complexity analyses. Such a bilinear bound is induced in cases by three restrictions on the substitution process (details are given in the paper), implemented by both machines, altough differently:
%   \begin{varenumerate}
%   \item \emph{No Substitution of Compound Inert Terms}: this restriction lowers the overhead on the size $\size{\tm_0}$  of the initial term from exponential to polynomial, and so it is mandatory for any reasonable machine;
%   \item \emph{No Substitution of Variables}: without, the dependency on the size of the number $n$ of steps is quadratic, with it, it is linear;
%   \item \emph{Substitution of Abstractions On-Demand}: without, the dependency on $\size{\tm_0}$ is quadratic, with it, it is linear.
%   \end{varenumerate}
%   The \glamour implements 1 and 3 and its overhead is $O((\size{\tm_0}+1)\cdot n^2)$, while the \eglamour implements 1 and 2 and its overhead $O((\size{\tm_0}+1)^2\cdot n)$. Such a fine view of how different optimizations impact on the overhead of the machine is a contribution of this paper.
  \end{varenumerate}

  \smallskip
  Summing up, this paper does not improve the known bound on the overhead of abstract machines for \ocbv, as the one obtained in \cite{fireballs} is already optimal. Its contributions instead are a simplification and a finer understanding of the subtleties of implementing \ocbv: we introduce simpler abstract machines whose complexity analyses are elementary and carry a new modular view of how different optimizations impact on the complexity of the overhead. 
  
  In particular, while \cite{fireballs} shows that \ocbv is subtler than \ccbv, here we show that \ocbv is simpler than \scbv, and that defining \scbv as iterated \ocbv, as done by \gregoire and Leroy in \cite{DBLP:conf/icfp/GregoireL02}, may introduce an explosion of the overhead, if done naively.\medskip
  
  \withoutproofs{A longer version of this paper is available on Arxiv  \cite{AccattoliGuerrieri17}.}
  \withproofs{This is a longer version of a paper accepted to FSEN 2017.}
  It contains two Appendices, one with a glossary of rewriting theory and one with omitted proofs.

 % !TEX root = main.tex
\section{\texorpdfstring{The Fireball Calculus $\firecalc$}{The Fireball Calculus} \& Open Size Explosion}
\label{sect:fireball}
% !TEX root = main.tex
\begin{figure}[t]
  \centering
  \scalebox{0.9}{
%     \ovalbox{
    $
    \begin{array}{c@{\hspace{.5cm}}rll}
	    	    \mbox{Terms} & \tm,\tmtwo,\tmthree,\tmfour &\grameq& \var \mid  \la\var\tm \mid \tm\tmtwo\\
    \mbox{Fireballs} & \!\!\!\!\!\!\!\!\!\!\!\!\!\!\!\!\!\!\!\!\!\!\fire, \firetwo, \firethree & \grameq & \la\var\tm \mid \gconst\\
	    \mbox{Inert Terms} & \,\,\,\gconst,\gconsttwo, \gconstthree & \grameq &  \var \fire_1\ldots \fire_n\ \ \ \ n\geq 0\\
    %	\mbox{Inert Terms} & \gconst,\gconsttwo, \gconstthree & \grameq &  \var \fire \mid \gconst\fire
	    \mbox{Evaluation Contexts} & \evctx  & \grameq & \ctxhole\mid \tm\evctx \mid \evctx\tm \\\\	
	    
      \textsc{Rule at Top Level} & \multicolumn{3}{c}{\textsc{Contextual closure}} \\
	    \!(\la\var\tm)(\la\vartwo\tmtwo) \rtobabs \tm\isub\var{\la\vartwo\tmtwo} &
	    \multicolumn{3}{c}{\evctxp \tm \tobabs \evctxp \tmtwo \textrm{~~~if } \tm \rtobabs \tmtwo} \\
	    \,\,(\l\var.\tm)\gconst\rtoin \tm\isub\var\gconst &
	    \multicolumn{3}{c}{\evctxp \tm \toin \evctxp \tmtwo \quad\textup{if } \tm \rtoin \tmtwo} \\\\ 

	    \mbox{Reduction} & \multicolumn{3}{c}{\tof \, \defeq \, \tobabs \!\cup \toin}
    \end{array}
    $}
 % }
  \caption{\label{fig:fireball-calculus} The Fireball Calculus $\firecalc$}
\end{figure}
In this section we introduce the fireball calculus, the presentation of \ocbv we work with in this paper, and show the example of size explosion peculiar to the open setting. Alternative presentations of \ocbv can be found in \cite{DBLP:conf/aplas/AccattoliG16}.

\paragraph{The Fireball Calculus.} The fireball calculus $\firecalc$ is defined in \reffig{fireball-calculus}.
The idea is that the values of the call-by-value $\l$-calculus, given by abstractions and variables, are generalized to fireballs, by extending variables to more general \emph{inert terms}. Actually fireballs and inert terms are defined by mutual induction (in \reffig{fireball-calculus}). 
For instance, $\la\var\vartwo$ is a fireball 
as an abstraction, while $\var$, $\vartwo(\la\var\var)$, $\var\vartwo$, and $(\varthree(\la\var\var))(\varthree\varthree) (\la\vartwo(\varthree\vartwo))$ are fireballs as inert terms. 

The main feature of inert terms is that they are open, normal, and that when plugged in a context they cannot create a redex, hence the name (they are not so-called \emph{neutral terms} because they might have $\beta$-redexes under abstractions). In \gregoire and Leroy's presentation 
\cite{DBLP:conf/icfp/GregoireL02}, inert terms are called \emph{accumulators} and fireballs are simply called values.

Terms are always identified up to $\alpha$-equivalence and the set of %the 
free variables of a term $\tm$ is denoted by $\fv{\tm}$. We use $\tm\isub\var\tmtwo$ for the term obtained by the capture-avoiding substitution of $\tmtwo$ for each free occurrence of $\var$ in $\tm$.

% Evaluation is given by \emph{call-by-fireball} $\beta$-reduction $\tof$: the $\beta$-rule can fire only when the argument is a fireball (\emph{fireball} is a catchier version of \emph{fireable term}). We actually distinguish two sub-rules: one that fires abstractions, noted $\tobabs$, and one that fires inert terms, noted $\toin$. Note that evaluation is weak. 
 Evaluation is given by \emph{call-by-fireball} $\beta$-reduction $\tof$: the $\beta$-rule can fire, \emph{lighting up} the argument, only when it is a fireball (\emph{fireball} is a catchier version of \emph{fire-able term}).
 We actually distinguish two sub-rules: one that \emph{lights up} abstractions, noted $\tobabs$,
 and one that \emph{lights up} inert terms, noted $\toin$ (see \reffig{fireball-calculus}). 
 Note that evaluation is weak (\ie it does not reduce under abstractions). 

\paragraph{Properties of the Calculus.} 
% A famous key property of \ccbv is \emph{harmony}, that is the fact that a closed term $\tm$ either diverges or it evaluates to a value, that can be rephrased as a statement on normal forms by saying that $\tm$ is $\betav$-normal iff $\tm$ is a value. The fireball calculus satisfies the same property with respect to fireballs. Moreover, on closed terms it collapses on \ccbv. It can be said, then, that the fireball calculus is a conservative extension of \ccbv (whose (weak) evaluation is noted $\tobv$). No other presentation of \ocbv has these properties.
A famous key property of \ccbv (whose evaluation is exactly $\tobabs$) is \emph{harmony}: given a closed term $\tm$, either it diverges or it evaluates to an abstraction, \ie $\tm$ is $\betaabs$-normal iff $\tm$ is an abstraction. 
The fireball calculus satisfies an analogous property in the \emph{open} setting by replacing abstractions with fireballs (\refpropp{distinctive-fireball}{open-harmony}). 
Moreover, the fireball calculus is a conservative extension of \ccbv: on closed terms it collapses on \ccbv (\refpropp{distinctive-fireball}{conservative}). 
No other presentation of \ocbv has these properties.

 \newcounter{prop:distinctive-fireball} %new counter in order to use it in appendix
 \addtocounter{prop:distinctive-fireball}{\value{proposition}}
\begin{proposition}[Distinctive Properties of $\firecalc$]%\label{prop:open-harmony}\label{prop:characterize-fnormal}
  \label{prop:distinctive-fireball}
  Let $\tm$ be a term.
  \NoteProof{propappendix:distinctive-fireball}
  \begin{varenumerate}
    \item\label{p:distinctive-fireball-open-harmony} \emph{Open Harmony}:  $\tm$ is $\betaf$-normal iff $\tm$ is a fireball.
    \item\label{p:distinctive-fireball-conservative} \emph{Conservative Open Extension}: $\tm \tof \tmtwo$ iff $\tm \tobabs \tmtwo$ whenever $\tm$ is closed.
  \end{varenumerate}
\end{proposition}

The rewriting rules of $\firecalc$ have also many good operational properties, studied in \cite{DBLP:conf/aplas/AccattoliG16} and summarized in the following proposition.

\newcounter{prop:basic-fireball} %new counter in order to use it in appendix
\addtocounter{prop:basic-fireball}{\value{proposition}}
% \begin{proposition}[Operational Properties of $\firecalc$, \cite{DBLP:conf/aplas/AccattoliG16}]
% \label{prop:basic-fireball}\hfill % \refpropp{basic-fireball}{strong-confluence}
% \NoteProof{propappendix:basic-fireball}
% \begin{varenumerate}
% 	  \item\label{p:basic-fireball-toin-strong-normalization}\label{p:basic-fireball-toin-strong-confluence} $\toin$ is strongly normalizing and strongly confluent.
% 	  \item\label{p:basic-fireball-tobv-toin-strong-commutation} $\tobabs$ and $\toin$ strongly commute.	  
% 	  \item\label{p:basic-fireball-strong-confluence} \label{p:basic-fireball-number-steps} $\tof$ is strongly confluent, and all $\betaf$-normalizing derivations $\deriv$ (if any) from a term $\tm$ have the same length $\sizef{\deriv}$, the same number $\sizebabs{\deriv}$ of $\betaabs$-steps, and the same number $\sizein{\deriv}$ of $\betain$-steps.	  
%   \end{varenumerate}
%   \end{proposition}
\begin{proposition}[Operational Properties of $\firecalc$, \cite{DBLP:conf/aplas/AccattoliG16}]
\label{prop:basic-fireball}
\NoteProof{propappendix:basic-fireball}
\label{p:basic-fireball-number-steps} 
  The reduction $\tof$ is strongly confluent, and all $\betaf$-normalizing derivations $\deriv$ (if any) from a term $\tm$ have the same length $\sizef{\deriv}$, the same number $\sizebabs{\deriv}$ of $\betaabs$-steps, and the same number $\sizein{\deriv}$ of $\betain$-steps.	  
\end{proposition}
  
  \paragraph{Right-to-Left Evaluation.} As expected from a \emph{calculus}, the evaluation rule $\tof$ of $\firecalc$ is \emph{non-deterministic}, because in the case of an application there is no fixed order in the evaluation of the left and right subterms. Abstract machines however implement \emph{deterministic} strategies. 
%   We then fix a deterministic strategy, that is the one implemented by the machines of the next sections% (which fires $\betaf$-redexes from right to left). 
  We then fix a deterministic strategy (which fires $\betaf$-redexes from right to left and is the one implemented by the machines of the next sections). 
%   By strong confluence (\refpropp{basic-fireball}{strong-confluence} in Appendix \ref{app:omitted}), 
  By \refprop{basic-fireball}, the choice of the strategy does not impact on existence of a result, nor on the result itself or on the number of steps to reach it. 
  It does impact however 
  on the design of the machine, which selects $\betaf$-redexes from right to left.

The  \emph{right-to-left evaluation strategy} $\torf$ is defined by closing the root rules $\rtobabs$ and $\rtoin$ in \reffig{fireball-calculus} by \emph{right contexts}, %given
a special kind of evaluation contexts defined by $\revctx \grameq \ctxhole\mid \tm\revctx\mid \revctx\fire$. 
The next lemma %guarantees
ensures our definition is correct.

\newcounter{l:prop-of-torf}
\addtocounter{l:prop-of-torf}{\value{lemma}}
\begin{lemma}[Properties of $\torf$]
\label{l:prop-of-torf}
  Let 
\NoteProof{lappendix:prop-of-torf}
  $\tm$ be a term.
  \begin{varenumerate}
	  \item \label{p:prop-of-torf-compl} \emph{Completeness}: $\tm$ has $\tof$-redex iff $\tm$  has a $\torf$-redex.
	  \item \label{p:prop-of-torf-determ} \emph{Determinism}: $\tm$ has at most one $\torf$-redex.
  \end{varenumerate}
\end{lemma}

\begin{example}
\label{ex:torf}  
  Let $\tm \defeq (\la{\varthree}{\varthree(\vartwo\varthree)})\la{\var}{\var}$.
  Then, $\tm \torf (\la{\var}{\var})(\vartwo \, \la{\var}{\var}) \torf \vartwo \, \la{\var}{\var}$, where the final term $\vartwo \, \la{\var}{\var}$ is a fireball (and $\betaf$-normal).
\end{example}

\paragraph{Open Size Explosion.} Fireballs are delicate, they easily \emph{explode}. The simplest instance of \emph{open size explosion} (not existing in \ccbv) is a variation over the famous looping term $\Omega \defeq (\la\var \var\var) (\la\var \var\var) \tobabs \Omega \tobabs \ldots$. In $\Omega$ there is an infinite sequence of duplications. 
In the size exploding family there is a sequence of $n$ nested duplications. 
We define two families, the family $\{\tm_n\}_{n\in\nat}$ of size exploding terms and the family $\{\gconst_n\}_{n\in\nat}$ of results of  evaluating $\{\tm_n\}_{n\in\nat}$:

\begin{center}
$\begin{array}{ccc@{\hspace{1cm}}ccc@{\hspace{2cm}}ccc@{\hspace{1cm}}cccccccccccccccccc}
  \tm_0 & \defeq & \vartwo &\tm_{n+1} & \defeq & (\la\var \var\var) \tm_n&\gconst_0 & \defeq & \vartwo
   &\gconst_{n+1} & \defeq & \gconst_n \gconst_n
\end{array}$  
\end{center}

We use $\size\tm$ for the size of a term, \ie the number of symbols to write it. 

\newcounter{prop:open-size-explosion}
\addtocounter{prop:open-size-explosion}{\value{proposition}}
\begin{proposition}[Open Size Explosion, \cite{fireballs}]
\label{prop:open-size-explosion}
Let $n\in \nat$.
\NoteProof{propappendix:open-size-explosion}
Then $\tm_n \toin^n \gconst_n$, moreover $\size{\tm_n} = O(n)$, $\size{\gconst_n} = \Omega(2^n)$, and $\gconst_n$ is an inert term.
\end{proposition}

\paragraph{Circumventing Open Size Explosion.} Abstract machines implementing the substitution of inert terms, such as the one described by \gregoire and Leroy in \cite{DBLP:conf/icfp/GregoireL02} are unreasonable because for the term $\tm_n$ of the size exploding family they compute the full result $\gconst_n$. The machines of the next sections are reasonable because they avoid the substitution of inert terms, that is justified by the following lemma.

\newcounter{l:inerts-and-creations}
\addtocounter{l:inerts-and-creations}{\value{lemma}}
\begin{lemma}[Inert Substitutions Can Be Avoided]
\label{l:inerts-and-creations}
  Let 
  \NoteProof{lappendix:inerts-and-creations}
  $\tm, \tmtwo$ be terms and $\gconst$ be an inert term. Then, $\tm \tof \tmtwo$ iff $\tm \isub\var\gconst \tof \tmtwo \isub\var\gconst$.
\end{lemma}

\reflemma{inerts-and-creations} states that the substitution of an inert term cannot create redexes, which is why it can be avoided. 
For general terms, 
only direction $\Rightarrow$ holds,
because substitution can create redexes, as in $(\var \vartwo) \isub \var { \la\varthree \varthree } = (\la\varthree \varthree) \vartwo$. Direction $\Leftarrow$, instead, is distinctive of inert terms, of which it justifies the name.

% In $\firecalc$, \emph{fireballs coincide with normal forms} (see \refprop{distinctive-fireball} in Appendix~\ref{app:omitted}).
 % !TEX root = main.tex
\section{Preliminaries on Abstract Machines, Implementations, and Complexity Analyses}
\label{sect:machines-intro}
% In this section we introduce general notions about abstract machines, given with respect to a generic machine $\mach$ and a generic strategy $\tostrat$ on $\l$-terms. Then we give an abstract notion of implementation and sufficient conditions for it. Finally, we provide a recipe for complexity analyses.

% \paragraph{Abstract Machines Glossary.} 
\begin{varitemize}
\item An abstract machine $\mach$ is given by \emph{states}, noted $\state$, and %deterministic 
  \emph{transitions} between them, noted $\tomachine$;%,   which are deterministic: if $\state \tomachine \statetwo$ and $\state \tomachine \statethree$  then $\statetwo = \statethree$;
\item A state is given by the \emph{code under evaluation} plus some \emph{data-structures};
\item The code under evaluation, as well as the other pieces of code scattered in the data-structures, are $\l$-terms \emph{not considered modulo $\alpha$-equivalence};
\item Codes are over-lined, to stress the different treatment of $\alpha$-equivalence;
\item A code $\code$ is \emph{well-named} if $\var$ may occur only in $\codetwo$ (if at all) for every sub-code $\la\var\codetwo$ of $\code$;
\item A state $\state$  is \emph{initial} if its code is well-named and its data-structures are empty; 
\item Therefore, there is a bijection $\compil\cdot$ (up to $\alpha$) between terms and initial states, called \emph{compilation}, sending a term $\tm$ to the initial state $\compil\tm$ on a well-named code $\alpha$-equivalent to $\tm$;
\item An \emph{execution} is a (potentially empty) sequence of transitions $\compil{\tm_0} \tomachine^* \state$ from an initial state obtained by compiling an (initial) term $\tm_0$;
\item A state $\state$  is \emph{reachable} if it can be obtained as the end state of an execution;
\item A state $\state$ is \emph{final} if it is reachable and no transitions apply to $\state$;
\item A machine comes with a map $\decode\cdot$ from states to terms, called \emph{decoding}, that on initial states is the inverse (up to $\alpha$) of compilation, \ie $\decode{\compil{\tm}} = \tm$ for any term $\tm$;
\item A machine $\mach$ has a set of \emph{$\beta$-transitions}, whose union is noted $\tomachb$, that are meant to be mapped to $\beta$-redexes by the decoding, while the remaining \emph{overhead transitions}, denoted by $\tomacho$, are mapped to equalities; 
% a generic (overhead or $\beta$) transition is denoted by $\tomachm$;
\item We use $\size\exec$ for the length of an execution $\exec$, and $\sizebeta\exec$ for the number of $\beta$-transitions in $\exec$.
\end{varitemize}

\paragraph{Implementations.} For every machine one has to prove that it correctly implements the strategy in the $\lambda$-calculus it was conceived for. Our notion, tuned towards complexity analyses, requires a perfect match between the number of $\beta$-steps of the strategy and the number of $\beta$-transitions of the machine execution.

\begin{definition}[Machine Implementation]
\label{def:implem}
A machine $\mach$ \emph{implements a strategy} $\tostrat$ on $\lambda$-terms via a decoding $\decode\cdot$ when given a $\l$-term $\tm$ the following holds:

\begin{varenumerate}
\item\label{p:implem-exec-to-deriv} \emph{Executions to Derivations}: for any $\mach$-execution $\exec \colon \compil\tm \tomachine^* \state$ there exists a $\tostrat$-derivation $\deriv \colon \tm \tostrat^* \decode\state$.

\item\label{p:implem-deriv-to-exec} \emph{Derivations to Executions}: for every $\tostrat$-derivation $\deriv \colon \tm \tostrat^* \tmtwo$ there exists a $\mach$-execution $\exec \colon \compil\tm \tomachine^* \state$ such that $\decode\state = \tmtwo$.

\item\label{p:implem-beta-matching} \emph{$\beta$-Matching}: in both previous points the number $\sizebeta\exec$ of $\beta$-transitions in $\exec$ is exactly the length $\size\deriv$ of the derivation $\deriv$, \ie $\size\deriv = \sizebeta\exec$.
\end{varenumerate}
\end{definition}

%Note that if a machine implements a strategy than the two are \emph{weakly bisimilar}, where weakness is given by the fact that overhead transitions do not have an equivalent on the calculus (hence their name).

\paragraph{Sufficient Condition for Implementations.} The proofs of implementation theorems tend to follow always the same structure, based on a few abstract properties collected here into the notion of implementation system% (\refdef{implementation})
. 
% The proof of the next abstract theorem  follows a standard reasoning used also in \cite{DBLP:conf/icfp/AccattoliBM14,fireballs,accattoliWPTE}.

\begin{definition}[Implementation System]
  \label{def:implementation}
  A machine $\mach$, a strategy $\tostrat$, and a decoding $\decode\cdot$ form an \emph{implementation system} if the following conditions hold:

  \begin{varenumerate}
		\item\label{p:def-beta-projection} \emph{$\beta$-Projection}: $\state \tomachhole\beta \statetwo$ implies $\decode\state \tostrat \decode\statetwo$;
		\item\label{p:def-overhead-transparency} \emph{Overhead Transparency}: $\state \tomacho \statetwo$ implies $\decode\state = \decode\statetwo$;
		\item 	\emph{Overhead Transitions Terminate}:  $\tomacho$ terminates;

	\item \emph{Determinism}: both $\mach$ and $\tostrat$ are deterministic;

	\item \emph{Progress}: $\mach$ final states decode to $\tostrat$-normal terms.
  \end{varenumerate}

\end{definition}

\newcounter{thm:abs-impl}
\addtocounter{thm:abs-impl}{\value{theorem}}
\begin{theorem}[Sufficient Condition for Implementations]
\label{thm:abs-impl}
Let 
\NoteProof{thmappendix:abs-impl}
$(\mach, \tostrat, \decode\cdot)$ be an \emph{implementation system}.
Then, $\mach$ implements $\tostrat$ via $\decode\cdot$.
\end{theorem}
The proof of \refthm{abs-impl} is a clean and abstract generalization of the concrete reasoning already at work in \cite{DBLP:conf/icfp/AccattoliBM14,fireballs,DBLP:conf/wollic/Accattoli16,accattoliWPTE}.

% \newcounter{thm:abs-impl}
% \addtocounter{thm:abs-impl}{\value{theorem}}
% \begin{theorem}[Sufficient Conditions for Implementations]
% \label{thm:abs-impl}
% Let 
% \NoteProof{thmappendix:abs-impl}
% $\mach$, $\tostrat$, and $\decode\cdot$ be a machine, a strategy, and a decoding forming an \emph{implementation system}, \ie such that:
% \begin{varenumerate}
% 		\item \emph{$\beta$-Projection}: $\state \tomachhole\beta \statetwo$ implies $\decode\state \tostrat \decode\statetwo$;
% 		\item \emph{Overhead Transparency}: $\state \tomacho \statetwo$ implies $\decode\state = \decode\statetwo$;
% 		\item 	\emph{Overhead Transitions Terminate}:  $\tomacho$ terminates;
% 
% 	\item \emph{Determinism}: both $\mach$ and $\tostrat$ are deterministic;
% 
% 	\item \emph{Progress}: $\mach$ final states decode to $\tostrat$-normal terms.
% \end{varenumerate}
% 
% Then, $\mach$ implements $\tostrat$ via $\decode\cdot$.
% \end{theorem}

\paragraph{Parameters for Complexity Analyses.} 
By the \emph{derivations-to-executions} part of the implementation (\refpoint{implem-deriv-to-exec} in \refdef{implem}), given a derivation $\deriv \colon \tm_0 \tostrat^n \tmtwo$ there is a shortest execution $\exec \colon \compil{\tm_0} \tomachine^* \state$ such that $\decode\state = \tmtwo$. 
Determining \emph{the complexity of a machine $\mach$} amounts to bound the complexity of a concrete implementation of $\exec$ on a RAM model, as a function of two fundamental parameters:
\begin{varenumerate}
  \item \emph{Input}: the size $\size{\tm_0}$ of the initial term $\tm_0$ of the derivation $\deriv$;
  \item \emph{$\beta$-Steps/Transitions}: the length $n = \size\deriv$ of the derivation $\deriv$, that coincides with the number $\sizebeta\exec$ of $\beta$-transitions in $\exec$  by the $\beta$-matching requirement for implementations (\refpoint{implem-beta-matching} in \refdef{implem}).
\end{varenumerate}
A machine is \emph{reasonable} if its complexity is polynomial in $\size{\tm_0}$ and $\sizebeta\exec$, and it is \emph{efficient} if it is linear in both parameters. 
So, a strategy is reasonable (resp.~efficient) if there is a reasonable (resp.~efficient) machine implementing it. 
In \refsects{eglamour}{eg-compl-anal} we study a reasonable machine implementing right-to-left evaluation $\torf$ in $\firecalc$, thus showing that it is a reasonable strategy. In \refsect{fglamour} we optimize the machine to make it efficient.
By \refprop{basic-fireball}, \mbox{this implies that \emph{every} strategy in $\firecalc$ is efficient.}

\paragraph{Recipe for Complexity Analyses.} 
For complexity analyses on a machine $\mach$, overhead transitions $\tomacho$ are further separated into two classes:
\begin{varenumerate}
  \item \emph{Substitution Transitions $\tomachse$}: they are in charge of the substitution process; 
  \item \emph{Commutative Transitions $\tomachc$}: they are in charge of searching for the next $\beta$ or substitution redex to reduce.
\end{varenumerate}

\noindent Then, the estimation of the complexity of a machine is done in three steps:

\begin{varenumerate}
  \item \emph{Number of Transitions}: bounding the length of the execution $\exec$% simulating the derivation $\deriv$
  , by bounding the number of overhead transitions. 
  This part %is organized in
  splits into two subparts:
  \begin{varenumerateroman}
    \item \emph{Substitution $vs$ $\beta$}: bounding the number $\sizee\exec$ of substitution transitions in $\exec$ using the number of $\beta$-transitions;
    \item \emph{Commutative $vs$ Substitution}: bounding the number $\sizecom\exec$ of substitution transitions in $\exec$ using the size of the input and $\sizee\exec$; the latter---by the previous point---induces a bound with respect to $\beta$-transitions.
  \end{varenumerateroman}
  \item \emph{Cost of Single Transitions}: bounding the cost of concretely implementing a single transition of $\mach$. Here it is usually necessary to go beyond the abstract level, making some (high-level) assumption on how codes and data-structure are concretely represented.
  Commutative transitions are designed on purpose to have constant cost. Each substitution transition has a cost linear in the size of the initial term thanks to an invariant (to be proved) ensuring that only subterms of the initial term are duplicated and substituted along an execution. 
  Each $\beta$-transition has a cost either constant or linear in the input.
  \item \emph{Complexity of the Overhead}: obtaining the total bound by composing the first two points, that is, by taking the number of each kind of transition times the cost of implementing it, and summing over all kinds of transitions.
\end{varenumerate}

\paragraph{(Linear) Logical Reading.}  Let us mention that our partitioning of transitions into $\beta$, substitution, and commutative ones admits a proof-theoretical view, as machine transitions can be seen as cut-elimination steps \cite{DBLP:journals/toplas/AriolaBS09,DBLP:conf/icfp/AccattoliBM14}. Commutative transitions correspond to commutative cases, while $\beta$ and substitution are principal cases. Moreover, in linear logic the $\beta$ transition corresponds to the multiplicative case while the substitution transition to the exponential one. See \cite{DBLP:conf/icfp/AccattoliBM14} for more details.
%The following sections will introduce two abstract machines, with the respective decodings, and prove that they form implementations systems with respect to right-to-left evaluation $\torf$ in the fireball calculus. In both cases we will spell out the invariants of the machines needed to prove the conditions of an implementation system, but the proof of the conditions will be omitted. They can be found in the appendix.
 % !TEX root = main.tex
% !TEX root = main.tex
\begin{figure}[!tb]
  \centering
%\ovalbox{
  \scalebox{0.9}{
$
\!\!\!\begin{array}{c}
\begin{array}{c|c}
\begin{array}{rclcrrlllllllll}
        \stackitem & \grameq & \pair{\la\var\codetwo}{\stempty} \mid  \pair{\var}{\stack} &&
	\genv		& \grameq & \stempty \mid \esub\var\stackitem \!\cons\! \genv \\
	\stack 	& \grameq & \stempty \mid \stackitem \cons \stack&&

	\state		& \grameq & \glamst\dump\code\stack\genv\\
	\dump	& \grameq & \stempty \mid \dump \cons \dentry\code\stack 
	\\
\end{array} &
\begin{array}{rcllrrlllllrrll}
	\decode{\stempty}	& \defeq & \ctxhole &&		\relunf\tm\stempty &\defeq &\tm   \quad\ \relunf\tm{ \esub\var\stackitem \cons \genv }	 \defeq  \relunf{ \tm \isub\var{\decode\stackitem} }{  \genv }
        \\
	\decode{ \stackitem \cons \stack} 			& \defeq & \decstackp{\ctxhole\decode{\stackitem}} &&
	\stctx{\state}				& \defeq & \relunf{\decdump\ctxholep{\decstack}}\genv \\
        \decode{\pair{\tm}{\stack}} & \defeq & \decstackp{\tm} &&
	\decode{\state}				& \defeq & \relunf{\decdump\ctxholep{\decstackp\code}}\genv = \stctx{\state}\ctxholep{\relunf\code\genv}  %&=&\decdump\ctxholep{\decstackp\code}\indsub\genv 
	\\
        \decode{\dump\!\cons\!\dentry\code\stack} & \defeq & \decdumpp{\decstackp{\code\ctxhole}} &&\multicolumn{3}{l}{\ \ \mbox{where $\state = \glamst\dump\code\stack\genv$}}
\end{array}
\end{array}
\\\\%\hline\\

%\mbox{Right-to-Left Open Machine Freely Substituting Abstractions, with overhead $O(\size\code^2\cdot \size\deriv)$}\\\\

  	{\setlength{\arraycolsep}{0.35em}
  	\begin{array}{c|c|c|c|c|c|c|c|ccc}
		\mbox{Dump} & \mbox{Code} & \mbox{Stack} & \mbox{Global Env} 
		&&
		\mbox{Dump} & \mbox{Code} & \mbox{Stack} & \mbox{Global Env}\\
\hhline{=|=|=|=|=|=|=|=|=}
		\dump & \code\codetwo & \stack & \genv
	  	&\tomachcone&
	  	\!\dump\!\cons\!\dentry\code\stack\! & \codetwo & \stempty &\genv
	  	\\
%\hhline{-|-|-|---|-|-|-}
		\!\dump\!\cons\!\dentry\code\stack\! & \la\var\codetwo & \stempty & \genv
		& \tomachctwo &
		\dump & \code & \pair{\la\var\codetwo}{\stempty} \cons\stack & \genv
		\\
%\hhline{-|-|-|---|-|-|-}		
%		\dump\cons\dentry\code\stack & \var & \stacktwo & \genv_1\esub\var{\pair{\vartwo}{\stackthree}}\genv_2
%		& \tomachcthree &
%		\dump & \code & \pair{\var}{\stacktwo}\cons\stack & \genv_1\esub\var{\pair{\vartwo}{\stackthree}}\genv_2
%		\\

		\!\dump\!\cons\!\dentry\code\stack\! & \var & \stacktwo & \genv
		& \tomachcthree &
		\dump & \code & \pair{\var}{\stacktwo}\cons\stack & \genv \\
		\multicolumn{9}{r}{\mbox{if $\genv(\var)=\undef$ or $\genv(\var)=\pair{\vartwo}{\stackthree}$}}
		\\
%		\hhline{-|-|-|---|-|-|-}
		\dump & \la\var\code & \stackitem\!\cons\!\stack & \genv
	  	&\tomachsm &
		\dump & \code & \stack & \esub\var\stackitem\genv
	  	\\
%\hhline{-|-|-|---|-|-|-}
% 		\dump & \la\var\code & \valp\cons\stack & \genv
% 	  \hhline{=|=|=|=|=}	&\tomachsmv &
% 		\dump & \code & \stack & \esub\var\valp\genv
% 	  	\\
% 		
% 		\dump & \la\var\code & \gconstp\cons\stack & \genv
% 	  	&\tomachsmi &
% 		\dump & \code & \stack & \esub\var\gconstp\genv
% 	  	\\

	  	\dump & \var & \stack & \genv_1\esub\var{\pair{\la\vartwo\codetwo}\stempty}\genv_2
		& \tomachse &
		\dump & \!\rename{\la\vartwo\codetwo}\! & \stack & \genv_1\esub\var{\pair{\la\vartwo\codetwo}\stempty}\genv_2\\

	\end{array}}
\\\\
\mbox{where $\rename{\la\vartwo\codetwo}$ is any well-named code $\alpha$-equivalent to $\la\vartwo\codetwo$ such that its }
\\ \mbox{bound names  are fresh with respect to those in $\dump$, $\stack$ and $\genv_1\esub\var{\pair{\la\vartwo\codetwo}\stempty}\genv_2$.}
\end{array}
$
%}
}
\caption{\eglamour machine: data-structures (stacks $\stack$, dumps $\dump$, global env. $\genv$, states $\state$), unfolding $\tm\!\downarrow_\genv$, decoding $\decode\cdot$ (stacks are decoded to contexts in postfix notation for plugging, \ie we write $\decstackp\code$ rather than $\decodep \stack \code$), and transitions.}
\label{fig:eglamour}
  \end{figure}
\section{\eglamour}
\label{sect:eglamour}

In this section we introduce the \eglamour, a simplified version of the \glamour machine from \cite{fireballs}%, not needing any labeling of codes. % and yet providing a reasonable implementation. %For a comparison with \gregoire and Leroy \cite{DBLP:conf/icfp/GregoireL02} see Sect. II of \cite{fireballs}.
: unlike the latter, the \eglamour does not need any labeling of codes to provide a reasonable implementation.

With respect to the literature on abstract machines for \cbv, our machines are unusual in two respects. 
First, and more importantly, they use a single \emph{global} environment instead of \emph{closures} and \emph{local environments}. 
Global environments are used in a minority of works \cite{DBLP:journals/entcs/FernandezS09,DBLP:conf/birthday/SandsGM02,DBLP:conf/ppdp/DanvyZ13,DBLP:conf/icfp/AccattoliBM14,fireballs,DBLP:conf/aplas/AccattoliBM15,DBLP:conf/wollic/Accattoli16} and induce simpler, more abstract machines where $\alpha$-equivalence is pushed to the meta-level (in the operation $\renamenop\code$ in $\tomachse$ in \reffigs{eglamour}{fglamour}). 
This %fact 
on-the-fly $\alpha$-renaming is harmless with respect to complexity analyses, see also discussions in \cite{DBLP:conf/icfp/AccattoliBM14,accattoliWPTE}. Second, argument stacks contain pairs of a code and a stack, to implement some of the machine transitions in constant time.

%\subsection{The Machine}
\paragraph{Background.} \glamour stands for \emph{Useful} (\ie optimized to be \emph{reasonable}) \emph{Open} (reducing open terms) \emph{Global} (using a single global environment) LAM, and LAM stands for \emph{Leroy Abstract Machine}, an ordinary machine implementing right-to-left \ccbv, defined in \cite{DBLP:conf/icfp/AccattoliBM14}. 
In \cite{fireballs} the study of the \glamour was done according to the distillation approach of \cite{DBLP:conf/icfp/AccattoliBM14}, \ie by decoding the machine towards a $\l$-calculus with explicit substitutions. Here we do not follow the distillation approach, we decode directly to $\firecalc$, which is simpler.

\paragraph{Machine Components.} The \eglamour is defined in \reffig{eglamour}. A machine state $\state$ is a quadruple $\glamst{\dump}{\code}{\stack}{\genv}$ given by:
\begin{varitemize}
\item \emph{Code $\code$}: a term not considered up to $\alpha$-equivalence, which is why it is over-lined;

\item  \emph{Argument Stack $\stack$}: it contains the arguments of the current code. Note that stacks items $\stackitem$ are pairs $\pair{\var}{\stack}$ and $\pair{\la\var\codetwo}{\stempty}$. These pairs allow to implement some of the transitions in constant time. 
The pair $\pair{\var}{\stack}$ codes the term %$\decstackp\code$ 
$\decstackp\var$ (defined in \reffig{eglamour}---the decoding is explained below) that would be obtained by putting $\var$ in the context obtained by decoding the argument stack $\stack$.  
The pair $\pair{\la\var\codetwo}{\stempty}$ is used to inject abstractions into pairs, so that items $\stackitem$ can be uniformly seen as pairs $\pair\code\stack$ of a code $\code$ and a stack $\stack$.

\item \emph{Dump $\dump$}: a second stack, that together with the argument stack $\stack$ is used to walk through the code and search for the next redex to reduce. 
The dump is extended with an entry $\dentry\code\stack$ every time evaluation enters in the right subterm $\codetwo$ of an application $\code\codetwo$. 
The entry saves the left part $\code$ of the application and the current stack $\stack$, to restore them when the evaluation of the right subterm $\codetwo$ is over. 
The dump $\dump$ and the stack $\pi$ decode to an evaluation context.

\item \emph{Global Environment $\genv$}: 
% it is used to implement micro-step evaluation (\ie the substitution for one variable occurrence at a time), storing the delayed substitutions for the redexes encountered so far. 
a list of explicit (\ie delayed) substitutions storing substitutions generated by the redexes encountered so far. It is used to implement micro-step evaluation (\ie the substitution for one variable occurrence at a time). 
We write $\genv(\var) = \bot$ if in $\genv$ there are no entries of the form $\esub\var\stackitem$. 
% To save space, sometimes
Often $\esub\var{\stackitem}\genv$ stands for $\esub\var{\stackitem}\!\cons\!\genv$.
\end{varitemize}

\paragraph{Transitions.} 
In the \eglamour there is one %kind of 
$\beta$-transition whereas overhead transitions are divided up into substitution and commutative transitions.
\begin{varitemize}
	\item \emph{$\beta$-Transition $\tomachsm$}: it morally fires a $\torf$-redex, the one corresponding to $(\la\var\code)\stackitem$, except that it puts a new delayed substitution $\esub\var\stackitem$ in the environment instead of doing the meta-level substitution $\code\isub\var\stackitem$ of the argument in the body of the abstraction;
	\item \emph{Substitution Transition $\tomachse$}: it substitutes the variable occurrence under evaluation with a (properly $\alpha$-renamed copy of a) code from the environment. It is a micro-step variant of meta-level substitution. 
	It is invisible on $\firecalc$ because the decoding produces the term obtained by meta-level substitution, and so the micro work done by $\tomachse$ cannot be observed at the \mbox{coarser granularity of $\firecalc$.}	
	\item \emph{Commutative Transitions $\tomachc$}: they locate and
expose the next redex according to the right-to-left strategy, by rearranging the data-structures. They are invisible on the calculus.
The commutative rule $\tomachcone$ forces evaluation to be right-to-left on applications: the machine processes first the right subterm $\codetwo$, saving the left sub-term $\code$ on the dump together with its current stack $\stack$.
The role of $\tomachctwo$ and $\tomachcthree$ is to backtrack to the entry on top of the dump. When the right subterm, \ie the pair $\pair\code\stack$ of current code and stack, is finally in normal form, it is pushed on the stack and the machine backtracks.
\end{varitemize}
\emph{O for Open}: note condition $\genv(\var)=\undef$  in $\tomachcthree$---that is how the \eglamour handles open terms. \emph{U for Useful}: note condition $\genv(\var)=\pair{\vartwo}{\stackthree}$ in $\tomachcthree$---inert terms are never substituted, according to \reflemma{inerts-and-creations}. Removing the useful side-condition one recovers \gregoire and Leroy's machine \cite{DBLP:conf/icfp/GregoireL02}. 
% Note something that will play a role when we will discuss optimizations, in \refsect{fglamour}: the terms substituted by $\tomachse$ are always abstractions and never variables.
Note that terms substituted by $\tomachse$ are always abstractions and never variables---this fact will play a role  in \refsect{fglamour}.
% \paragraph{Garbage Collection.} Garbage collection is here simply ignored, or, more precisely, it is encapsulated at the meta-level, in the decoding function. It is well-known that this is harmless for the study of time complexity.
\emph{Garbage Collection}: it is here simply ignored, or, more precisely, it is encapsulated at the meta-level, in the decoding function. It is well-known that this is harmless for the study of time complexity.

\paragraph{Compiling, Decoding and Invariants.}
A term $\tm$ is compiled to the machine \emph{initial state} $\compil{\tm} = \glamst\stempty\code\stempty\stempty$, where $\code$ is a well-named term $\alpha$-equivalent to $\tm$. 
Conversely, every machine state $\state$ decodes to a term $\decode\state$ (see the top right part of \reffig{eglamour}), having the shape $\stctx{\state}\ctxholep{\relunf{\code}\genv}$, where $\relunf{\code}\genv$ is a $\l$-term, obtained by applying to the code the meta-level substitution $\relunf{}\genv$ induced by the global environment $\genv$, and $\stctx{\state}$ is an evaluation context, obtained by decoding the stack $\stack$ and the dump $\dump$ and then applying $\relunf{}\genv$. Note that, to improve readability, stacks are decoded to contexts in postfix notation for plugging, \ie we write $\decstackp\code$ rather than $\decodep \stack \code$ because $\stack$ is a context that puts arguments in front of $\code$.

\begin{example}
\label{ex:easy}
  To have a glimpse of how the \eglamour works, let us show how it implements the derivation  $\tm \defeq (\la{\varthree}{\varthree(\vartwo\varthree)})\la{\var}{\var} \torf^2 \vartwo \, \la{\var}{\var}$ of \refex{torf}:
  \begin{equation*}
  \scalebox{0.87}{
  $\begin{array}{l|r|c|rc}%Easy
    \text{Dump} & \text{Code} & \text{Stack} & \text{Global Environment} \\
    \cline{1-4}
    \glamsttab {\stempty} {(\la{\varthree}{\varthree(\vartwo\varthree)})\la{\var}{\var}} {\stempty} {\stempty}  &\ \tomachcone \\
    \glamsttab{\dentry{\la{\varthree}{\varthree(\vartwo\varthree)}}{\stempty}} {\la{\var}{\var}} {\stempty} {\stempty} &\ \tomachctwo \\
    \glamsttab{\stempty}{\la{\varthree}{\varthree(\vartwo\varthree)}}{\pair{\la{\var}{\var}}{\stempty}}{\stempty} &\ \tomachsm \\
    \glamsttab{\stempty} {\varthree(\vartwo\varthree)} {\stempty} {\esub{\varthree}{\pair{\la{\var}{\var}}{\stempty}}} &\ \tomachcone \\ 
    \glamsttab{\dentry{\varthree}{\stempty}} {\vartwo\varthree} {\stempty} {\esub{\varthree}{\pair{\la{\var}{\var}}{\stempty}}} &\ \tomachcone \\
    \glamsttab {\dentry{\varthree}{\stempty} \cons \dentry{\vartwo}{\stempty}} {\varthree} {\stempty} {\esub{\varthree}{\pair{\la{\var}{\var}}{\stempty}}} &\ \tomachse \\
    \glamsttab {\dentry{\varthree}{\stempty} \cons \dentry{\vartwo}{\stempty}} {\la{\var'\!}{\var'}} {\stempty} {\esub{\varthree}{\pair{\la{\var}{\var}}{\stempty}}} &\ \tomachctwo \\
    \glamsttab {\dentry{\varthree}{\stempty}} {\vartwo} {\pair{\la{\var'\!}{\var'}}{\stempty}} {\esub{\varthree}{\pair{\la{\var}{\var}}{\stempty}}} &\ \tomachcthree \\
    \glamsttab {\stempty} {\varthree} {\pair{\vartwo}{(\pair{\la{\var'\!}{\var'}}{\stempty})}} {\esub{\varthree}{\pair{\la{\var}{\var}}{\stempty}}} &\ \tomachse \\
    \glamsttab {\stempty} {\la{\var''\!}{\var''}} {\pair{\vartwo}{(\pair{\la{\var'\!}{\var'}}{\stempty})}} {\esub{\varthree}{\pair{\la{\var}{\var}}{\stempty}}} &\ \tomachsm \\
    \glamsttab {\stempty} {\var''} {\stempty} {\esub{\var''}{\pair{\vartwo}{(\pair{\la{\var'\!}{\var'}}{\stempty})}} \!\cons\! \esub{\varthree}{\pair{\la{\var}{\var}}{\stempty}}} 
  \end{array}$
  }
  \end{equation*}    
  Note that  the initial state is the compilation of the term $\tm$, the final state decodes to the term $\vartwo \, \la{\var}{\var}$, and the two $\beta$-transitions in the execution correspond to the two $\torf$-steps in the derivation considered in \refex{torf}.
\end{example}

The study of the \eglamour machine relies on the following invariants. 

\newcounter{l:eglamour-invariants} %new counter in order to use it in appendix
\addtocounter{l:eglamour-invariants}{\value{lemma}}
\begin{lemma}[\eglamour Qualitative Invariants]
	\label{l:eglamour-invariants} % \reflemmap{eglamour-invariants}{subterm}
	Let%
\NoteProof{lappendix:eglamour-invariants}
	$\state = (\dump,\code,\stack,\genv)$ be a reachable state. Then:
	\begin{varenumerate}
		\item \emph{Name:} \label{p:eglamour-invariants-name}
		\begin{varenumerate}
			\item \emph{Explicit Substitution}: \label{p:eglamour-invariants-name-es}
			if $\genv = \genvtwo \esub\var\codetwo \genvthree$  then $\var$ is fresh wrt $\codetwo$ and $\genvthree$;

			\item \emph{Abstraction}: \label{p:eglamour-invariants-name-abs}
			if $\la\var\codetwo$ is a subterm of $\dump$, $\code$, $\stack$, or $\genv$, %then 
			$\var$ may occur only in $\codetwo$;
		\end{varenumerate}
		
		\item \label{p:eglamour-invariants-fireball-stack}\emph{Fireball Item}: $\decode\stackitem$ and $\relunf{\decode\stackitem}\genv$ are inert terms if $\stackitem = \pair\var\stacktwo$, and abstractions otherwise, for every item $\stackitem$ in $\stack$, in $\genv$, and in every stack in $\dump$;

		\item \label{p:eglamour-invariants-ev-ctx}\emph{Contextual Decoding}: $\stctx{\state} = \relunf{\decdumpp\decstack}\genv$ is a right context.

		%\item \label{p:eglamour-invariants-env-size}\emph{Environment Size}: the length of the global environment $\genv$ is bound by $\sizem\exec$. Moreover, the number of entries in $\genv$ that contain an abstraction is exactlty $\sizemv\exec$.
		
	\end{varenumerate}
\end{lemma}

\paragraph{Implementation Theorem.} The invariants are used to prove the implementation theorem by proving that the hypotheses of \refthm{abs-impl} hold, that is, that the \eglamour, $\torf$ and $\decode\cdot$ form an implementation system.

\newcounter{thm:eglamour-implementation} %new counter in order to use it in appendix
\addtocounter{thm:eglamour-implementation}{\value{theorem}}
\begin{theorem}[\eglamour Implementation]
  \label{thm:eglamour-implementation}
  The 
  \NoteProof{thmappendix:eglamour-implementation}
  \eglamour implements right-to-left evaluation $\torf$ in $\firecalc$ (via the decoding $\decode\cdot$).
\end{theorem}

%The invariants are used to prove the following lemmas.
%
\newcounter{l:eglamour-trans-projection} %new counter in order to use it in appendix
\addtocounter{l:eglamour-trans-projection}{\value{lemma}}
%\begin{lemma}[\eglamour One-Step Weak Simulation]
%  \label{l:eglamour-trans-projection} % \reflemmap{eglamour-trans-projection}{exp-com}
%	Let% 
%\NoteProof{lappendix:eglamour-trans-projection}
%	$\state$ be a reachable state.
%	\begin{varenumerate}
%		\item \label{p:eglamour-trans-projection-exp-com} \emph{Commutative \& Exponential}: 
%		if $\state\tomachhole{\esym,\csym_{1,2,3}}\statetwo$ then $\decode\state=\decode\statetwo$;
%		\item \label{p:eglamour-trans-projection-mult}\emph{Multiplicative}: if $\state\tomachsm\statetwo$ then $\decode\state\torf\decode\statetwo$.	  
%		
%	\end{varenumerate}
%\end{lemma}
%
\newcounter{l:eglamour-progress} %new counter in order to use it in appendix
\addtocounter{l:eglamour-progress}{\value{lemma}}
%\begin{lemma}[\eglamour Progress]
%\label{l:eglamour-progress}
%  Let%
%  \NoteProof{lappendix:eglamour-progress}
%  $\state$ be a reachable final state. Then $\decode\state$ is a fireball, \ie it is $\betaf$-normal. 
%\end{lemma}
%
%From these properties it follows the theorem that the \eglamour correctly implements right-to-left evaluation $\torf$ 
%
\newcounter{thm:weak-bis} %new counter in order to use it in appendix
\addtocounter{thm:weak-bis}{\value{theorem}}
 \section{Complexity Analysis of the \eglamour}
\label{sect:eg-compl-anal}
The analysis of the \eglamour is done according to the recipe given at the end of \refsect{machines-intro}. 
The result (see \refthm{eglamour-overhead-bound} below) is that the \eglamour is linear in the number $\sizem{\exec}$ of $\beta$-steps/transitions and quadratic in the size $\size{\tm_0}$ of the initial term $\tm_0$, \ie its overhead has complexity $O((1+\sizem{\exec})\cdot{\size{\tm_0}}^2)$.

The analysis relies on a quantitative invariant, the crucial \emph{subterm invariant}, ensuring that $\tomachse$ duplicates only subterms of the initial term, so that the cost of duplications is connected to one of the two parameters for complexity analyses.

\newcounter{l:subterm-invariant} %new counter in order to use it in appendix
\addtocounter{l:subterm-invariant}{\value{lemma}}
\begin{lemma}[Subterm Invariant]
\label{l:subterm-invariant} % \reflemmap{eglamour-invariants}{subterm}
	Let 
\NoteProof{lappendix:subterm-invariant}
	$\exec \colon \compil{\tm_0} \tomach^* (\dump,\code,\stack,\genv)$ be an \eglamour execution. 
	Every subterm $\la\var\codetwo$ of \mbox{$\dump$, $\code$, $\stack$, or $\genv$ is a subterm of $\tm_0$.}
\end{lemma}

\paragraph{Intuition About Complexity Bounds.} The number $\sizee\exec$ of substitution transitions $\tomachse$ depends on both parameters for complexity analyses, the number $\sizem{\exec}$ of $\beta$-transitions \emph{and} the size $\size{\tm_0}$ of the initial term. 
Dependency on $\sizem{\exec}$ is standard, and appears in every machine \cite{DBLP:conf/icfp/BlellochG96,DBLP:conf/birthday/SandsGM02,DBLP:conf/icfp/AccattoliBM14,fireballs,DBLP:conf/aplas/AccattoliBM15,DBLP:conf/wollic/Accattoli16}---sometimes it is quadratic, here it is linear, in \refsect{fglamour} we come back to this point. 
Dependency on $\size{\tm_0}$ is also always present, but usually only for \emph{the cost} of a single $\tomachse$ transition, %given by the fact that
since only subterms of $\tm_0$ are duplicated, as ensured by the subterm invariant. For the \eglamour, instead, also \emph{the number} of $\tomachse$ transitions depends---linearly---on $\size{\tm_0}$: this is a side-effect of dealing with open terms. Since both the cost and the number of $\tomachse$ transitions depend on $\size{\tm_0}$, the dependency is quadratic.

%  Let us give a simple example of a family of terms showing the dependency from $\size{\tm_0}$ in isolation (\ie, with no dependency from $\sizem\exec$). 
%  Let $\tm_n \defeq \la\var(\ldots((\vartwo\underbrace{\var)\var)\ldots)\var}_{n}$ and consider the family:
The following family of terms shows the dependency on $\size{\tm_0}$ in isolation (\ie, with no dependency on $\sizem\exec$). 
Let $\tmfour_n \defeq \la\var(\ldots((\vartwo\underbrace{\var)\var)\ldots)\var}_{n}$ and consider:
 
\vspace{-1.5\baselineskip}
\begin{equation}
\tmtwo_n \defeq \tmfour_n \tmfour_n = (\la\var(\ldots((\vartwo\overbrace{\var)\var)\ldots)\var}^{n}) \tmfour_n \tobabs (\ldots((\vartwo\overbrace{\tmfour_n) \tmfour_n)\ldots)\tmfour_n}^{n}\,.
\label{eq:quad-dependency}
\end{equation}
Forgetting about commutative transitions, the \eglamour would evaluate $\tmtwo_n$ with one $\beta$-transition $\tomachsm$ and $n$ substitution transitions $\tomachse$, each one duplicating $\tmfour_n$, whose size (as well as the size of the initial term $\tmtwo_n$) is linear in $n$. 

The number $\sizecom\exec$ of commutative transitions $\tomachc$, roughly, is linear in the amount of code involved in the evaluation process. This amount is given by the initial code plus the code produced by duplications, that is bounded by the number of substitution transitions times the size of the initial term. The number of commutative transitions is then $O((1+\sizem\exec)\cdot\size{\tm_0}^2)$. Since each one has constant cost, this is also a bound to their cost.

\paragraph{Number of Transitions 1: Substitution $vs$ $\beta$ Transitions.} The number $\sizee\exec$ of substitution transitions is proven (see \refcor{exp-bilinear} below) to be bilinear, \ie linear in $\size{\tm_0}$ and $\sizem{\exec}$, by means of a measure.

The \emph{free size} $\sizefree \cdot$ of a code counts the number of
free variable occurrences that are not under abstractions. 
It is defined and extended to states as follows:
%  \scalebox{0.9}{\parbox{\linewidth}{
\begin{align*}
   \sizefree \var & \defeq  1 						&   \sizefree {\stempty} & \defeq   0\\
   \sizefree {\la\vartwo\codetwo} & \defeq  0\						&    \sizefree {\stackitem \cons \stack} & \defeq   \sizefree\stackitem + \sizefree\stack\\
   \sizefree {\code \codetwo} & \defeq  \sizefree \tm + \sizefree \tmtwo	&
   \sizefree {\dump\cons(\code,\stack)} & \defeq   \sizefree\tm + \sizefree\stack + \sizefree\dump
\end{align*}
% 
% \vspace{-1.2\baselineskip}
\begin{equation*}
  \sizefree{ \eglamst \dump \code \stack \genv} \defeq \sizefree\dump + \sizefree\code + \sizefree\stack.
\end{equation*}
% }}

\newcounter{l:free-occ-inv} %new counter in order to use it in appendix
\addtocounter{l:free-occ-inv}{\value{lemma}}
\begin{lemma}[Free Occurrences Invariant]
\label{l:free-occ-inv}
  Let%
\NoteProof{lappendix:free-occ-inv} 
  $\exec \colon \compil{\tm_0} \tomach^* \state$ be an \eglamour execution.
  Then, $\sizefree\state \leq \sizefree{\tm_0} + \size{\tm_0} \cdot \sizem\exec - \sizee \exec$.
\end{lemma}

\newcounter{coro:exp-bilinear} %new counter in order to use it in appendix
\addtocounter{coro:exp-bilinear}{\value{corollary}}
\begin{corollary}[Bilinear Number of Substitution Transitions]
\label{coro:exp-bilinear}
  Let%
\NoteProof{coroappendix:exp-bilinear}
  $\exec: \compil{\tm_0} \tomach^* \state$ be an \eglamour execution.
  Then, $\sizee \exec \leq  (1 + \sizem\exec) \cdot \size{\tm_0}$.
\end{corollary}

\paragraph{Number of Transitions 2: Commutative $vs$ Substitution Transitions.} 
The bound on the number $\sizecom\exec$ of commutative transitions is found by means of a (different) measure on states. 
The bound is linear in $\size{\tm_0}$ and in $\sizee\exec$, which means---by applying the result just obtained in
\refcor{exp-bilinear}---%actually 
\emph{quadratic} in $\size{\tm_0}$ and linear in $\sizem{\exec}$.

 The \emph{commutative size} of a state is defined as $\sizecom{ \eglamst \dump \code \stack \genv } \defeq \size\code + \Sigma_{
 \dentry{\codetwo}{\stacktwo} \in\dump} \size\codetwo$, where $\size\code$ is the usual size of codes.

\newcounter{l:comm-bound}
\addtocounter{l:comm-bound}{\value{lemma}}
\begin{lemma}[Number of Commutative Transitions]
\label{l:comm-bound} % \refcoro{comm-bound}
  Let 
\NoteProof{lappendix:comm-bound}
  $\exec \colon \compil{\tm_0} \tomach^* \state$ be an \eglamour execution. Then 
  $\sizecom\exec 
   \leq  \sizecom\exec + \sizecom\state 
  \leq (1+\sizee\exec)\cdot \size{\tm_0} \in O((1+\sizem\exec) \cdot \size{\tm_0}^2)$.
\end{lemma}

\paragraph{Cost of Single Transitions.} We need to make some hypotheses on how the \eglamour is going to be itself implemented on RAM:
  
  \begin{varenumerate}  
    \item \emph{Variable (Occurrences) and Environment Entries}: a variable is a memory location, a variable occurrence is a reference to it, and an environment entry $\esub\var\stackitem$ is the fact that the location associated to $\var$ contains $\stackitem$.    
    \item \emph{Random Access to Global Environments}: the environment $\genv$ can be accessed in $O(1)$ (in $\tomachse$) by just following the reference given by the variable occurrence under evaluation, with no need to access $\genv$ sequentially, thus ignoring its list structure (used only to ease the definition of the decoding).
  \end{varenumerate}
  
With these hypotheses it is clear that $\beta$ and overhead transitions can be implemented in $O(1)$. The substitution transition $\tomachse$ needs to copy a code from the environment (the renaming $\renamenop\code$) and can be implemented in $O(\size{\tm_0})$, as the subterm to copy is a subterm of $\tm_0$ by the subterm invariant (\reflemma{subterm-invariant}) and the environment can be accessed in $O(1)$.

\paragraph{Summing Up.} By putting together the bounds on the number of transitions with the cost of single transitions we obtain the overhead of the machine.

\newcounter{thm:eglamour-overhead-bound} %new counter in order to use it in appendix
\addtocounter{thm:eglamour-overhead-bound}{\value{theorem}}
\begin{theorem}[\eglamour Overhead Bound]
\label{thm:eglamour-overhead-bound}
  Let% 
\NoteProof{thmappendix:eglamour-overhead-bound}
  $\exec \colon \compil{\tm_0} \tomach^* \state$ be an \eglamour execution. Then $\exec$ is implementable on RAM in $O((1+\sizem{\exec})\cdot \size{\tm_0}^2)$, \ie linear in the number of $\beta$-transitions (aka the length of the derivation $\deriv \colon \tm_0 \torf^* \decode{\state}$ implemented by $\exec$) and quadratic in the size of the initial term $\tm_0$.
\end{theorem}

 % !TEX root = main.tex
\section{\fglamour}
\label{sect:fglamour}
In this section we optimize the \eglamour, obtaining a machine, the \fglamour, whose dependency from the size of the initial term is linear, instead of quadratic, providing a bilinear---thus optimal---overhead (see \refthm{fglamour-overhead-bound} below and compare it with \refthm{eglamour-overhead-bound} on the \eglamour). 
We invite the reader to go back to equation \refeq{quad-dependency} at page \pageref{eq:quad-dependency}, where the quadratic dependency was explained. 
Note that in that example the substitutions of $\tmfour_n$ do not create $\betaf$-redexes, and so they are useless. 
The \fglamour avoids these useless substitutions and it implements the example with no substitutions at all.

\paragraph{Optimization: Abstractions On-Demand.} The difference between the \eglamour and the machines in \cite{fireballs} is that, whenever the former encounters a variable occurrence $\var$ bound to an abstraction $\la\vartwo\code$ in the environment, it replaces $\var$ with $\la\vartwo\code$, while the latter are more parsimonious. 
They implement an optimization that we call \emph{substituting abstractions on-demand}: $\var$ is replaced by $\la\vartwo\code$ only if this is useful to obtain a $\beta$-redex, that is, only if the argument stack is non-empty. 
The \fglamour, defined in \reffig{fglamour}, upgrades the \eglamour with \emph{substitutions of abstractions on-demand}---note the new side-condition for $\tomachcthree$ and the non-empty stack in $\tomachse$.

% !TEX root = main.tex
\begin{figure}[!t]
  \centering
%	\ovalbox{
  \scalebox{0.9}{
  $%\begin{array}{c}
  	{\setlength{\arraycolsep}{0.35em}
  	\begin{array}{c|c|c|c|c|c|c|c|ccc}
		\mbox{Dump} & \mbox{Code} & \mbox{Stack} & \mbox{Global Env} 
		&&
		\mbox{Dump} & \mbox{Code} & \mbox{Stack} & \mbox{Global Env}\\
\hhline{=|=|=|=|=|=|=|=|=}
		\dump & \code\codetwo & \stack & \genv
	  	&\tomachcone&
	  	\!\dump\!\cons\!\dentry\code\stack\! & \codetwo & \stempty &\genv
	  	\\
	  	
		\!\dump\!\cons\!\dentry\code\stack\! & \la\var\codetwo & \stempty & \genv
		& \tomachctwo &
		\dump & \code & \!\pair{\la\var\codetwo}{\stempty} \cons\stack\! & \genv
		\\

		\!\dump\!\cons\!\dentry\code\stack\! & \var & \stacktwo & \genv
		& \tomachcthree &
		\dump & \code & \pair{\var}{\stacktwo}\!\cons\stack & \genv \\
		\multicolumn{9}{r}{\mbox{if $\genv(\var)=\undef$ or $\genv(\var)=\pair{\vartwo}{\stackthree}$ or  ($\genv(\var)=\pair{\la\vartwo\codetwo}\stempty$ and $\stacktwo = \stempty$)}}
		\\
		
		\dump & \la\var\code & \pair{\vartwo}{\stempty}\!\cons\!\stack & \genv
	  	&\tomachmone &
		\dump & \code\isub\var\vartwo & \stack & \genv
	  	\\
	  	
	  	\dump & \la\var\code & \stackitem\cons\stack & \genv
	  	&\tomachmtwo &
		\dump & \code & \stack & \esub\var\stackitem\genv \\
		\multicolumn{9}{r}{\mbox{  if $\stackitem\neq \pair{\vartwo}{\stempty}$}}\\

	  	\dump & \var & \stackitem\cons\stack & \genv_1\esub\var{\pair{\la\vartwo\codetwo}\stempty}\genv_2
		& \tomachse &
		\dump & \!\rename{\la\vartwo\codetwo}\! & \stackitem\cons\stack & \genv_1\esub\var{\pair{\la\vartwo\codetwo}\stempty}\genv_2\\
	\end{array}}
%\end{array}
$
%}
}
\caption{\fglamour (data-structures, decoding, and $\rename{\la\vartwo\codetwo}$ defined as in \reffig{eglamour}).}
\label{fig:fglamour}
  \end{figure}

\paragraph{Abstractions On-Demand and the Substitution of Variables.}The new optimization however has a consequence. To explain it, let us recall the role of another optimization, \emph{no substitution of variables}. In the \eglamour, abstractions are at depth 1 in the environment: there cannot be chains of renamings, \ie of substitutions of variables for variable, ending in abstractions (so, there cannot be chains like $\esub\var {\pair\vartwo\stempty}\esub\vartwo {\pair\varthree\stempty} \esub\varthree {\pair{\la{\varthree'}\code}\stempty}$). This property implies that the overhead is linear in $\sizebeta\exec$ and it is induced by the fact that variables cannot be substituted. 
If variables can be substituted then the overhead becomes quadratic in $\sizebeta\exec$---this is what happens in the \glamour machine in \cite{fireballs}. The relationship between \emph{substituting variables} and a linear/quadratic overhead is studied in-depth in \cite{DBLP:conf/wollic/AccattoliC14}.

Now, because the \fglamour substitutes abstractions on-demand, variable occurrences that are not applied are not substituted by abstractions. The question becomes what to do when the code is an abstraction $\la\var\code$ and the top of the stack argument $\stackitem$ is a simple variable occurrence $\stackitem = \pair\vartwo\stempty$ (potentially bound to an abstraction in the environment $\genv$) because if one admits that $\esub\var{\pair\vartwo\stempty}$ is added to $\genv$ then the depth of abstractions in the environment may be arbitrary and so the dependency on $\sizebeta\exec$ may be quadratic, as in the \glamour. 
There are two possible solutions to this issue. The complex one, given by the \uglamour in \cite{fireballs}, is to add labels and a further unchaining optimization. The simple one is to split the $\beta$-transition in two, handling this situation with a new rule that renames $\var$ as $\vartwo$ in the code $\code$ without touching the environment---this exactly what the \fglamour does with $\tomachmone$ and $\tomachmtwo$. The consequence is that abstractions stay at depth 1 in $\genv$, and so the overhead is indeed bilinear.

The simple solution is taken from Sands, Gustavsson, and Moran's  \cite{DBLP:conf/birthday/SandsGM02}, where they use it on a call-by-name machine. Actually, it repeatedly appears in the literature on abstract machines often with reference to space consumption and \emph{space leaks}, 
% often with reference to space consumption and \emph{space leaks}, 
for instance in Wand's \cite{DBLP:journals/lisp/Wand07}, Friedman et al.'s \cite{DBLP:journals/lisp/FriedmanGSW07}, and Sestoft's \cite{DBLP:journals/jfp/Sestoft97}.

\paragraph{\fglamour.} The machine is in \reffig{fglamour}. Its data-structures, compiling and decoding are exactly as for the \eglamour. 
% For instance, the term $\tm \defeq (\la{\varthree}{\varthree(\vartwo\varthree)})\la{\var}{\var}$ is compiled in the initial state of the execution in \refex{fast}, and the final state of that execution decodes to the term $\vartwo \, \la{\var}{\var}$, where $\tm \torf^* \vartwo \, \la{\var}{\var}$ as shown in \refex{rtof} (see also \refex{easy} where the \eglamour is executed with the same initial state as in \refex{fast}).

\begin{example}
\label{ex:fast}
  Let us now show how the derivation  $\tm \defeq (\la{\varthree}{\varthree(\vartwo\varthree)})\la{\var}{\var} \torf^2 \vartwo \, \la{\var}{\var}$ of \refex{torf} is implemented by the \fglamour.
  The execution is similar to that of the \eglamour in \refex{easy}, since they implement the same derivation and hence have the same initial state. 
  In particular, the first five transitions in the \fglamour (omitted here) are the same as in the \eglamour (see \refex{easy} and replace $\tomachsm$ with $\tomachmtwo$). 
  Then, the \fglamour executes:
  \begin{equation*}
  \scalebox{0.87}{
  $\begin{array}{l|r|c|rc}
    \text{Dump} & \text{Code} & \text{Stack} & \text{Global Environment} \\
    \cline{1-4}  
%     \glamsttab{\stempty}{(\la{\varthree}{\varthree(\vartwo\varthree)})\la{\var}{\var}}{\stempty}{\stempty} &\ \tomachcone \\
%     \glamsttab{\dentry{\la{\varthree}{\varthree(\vartwo\varthree)}}{\stempty}}{\la{\var}{\var}}{\stempty}{\stempty} &\ \tomachctwo \\
%     \glamsttab{\stempty}{\la{\varthree}{\varthree(\vartwo\varthree)}}{\pair{\la{\var}{\var}}{\stempty}}{\stempty} &\ \tomachmtwo \\
%     \glamsttab{\stempty} {\varthree(\vartwo\varthree)} {\stempty} {\esub{\varthree}{\pair{\la{\var}{\var}}{\stempty}}} &\ \tomachcone \\
%     \glamsttab{\dentry{\varthree}{\stempty}} {\vartwo\varthree} {\stempty} {\esub{\varthree}{\pair{\la{\var}{\var}}{\stempty}}} &\ \tomachcone \\
    \glamsttab {\dentry{\varthree}{\stempty} \cons \dentry{\vartwo}{\stempty}} {\varthree} {\stempty} {\esub{\varthree}{\pair{\la{\var}{\var}}{\stempty}}} &\ \tomachcthree \\
    \glamsttab {\dentry{\varthree}{\stempty}} {\vartwo} {\pair{\varthree}\stempty} {\esub{\varthree}{\pair{\la{\var}{\var}}{\stempty}}} &\ \tomachcthree \\
    \glamsttab {\stempty} {\varthree} {\pair{\vartwo}{(\pair{\varthree}{\stempty})}} {\esub{\varthree}{\pair{\la{\var}{\var}}{\stempty}}} &\ \tomachse \\
    \glamsttab {\stempty} {\la{\var''\!}{\var''}} {\pair{\vartwo}{(\pair{\varthree}{\stempty})}} {\esub{\varthree}{\pair{\la{\var}{\var}}{\stempty}}} &\ \tomachmtwo \\
    \glamsttab {\stempty} {\var''} {\stempty} {\esub{\var''}{\pair{\vartwo}{(\pair{\varthree}{\stempty})}} \cons \esub{\varthree}{\pair{\la{\var}{\var}}{\stempty}}}
  \end{array}$
  }
  \end{equation*}
%   The execution is similar to that of the \eglamour in \refex{easy} (they implement the same derivation and hence have the same initial state), but the 
  The \fglamour executes only one substitution transition (%while 
  the \eglamour takes two) since the replacement of $\varthree$  with $\la{\var}{\var}$ from the environment is \emph{on-demand} (\ie useful to obtain a $\beta$-redex) only for the first \mbox{occurrence of $\varthree$ in $\varthree(\vartwo\varthree)$.}
\end{example}

% An example of execution of the \fglamour is in \refex{fast}: its initial state is the compilation of the term $\tm \defeq (\la{\varthree}{\varthree(\vartwo\varthree)})\la{\var}{\var}$, and its final state decodes to the term $\vartwo \, \la{\var}{\var}$, where $\tm \torf^* \vartwo \, \la{\var}{\var}$ as shown in \refex{torf} (see also \refex{easy} where the \eglamour is executed with the same initial state as in \refex{fast}).
The \fglamour satisfies the same invariants (the qualitative ones---the \emph{fireball item} is slightly different---as well as the subterm one, see \withoutproofs{\cite{AccattoliGuerrieri17}}\withproofs{Appendix~\ref{subsect:fglamour-proofs}}) and also forms an implementation system with respect to $\torf$ and $\decode\cdot$% (again, see Appendix~\ref{subsect:fglamour-proofs})
. Therefore, %we obtain,

% 
% The invariants are used to prove the following lemmas.
% 
 \newcounter{l:fglamour-trans-projection} %new counter in order to use it in appendix
 \addtocounter{l:fglamour-trans-projection}{\value{lemma}}
% \begin{lemma}[\fglamour One-Step Weak Simulation]
%   \label{l:fglamour-trans-projection} % \reflemmap{eglamour-trans-projection}{exp-com}
% 	Let 
% \NoteProof{lappendix:fglamour-trans-projection}
% 	$\state$ be a reachable state.
% 	\begin{varenumerate}
% 		\item \label{p:fglamour-trans-projection-exp-com} \emph{Commutative \& Exponential}: 
% 		if $\state\tomachhole{\esym,\csym_{1,2,3}}\statetwo$ then $\decode\state=\decode\statetwo$;
% 		\item \label{p:fglamour-trans-projection-mult}\emph{Multiplicative}: if $\state\tomachhole{\msym_1,\msym_2}\statetwo$ then $\decode\state\torf\decode\statetwo$.	  
% 		
% 	\end{varenumerate}
% \end{lemma}
% 
 \newcounter{l:fglamour-progress} %new counter in order to use it in appendix
 \addtocounter{l:fglamour-progress}{\value{lemma}}
% \begin{lemma}[\fglamour Progress]
% \label{l:fglamour-progress}
%   Let
%   \NoteProof{lappendix:fglamour-progress}
%   $\state$ be a reachable final state. Then $\decode\state$ is a fireball, \ie it is $\betaf$-normal. 
% \end{lemma}
% 
% The theorem of correctness and completeness of the machine with respect to $\torf$ follows. The bisimulation is \emph{weak} because transitions other than $\tomachm$ are invisible on $\firecalc$. For a machine execution $\exec$ we denote with $\size\exec$ (resp. $\sizex\exec$) the number of transitions (resp. $\mathtt{x}$-transitions for $\mathtt{x} \in \set{\msym, \esym, \csym, \ldots}$) in $\exec$.
% 
 \newcounter{thm:weak-bis-fast} %new counter in order to use it in appendix
 \addtocounter{thm:weak-bis-fast}{\value{theorem}}
% \begin{theorem}[Weak Bisimulation]
% \label{thm:weak-bis-fast} % \refthm{weak-bis}{rev-sim}
%   Let 
% % \NoteProof{thmappendix:weak-bis-fast}
%   $\state$ be an initial state of code $\code$.
%   
%  \begin{varenumerate}
%   \item \emph{Simulation:} For every execution $\exec:\state\tomach^*\statetwo$ there exists a derivation $\deriv \colon \decode\state\torf^*\decode\statetwo$ such that $\sizef\deriv = \sizem\exec$;
%   
%   \item \label{p:weak-bis-fast-rev-sim} \emph{Reverse Simulation:} For every derivation $\deriv \colon \code\torf^*\tmtwo$ there is an execution $\exec:\state\tomach^*\statetwo$ such that $\decode\statetwo = \tmtwo$ and $\sizef{\deriv} = \sizem\exec$.
%  \end{varenumerate}
% \end{theorem}
\newcounter{thm:fglamour-implementation}
\addtocounter{thm:fglamour-implementation}{\value{theorem}}
\begin{theorem}[\fglamour Implementation]
\label{thm:fglamour-implementation}
  The 
\NoteProof{thmappendix:fglamour-implementation}
  \fglamour implements right-to-left evaluation $\torf$ in $\firecalc$ (via the decoding $\decode\cdot$).
\end{theorem}

\paragraph{Complexity Analysis.} What changes is the complexity analysis, that, surprisingly, is simpler. First, we focus on \emph{the number} of overhead transitions. The \emph{substitution $vs$ $\beta$ transitions} part is simply trivial. Note that a substitution transition $\tomachse$ is always  immediately followed by a $\beta$-transition, because substitutions are done only \emph{on-demand}---therefore, $\sizee \exec \leq  \sizem\exec + 1$. 
It is easy to remove the +1: executions must have a $\tomachmtwo$ transition before any substitution one, otherwise the environment is empty and no substitutions are possible---thus $\sizee \exec \leq  \sizem\exec$. 

For the \emph{commutative $vs$ substitution transitions} the exact same measure and the same reasoning of the \eglamour provide the same bound, namely $\sizecom\exec \leq (1+\sizee\exec)\cdot \size{\tm_0}$. What improves is the dependency of the commutatives from $\beta$-transitions (obtained by substituting the bound for substitution transitions), that is now linear because so is that of substitutions---so, $\sizecom\exec \leq (1+\sizem\exec)\cdot \size{\tm_0}$.

\newcounter{l:exp-linear}
\addtocounter{l:exp-linear}{\value{lemma}}
\begin{lemma}[Number of Overhead Transitions]
\label{l:exp-bilinear}
  Let%
\NoteProof{lappendix:exp-bilinear}
	$\exec \colon \compil{\tm_0} \tomach^* \state$ be a \fglamour execution. Then,
  \begin{varenumerate}
  \item \emph{Substitution $vs$ $\beta$ Transitions}: $\sizee \exec \leq  \sizem\exec$.
  \item \emph{Commutative $vs$ Substitution Transitions}: $\sizecom\exec \leq (1+\sizee{\exec})\cdot \size{\tm_0} \leq (1+\sizem{\exec})\cdot \size{\tm_0}$.
  \end{varenumerate}
\end{lemma}

\paragraph{Cost of Single Transitions and Global Overhead.} For the cost of single transitions, note that $\tomachc$ and $\tomachmtwo$ have (evidently) cost $O(1)$ while $\tomachse$ and $\tomachmone$ have cost $O(\size{\tm_0})$ by the subterm invariant. 
Then we can conclude with

\newcounter{thm:fglamour-overhead-bound} %new counter in order to use it in appendix
\addtocounter{thm:fglamour-overhead-bound}{\value{theorem}}
\begin{theorem}[\fglamour Bilinear Overhead]
\label{thm:fglamour-overhead-bound}
  Let % 
\NoteProof{thmappendix:fglamour-overhead-bound}
  $\exec \colon \compil{\tm_0} \tomach^* \state$ be a \fglamour execution. Then $\exec$ is implementable on RAM in $O((1+\sizem{\exec})\cdot \size{\tm_0})$, \ie linear in the number of $\beta$-transitions (aka the length of the derivation $\deriv \colon \tm_0 \torf^* \decode{\state}$ implemented by $\exec$) and the size of the initial term.
\end{theorem}
% \begin{proof}
%     The cost of implementing $\exec$ is the sum of the costs of implementing the $\beta$, substitution, and commutative transitions:
%   \begin{varenumerate}
%     \item \emph{$\beta$-Transitions $\tomachmone$ and $\tomachmtwo$}: transition $\tomachmone$  costs $O(\size{\tm_0})$ because the code has to be renamed and by the subterm invariant (\reflemma{subterm-invariant}) the size of the code is bound by $\size{\tm_0}$. Transition $\tomachmtwo$ instead takes constant time. In the worst case all together they cost $O(\sizem{\exec}\cdot\size{\tm_0})$.
% 
%     \item \emph{Substitution Transition $\tomachse$}: by \reflemma{exp-bilinear} we have $\sizee \exec \leq  \sizem\exec$. By the subterm invariant (\reflemma{subterm-invariant}), each substitution step costs at most $O(\size{\tm_0})$, and so their full cost is $O(\sizem{\exec}\cdot{\size{\tm_0}})$.
%     \item \emph{Commutative Transitions $\tomachc$}: by \reflemma{exp-bilinear} $\sizecom\exec \leq (1+\sizem{\exec})\cdot \size{\tm_0}$.
%       Since every commutative transition evidently takes constant time, the whole cost of the commutative transitions is bound by $O((1+\sizem{\exec})\cdot{\size{\tm_0}})$.
%   \end{varenumerate}
%   Then, the cost of implementing $\exec$ is $O((1+\sizem{\exec})\cdot{\size{\tm_0}})$.\qedhere
% \end{proof}

 % !TEX root = main.tex
\section{Conclusions}
\label{sect:conclusions}

\paragraph{Modular Overhead.} The overhead of implementing \ocbv is measured with respect to the size $\size{\tm_0}$ of the initial term and the number $n$ of $\beta$-steps. We showed that its complexity depends crucially on three choices about substitution.

The first is whether to substitute inert terms that are not variables. If they are substituted, as in \gregoire and Leroy's machine \cite{DBLP:conf/icfp/GregoireL02}, then the overhead is exponential in $\size{\tm_0}$ because of open size explosion (\refprop{open-size-explosion}) and the implementation is then unreasonable. 
If they are not substituted, as in the machines studied here and in \cite{fireballs}, then the overhead is polynomial.

The other two parameters are whether to substitute variables, and whether abstractions are substituted whenever or only \emph{on-demand}, and they give rise to the following table of machines and reasonable overheads:
\begin{center}
  \scalebox{0.87}{
\begin{tabular}{r||c|ccc}
% \hline
& \emph{Sub of Abs Whenever} & \emph{Sub of Abs On-Demand}
\\
\hhline{=||=|=}
\emph{Sub of Variables} & Slow \glamour   & \glamour \\
& $O((1+n^2)\cdot \size{\tm_0}^2)$ & $O((1+n^2)\cdot \size{\tm_0})$
\\\hline
\emph{No Sub of Variables} & \eglamour & Fast\,/\,Unchaining \glamour \\
& $O((1+n)\cdot \size{\tm_0}^2)$ & $O((1+n)\cdot \size{\tm_0})$
% \\\hline
\end{tabular}
}
\end{center}
The Slow \glamour has been omitted for lack of space, because it is slow and involved, as it requires the labeling mechanism of the (Unchaining) \glamour developed in \cite{fireballs}. It is somewhat surprising that the \fglamour presented here has the best overhead and it is also the easiest to analyze.

\paragraph{Abstractions On-Demand: \ocbv is simpler than \scbv. %\scbv, and Coq.
} We explained that \gregoire and Leroy's machine for Coq as described in  \cite{DBLP:conf/icfp/GregoireL02} is unreasonable. 
Its actual implementation, on the contrary, does not substitute non-variable inert terms, so it is reasonable for \ocbv. None of the versions, however, substitutes abstractions on-demand (nor, to our knowledge, does any other implementation), despite the fact that it is a necessary optimization in order to have a reasonable implementation of \scbv, as we now show. Consider the following size exploding family (obtained by applying $\tmthree_n$ to the identity $I \defeq \la{\var}{\var}$), from \cite{accattoliWPTE}:
% $$\begin{array}{ccc@{\hspace{.5cm}}ccc@{\hspace{.5cm}}ccc@{\hspace{.5cm}}ccc}
% \tmthree_1 & \defeq & \la{\var}\la{\vartwo}(\vartwo \var \var)  & 
% \tmthree_{n+1} & \defeq &\la{\var}(\tmthree_n (\la{\vartwo}(\vartwo \var \var))) & 
% \tmfour_0 & \defeq & I &
% \tmfour_{n+1} & \defeq & \la{\vartwo}(\vartwo \tmfour_{n} \tmfour_{n})	
% \end{array}$$
\begin{align*}
  \tmthree_1 & \defeq \la{\var}\la{\vartwo}(\vartwo \var \var)  & 
  \tmthree_{n+1} & \defeq \la{\var}(\tmthree_n (\la{\vartwo}(\vartwo \var \var))) & 
  \quad \tmfour_0 & \defeq I &
  \tmfour_{n+1} & \defeq \la{\vartwo}(\vartwo \tmfour_{n} \tmfour_{n})
\end{align*}

\newcounter{prop:abs-size-explosion}
\addtocounter{prop:abs-size-explosion}{\value{proposition}}
\begin{proposition}[Abstraction Size Explosion]
\label{prop:abs-size-explosion}
  Let 
  \NoteProof{propappendix:abs-size-explosion}
  $n \!>\! 0$. Then $\tmthree_n I \tobabs^n  \tmfour_n$. Moreover, $\size{\tmthree_n I} = O(n)$, $\size{\tmfour_n} = \Omega(2^n)$, $\tmthree_n I$ is closed, and $\tmfour_n$ is normal.
\end{proposition}
The evaluation of $\tmthree_n I$ produces $2^n$ non-applied copies of $I$ (in $\tmfour_n$), so a strong evaluator not substituting abstractions on-demand must have an exponential overhead. Note that evaluation is weak but the $2^n$ copies of $I$ are substituted under abstraction: this is why machines for Closed and Open \cbv can be reasonable without substituting abstractions on-demand.

\paragraph{The Danger of Iterating \ocbv Naively.} The size exploding example in \refprop{abs-size-explosion} also shows that iterating reasonable machines for \ocbv is subtle, as it may induce unreasonable machines for \scbv, if done naively. Evaluating \scbv by iterating the \eglamour (that does not substitute abstractions on-demand), indeed, induces an exponential overhead, while iterating the \fglamour provides an efficient implementation.

\paragraph{Acknowledgements.} This work has been partially funded by the ANR JCJC grant COCA HOLA (ANR-16-CE40-004-01).
%The statement we are going to prove is in fact more general, about $\tmthree_n \tmfour_m$ instead of just $\tmthree_n I$, in order to obtain a simple inductive proof.
%
%\begin{proposition}[Closed and Strategy-Independent Size-Explosion]
%\label{prop:abs-size-explosion}
%  Let $n \!>\! 0$. Then $\tmthree_n \tmfour_m \tob^n  \tmfour_{n+m}$, and in particular $\tmthree_n I \tob^n  \tmfour_n$. Moreover, $\size{\tmthree_n I} = O(n)$, $\size{\tmfour_n} = \Omega(2^n)$, $\tmthree_n I$ is closed, and $\tmfour_n$ is normal.
%\end{proposition}
%
%\begin{proof}
%By induction on $n$. The base case: $\tmthree_1 \tmfour_m = \la{\var}\la{\vartwo}(\vartwo \var \var) \tmfour_m \tob (\la{\vartwo}(\vartwo \tmfour_m \tmfour_m)) = \tmfour_{m+1}$. The inductive case: $\tmthree_{n+1} \tmfour_m = \la{\var}(\tmthree_n (\la{\vartwo}(\vartwo \var \var))) \tmfour_m  \tob \tmthree_n (\la{\vartwo}(\vartwo \tmfour_m \tmfour_m)) = \tmthree_n \tmfour_{m+1} \tob^n \tmfour_{n+m+1}$, where the second sequence is obtained by the \ih The rest of the statement is immediate.
%\end{proof}
%

% Note that the evaluation of the family is weak but it substitutes an exponential number of abstractions under abstraction, and so closed and open machines never do them anyway.
% why machines for \ccbv and \ocbv are not concerned with the question of substituting abstractions on-demand

\phantomsection
\addcontentsline{toc}{section}{References}
\bibliographystyle{splncs03}
\bibliography{biblio} %For Arxiv

\withproofs{
\newpage
\appendix
\chapter*{Technical Appendix}
\newcounter{appendix}
\setcounter{appendix}{\value{theorem}}%inizializzo il contatore appendix al valore dell'ultimo teorema/lemma/proposizione/definizione nella sezione precedente

% !TEX root = main.tex
\section{Rewriting Theory: Definitions, Notations, and Basic Results}
\label{app:rewriting}

Given %a set $I$ and 
a binary relation $\to_\mathsf{r}$ on a set $I$, the reflexive-transitive (resp.~reflexive; transitive; reflexive-transitive and symmetric) closure of $\to_\mathsf{r}$ is denoted by $\to^*$ (resp.~$\to_\mathsf{r}^=$; $\to_\mathsf{r}^+$; $\simeq_\mathsf{r}$).
The transpose of $\to_\mathsf{r}$ is denoted by \!\!$\lRew{\mathsf{r}}$\!\!.
A (\emph{$\mathsf{r}$-})\emph{derivation $\deriv$ from $\tm$ to~$\tmtwo$}, denoted by $\deriv \colon \tm \to_\mathsf{r}^* \tmtwo$, is a finite sequence $(\tm_i)_{0 \leq i \leq n}$ of elements of $I$ (with $n \in \nat$) s.t.~$\tm = \tm_0$, $\tmtwo = \tm_n$ and $\tm_i \to_\mathsf{r} \tm_{i+1}$ for all $1 \leq i < n$;
% we set $\size\deriv = n$ (so, $\size\deriv$ is the length, \ie the number of $\to$-steps, of $\deriv$); also, if $\to_1 \, \subseteq \, \to$, $\size{\deriv}_1$ is the number of $\to_1$-steps in $\deriv$.

The \emph{number of $\mathsf{r}$-steps of a derivation $\deriv$}, \ie its \emph{length}, is denoted by $\size{\deriv}_\mathsf{r} \defeq n$, or simply $\size\deriv$. If $\to_\mathsf{r} \, = \, \to_1 \!\cup \to_2$ with $\to_1 \!\cap \to_2 \, = \emptyset$, $\size{\deriv}_i$ is the number of $\to_i$-steps in $\deriv$, for $i = 1,2$.
We say that:
\begin{itemize}
  \item $\tm \!\in\! I$ is \emph{$\mathsf{r}$-normal} or a \emph{$\mathsf{r}$-normal form} if 
%   there is no $\tmtwo \!\in\! I$ s.t.~$\tm \to_\mathsf{r} \tmtwo$; 
  $\tm \not\to_\mathsf{r} \tmtwo$ for all $\tmtwo \!\in\! I$;
  %\item 
  $\tmtwo \in I$ is a \emph{$\mathsf{r}$-normal form of} $\tm %\in I
  $ if $\tmtwo$ is $\mathsf{r}$-normal and $\tm \to_\mathsf{r}^* \tmtwo$; 
  \item $\tm \in I$ is \emph{$\mathsf{r}$-normalizable} if there is a $\mathsf{r}$-normal $\tmtwo \in I$ s.t. $\tm \to_\mathsf{r}^* \tmtwo$;
  $\tm$ is \emph{strongly $\mathsf{r}$-normalizable} if there is no infinite sequence $(\tm_i)_{i \in \nat}$ s.t. $\tm_0 = \tm$ and $\tm_i \to_\mathsf{r} \tm_{i+1}$;
% \end{itemize}
% \begin{itemize}
  \item a $\mathsf{r}$-derivation $\deriv \colon \tm \to_\mathsf{r}^* \tmtwo$ is (\emph{$\mathsf{r}$-})\emph{normalizing} if $\tmtwo$ is $\mathsf{r}$-normal;
  \item $\to_\mathsf{r}$ is \emph{strongly normalizing} if all $\tm \!\in\! I$ is strongly $\mathsf{r}$-normalizable;
%   \item $\to$ is \emph{confluent} if, given $a \to^* b$ and $a \to^* c$, there is $a' \in I$ such that $b \to^* a'$ and $c \to^* a'$;
  \item $\to_\mathsf{r}$ is \emph{strongly confluent} if, for all $\tm, \tmtwo, \tmthree \in\! I$ s.t.~$\tmthree \lRew{\mathsf{r}} \tm \to_\mathsf{r} \tmtwo$ and $\tmtwo \neq \tmthree$, there is $\tmfour \in I$ s.t. $\tmthree \to_\mathsf{r} \tmfour \lRew{\mathsf{r}} \tmtwo$;
  $\to_\mathsf{r}$ is \emph{confluent} \mbox{if $\to_\mathsf{r}^*\!$ is strongly confluent.}
\end{itemize}

Let $\to_1, \to_2 \, \subseteq I \times I$. 
Composition of relations is denoted by juxtaposition: for instance, $\tm \to_1\to_2 \tmtwo$ means that there is $\tmthree \in I$ s.t. $\tm \to_1 \tmthree \to_2 \tmtwo$; for any $n \in \nat$, $\tm \to_1^n \tmtwo$ means that there is a $\to_1$-derivation with length $n$ ($\tm = \tmtwo$ for $n = 0$).
We say that $\to_1$ and $\to_2$ \emph{strongly commute} if, for any $\tm, \tmtwo, \tmthree \in I$ s.t. $\tmtwo \lRew{1} \tm \to_2 \tmthree$, one has $\tmtwo \neq \tmthree$ and there is $\tmfour \in I$ s.t. $\tmtwo \to_2 \tmfour \lRew{1} \tmthree$.
Note that if $\to_1$ and $\to_2$ strongly commute and $\to \, = \, \to_1 \!\cup \to_2$, then for any derivation $\deriv \colon \tm \to^* \tmtwo$ the sizes $\size{\deriv}_1$ and $\size{\deriv}_2$ are uniquely determined.

The following proposition collects some basic and well-known results of rewriting theory.

\begin{proposition}\label{prop:basic-confluence}
  Let $\to_\mathsf{r}$ be a binary relation on a set $I$.
  \begin{enumerate}
    \item\label{p:basic-confluence-unique-normal} If $\to_\mathsf{r}$ is confluent then: 
    \hfill
    \begin{enumerate}
      \item every $\mathsf{r}$-normalizable term has a unique $\mathsf{r}$-normal form;
      \item\label{p:basic-confluence-equivalent} for all $\tm, \tmtwo \in I$, $\tm \simeq_\mathsf{r} \tmtwo$ iff there is $\tmthree \in I$ s.t.~$\tm \to_\mathsf{r}^* \tmthree \lRewn{\mathsf{r}} \tmtwo$.
    \end{enumerate}
    \item\label{p:basic-confluence-strong} If $\to_\mathsf{r}$ is strongly confluent then $\to_\mathsf{r}$ is confluent and, for any $\tm \in I$, one has:
    \begin{enumerate}
      \item\label{p:basic-confluence-strong-length} all normalizing $\mathsf{r}$-derivations from $\tm$ have the same length;
      \item\label{p:basic-confluence-strong-normalizable} $\tm$ is strongly $\mathsf{r}$-normalizable if and only if $\tm$ is $\mathsf{r}$-normalizable.
    \end{enumerate}
  \end{enumerate}
\end{proposition}

As all incarnations of Open CBV we consider are confluent, the use of \refpropp{basic-confluence}{unique-normal} is left implicit.

For $\firecalc$ and $\vsubcalc$, we use \refpropp{basic-confluence}{strong} and the following more informative version of Hindley--Rosen Lemma, whose proof is just a more accurate reading of the proof in \cite[Prop.\,3.3.5.(i)]{Barendregt84}:

\begin{lemma}[Strong Hindley--Rosen]\label{l:hindley-rosen}
  Let $\to \, = \, \to_1 \cup \to_2$ be a binary relation on a set $I$ s.t. $\to_1$ and $\to_2$ are strongly confluent.
  If $\to_1$ and $\to_2$ strongly com\-mute, then $\to$ is strongly confluent and, 
% for all $\tm \in I$ and normalizing reduction sequences $\deriv \colon \tm \to^* \tmtwo$ and $\derivtwo \colon \tm \to^* \tmtwo$, \mbox{one has $\size\deriv = \size\derivtwo$, $\size{\deriv}_1 = \size{\derivtwo}_1$ and $\size{\deriv}_2 = \size{\derivtwo}_2$.}
 for any $\tm \in I$ and any normalizing derivations $\deriv$ and $\derivtwo$ from $\tm$, one has $\size\deriv = \size\derivtwo$, $\size{\deriv}_1 = \size{\derivtwo}_1$ and $\size{\deriv}_2 = \size{\derivtwo}_2$.
\end{lemma}
% !TEX root = main.tex
\section{Omitted Proofs}
\label{app:omitted}

%\subsection{Proofs of Section~\ref{sect:cost-models} (How to Stop Worrying and Love the Bomb)}
%
%\input{B4_-_How_to_Stop_Worrying_and_Love_the_Bomb}

\subsection{Proofs of Section~\ref{sect:fireball} (The Fireball Calculus)}

% !TEX root = main.tex
In this section we start by recalling the relevant properties of the fireball calculus, that have been omitted from the paper for lack of space. For the sake of completeness we include all proofs, but most of them are actually taken from (the long versions of) our previous works \cite{DBLP:conf/aplas/AccattoliG16,fireballs}.

\paragraph{Distinctive Properties of $\firecalc$.} To prove the distinctive properties of $\firecalc$ we need the following auxiliary lemma.

\begin{lemma}[Values and inert terms are $\betaf$-normal]
\label{l:fnormal}
\hfill
  \begin{enumerate}
    \item\label{p:fnormal-value} Every abstraction is $\betaf$-normal.
%     \item\label{p:fnormal-inert-nobetaredex} Every inert term contains no $\beta$-redexes (\ie subterms of the form $(\la\var\tm)\tmtwo$).%FALSE! Under a $\lambda$ a term might have a $\beta$-redex
    \item\label{p:fnormal-inert} Every inert term is $\betaf$-normal.
  \end{enumerate}
\end{lemma}

\begin{proof}\hfill
  \begin{enumerate}
    \item Immediate, since $\tof$ does not reduce under $\l$'s.
    
    \item By induction on the definition of inert term $\gconst$.
  
    \begin{itemize}
%       \item If $\gconst = \var\val$ then $\val$ is $\betaf$-normal by \reflemmap{fnormal}{value}, hence $\gconst$ is $\betaf$-normal.
%       
%       \item If $\gconst = \var\gconsttwo$ then $\gconsttwo$ is $\betaf$-normal by \ih, hence $\gconst$ is $\betaf$-normal.
      \item If $\gconst = \var$ then $\gconst$ is obviously $\betaf$-normal.
%       \item If $\gconst = \gconsttwo\!\val$ then $\gconsttwo\!$ and $\val$ are $\betaf$-normal by \ih and \reflemmap{fnormal}{value} re\-spectively, besides $\gconsttwo$\! is not an abstraction, \mbox{so $\gconst$ is $\betaf$-normal.}
      \item If $\gconst = \gconsttwo\la\var\tm$ then $\gconsttwo$ and $\la\var\tm$ are $\betaf$-normal by \ih and \reflemmap{fnormal}{value} re\-spectively, besides $\gconsttwo$\! is not an abstraction, \mbox{so $\gconst$ is $\betaf$-normal.}

      \item Finally, if $\gconst = \gconsttwo\gconstthree$ then $\gconsttwo$ and $\gconstthree$ are $\betaf$-normal by \ih, moreover $\gconsttwo$ is not an abstraction, hence $\gconst$ is $\betaf$-normal.
      \qedhere
    \end{itemize}
  \end{enumerate}
\end{proof}

The following proposition collects two main features of $\firecalc$, showing why it is interesting to study this calculus.
\refpoint{distinctive-fireball-open-harmony} generalizes the property of \ccbv, that we like to call \emph{harmony}, for which a closed term is $\betav$-normal iff it is a value. \refpoint{distinctive-fireball-conservative} instead states that if the evaluation of a closed term $\tm$ in the fireball calculus is \emph{exactly} the evaluation  of $\tm$ in \ccbv. This second point is an observation that never appeared in print before.

\setcounter{propositionAppendix}{\value{prop:distinctive-fireball}}
\begin{propositionAppendix}[Distinctive Properties of $\firecalc$]
\label{propappendix:distinctive-fireball}
  Let $\tm$ be a term.
\NoteState{prop:distinctive-fireball}
  \begin{varenumerate}
    \item\label{pappendix:distinctive-fireball-open-harmony}  \emph{Open Harmony}:  $\tm$ is $\betaf$-normal iff $\tm$ is a fireball.
    \item\label{pappendix:distinctive-fireball-conservative} \emph{Conservative Open Extension}: $\tm \tof \tmtwo$ iff $\tm \tobabs \tmtwo$ whenever $\tm$ is closed.
  \end{varenumerate}

%   A (possibly open) term $\tm$ is a fireball iff it is a $\betaf$-normal form.
%   \NoteState{prop:open-harmony}
\end{propositionAppendix}

\begin{proof}\hfill
  \begin{enumerate}
    \item   
    \begin{description}
      \item [$\Rightarrow$:] Proof by induction on the term $\tm$.
      If $\tm$ is a value then $\tm$ is a fireball.
      
      Otherwise $\tm = \tmtwo\tmthree$ for some terms $\tmtwo$ and $\tmthree$.
      Since $\tm$ is $\betaf$-normal, then $\tmtwo$ and $ \tmthree$ are $\betaf$-normal, and either $\tmtwo$ is not an abstraction or $\tmthree$ is not a fireball.
      By induction hypothesis, $\tmtwo$ and $\tmthree$ are fireballs. 
      Summing up, $\tmtwo$ is either a variable or an inert term, and $\tmthree$ is a fireball, therefore $\tm = \tmtwo\tmthree$ is an inert term and hence a fireball.
      
      \item [$\Leftarrow$:] By hypothesis, $\tm$ is either a value or an inert term. 
      If $\tm$ is a value, then it is $\betaf$-normal by \reflemmap{fnormal}{value}. 
      Otherwise $\tm$ is an inert term and then it is $\betaf$-normal by \reflemmap{fnormal}{inert}.
    \end{description}

    \item 
    \begin{description}
      \item [$\Rightarrow$:] By induction on the definition of $\tm \tof \tmtwo$. Cases:
      \begin{itemize}
        \item \emph{Step at the root}, \ie $\tm = (\la\var\tmthree)\fire \rtof \tmthree\isub\var{\fire} = \tmtwo$. Since $\tm$ is closed, then $\fire$ is closed and hence cannot be an inert term, therefore $\fire$ is a (closed) abstraction and thus $\tm = (\la\var\tmthree)\fire \rtobabs \tmthree\isub\var{\fire} = \tmtwo$.
        
        \item \emph{Left Application}, \ie $\tm = \tmthree\tmfour \tof \tmthreep\tmfour = \tmtwo$ with $\tmthree \tof \tmthreep$. Since $\tm$ is closed, $\tmthree$ is closed and hence $\tmthree \tobv \tmthreep$ by \ih, so $\tm = \tmthree\tmfour \tobabs \tmthreep\tmfour = \tmtwo$.
        
        \item \emph{Right Application}, \ie $\tm = \tmfour\tmthree \tof \tmfour\tmthreep = \tmtwo$ with $\tmthree \tof \tmthreep$. Since $\tm$ is closed, $\tmthree$ is closed and hence $\tmthree \tobv \tmthreep$ by \ih, so $\tm = \tmfour\tmthree \tobabs \tmfour\tmthreep = \tmtwo$.
      \end{itemize}

      \item [$\Leftarrow$:] We have $\tobabs \, \subseteq \, \tof$ since variables and abstractions are fireballs.
      \qedhere
    \end{description}

  \end{enumerate}

\end{proof}

\paragraph{Operational Properties of $\firecalc$.} The rewriting theory of the fireball calculus is very well behaved. The following propositions resumes its main properties.
% \newcounter{prop:basic-fireball} %new counter in order to use it in appendix
% \addtocounter{prop:basic-fireball}{\value{proposition}}
\begin{propositionAppendix}[Operational Properties of $\firecalc$]
\label{propappendix:basic-fireball}\hfill % \refpropp{basic-fireball}{strong-confluence}
\NoteState{prop:basic-fireball}
\begin{varenumerate}
	  \item\label{pappendix:basic-fireball-toin-strong-normalization}\label{p:basic-fireball-toin-strong-confluence} $\toin$ is strongly normalizing and strongly confluent.
	  \item\label{pappendix:basic-fireball-tobv-toin-strong-commutation} $\tobabs$ and $\toin$ strongly commute.	  
	  \item\label{pappendix:basic-fireball-strong-confluence} \label{pappendix:basic-fireball-number-steps} $\tof$ is strongly confluent, and all $\betaf$-normalizing derivations $\deriv$ from a term $\tm$ (if any) have the same length $\sizef{\deriv}$, the same number $\sizebabs{\deriv}$ of $\betaabs$-steps, and the same number $\sizein{\deriv}$ of $\betain$-steps.  
  \end{varenumerate}
\end{propositionAppendix}

\begin{proof}
  See \cite[Proposition~3]{DBLP:conf/aplas/AccattoliG16}.
\end{proof}

\paragraph{The Right-to-Left Strategy.} Here we prove the the well-definedness of the right-to-left strategy.

\setcounter{lemmaAppendix}{\value{l:prop-of-torf}}
\begin{lemmaAppendix}[Properties of $\torf$]
\label{lappendix:prop-of-torf}
  Let
\NoteState{l:prop-of-torf}
  $\tm$ be a term.
  \begin{varenumerate}
	  \item \label{pappendix:prop-of-torf-compl} \emph{Completeness}: $\tm$ has $\tof$-redex iff $\tm$  has a $\torf$-redex.
	  \item \label{pappendix:prop-of-torf-determ} \emph{Determinism}: $\tm$ has at most one $\torf$ redex.
  \end{varenumerate}
\end{lemmaAppendix}

\begin{proof}\hfill
\begin{enumerate}
\item  ($\Rightarrow$) Immediate, as right contexts are in particular evaluation contexts, and thus $\torf \, \subseteq \, \tof$. 

($\Leftarrow$) Let $\evctx$ the evaluation context of the rightmost redex of $\tm$. We show that $\evctx$ is a right context. By induction on $\evctx$. Cases:
 \begin{enumerate}
  \item \emph{Empty}, \ie $\evctx = \ctxhole$. Then clearly $\evctx$ is a right context.
  \item \emph{Right Application}, \ie $\tm = \tmtwo \tmthree$ and $\evctx = \tmtwo \evctxtwo$. By \ih $\evctxtwo$ is a right context in $\tmthree$ and so is $\evctx$ with respect to $\tm$.
  \item \emph{Left Application}, \ie $\tm = \tmtwo \tmthree$ and $\evctx = \evctxtwo \tmthree$. By \ih $\evctxtwo$ is a right context in $\tmtwo$. Since $\evctx$ is the rightmost evaluation context, $\tmthree$ is $\tof$-normal, and so by open harmony (\refpropp{distinctive-fireball}{open-harmony}) it is a fireball. Therefore $\evctx$ is a right context.
\end{enumerate}
\item By induction on $\tm$. Note that by completeness of $\torf$ (\refpoint{prop-of-torf-compl}) open harmony (\refpropp{distinctive-fireball}{open-harmony}) holds with respect to $\torf$, \ie a term is $\torf$-normal iff it is a fireball. We use this fact implicitly in the following case analysis.	 
	 Cases:
	\begin{itemize}
		\item \emph{Value}. No redexes.
		\item \emph{Application $\tm = \tmtwo\tmthree$}. By \ih, there are two cases for $\tmthree$:
		\begin{enumerate}
			\item \emph{$\tmthree$ has exactly one $\torf$ redex}. Then $\tm$ has a $\torf$ redex, because $\tmtwo\ctxhole$ is an evaluation context. Moreover, no $\torf$ redex for $\tm$ can lie in $\tmtwo$, because by open harmony $\tmthree$ is not a \fireball, and so $\ctxhole\tmthree$ is not a right context.
			\item \emph{$\tmthree$ has no $\torf$ redexes}. Then $\tmthree$ is a \fireball. Consider $\tmtwo$. By \ih, there are two cases:
			\begin{enumerate}
				\item \emph{$\tmtwo$ has exactly one $\torf$ redex}. Then $\tm$ has a $\torf$ redex, because $\ctxhole\tmthree$ is an evaluation context and $\tmthree$ is a \fireball. Uniqueness follows from the fact that $\tmthree$ has no $\torf$ redexes.
				\item \emph{$\tmtwo$ has no $\torf$ redexes}. By open harmony $\tmtwo$ is a \fireball, and there are two cases:
				\begin{itemize}
					\item \emph{$\tmtwo$ is an inert term $\gconst$ or a variable $\var$}. Then $\tm$ is a \fireball. 
					\item \emph{$\tmtwo$ is an abstraction $\la\var\tmfour$}. Then $\tm = (\la\var\tmfour) \tmthree$ is a $\torf$-redex, because $\tmthree$ is a \fireball. Moreover, there are no other $\torf$ redexes, because evaluation does not go under abstractions and $\tmthree$ is a \fireball.\qedhere
				\end{itemize}
			\end{enumerate}
		\end{enumerate}
	\end{itemize}
\end{enumerate}
\end{proof}

\paragraph{Open Size Explosion.} The proof of open size explosion is a particularly simple induction on the index of the size exploding family.

\setcounter{propositionAppendix}{\value{prop:open-size-explosion}}
\begin{propositionAppendix}[Open Size Explosion]
\label{propappendix:open-size-explosion}
Let $n\in \nat$.
\NoteState{prop:open-size-explosion}
Then $\tm_n \toin^n \gconst_n$, moreover $\size{\tm_n} = O(n)$, $\size{\gconst_n} = \Omega(2^n)$, and $\gconst_n$ is an inert term.
\end{propositionAppendix}

\begin{proof}
By induction on $n$. The base case is immediate. The inductive case: $\tm_{n+1} = (\la\var \var\var) \tm_n \toin^n (\la\var \var\var) \gconst_n \toin \gconst_n \gconst_n = \gconst_{n+1}$, where the first sequence is obtained by the \ih The bounds on the sizes are immediate, as well as the fact that $\gconst_{n+1}$ is inert.
\end{proof}

\paragraph{Circumventing Open Size Explosion.} To prove that the substitution of inert terms can be avoided we need two auxiliary simple but technical lemmas about substitution, fireballs, and reductions.

\begin{lemma}[Fireballs are Closed Under Substitution and Anti-Substitution of \Quiet Terms]
\label{l:inert-anti-sub-bis} % \reflemmap{inert-anti-sub}{value}
  Let $\tm$ be a term and $\gconst$ be an inert term.
  \begin{varenumerate}
    \item \label{p:inert-anti-sub-bis-abs} $\tm\isub\var\gconst$ is an abstraction iff $\tm$ is an abstraction. %More precisely, if $\tm\isub\var\gconst = \la\vartwo\tmtwo$ then $\tm = \la\vartwo\tmthree$ where $\tmtwo = \tmthree\isub\var\gconst$ and $\vartwo \notin \fv{\gconst} \cup \{\var\}$;          
    \item \label{p:inert-anti-sub-bis-inert} $\tm\isub\var\gconst$ is an inert term iff $\tm$ is an inert term;
    \item\label{p:inert-anti-sub-bis-fire} $\tm\isub\var\gconst$ is a fireball iff $\tm$ is a fireball.
  \end{varenumerate}
\end{lemma}

\begin{proof}\hfill
  \begin{enumerate}
    \item If $\tm\isub\var\gconst = \la\vartwo\tmthree$ then we can suppose without loss of generality that $\vartwo \notin \fv{\gconst} \cup \{\var\}$ and thus there is $\tmfour$ such that $\tmthree = \tmfour\isub\var\gconst$, hence  $\tm\isub\var\gconst = \la\vartwo(\tmfour\isub\var\gconst) = (\la\vartwo\tmfour)\isub\var\gconst$, therefore $\tm = \la\vartwo\tmfour$ is an abstraction.   

    Conversely, if $\tm = \la\vartwo\tmthree$ then we can suppose without loss of generality that $\vartwo \notin \fv{\gconst} \cup \{\var\}$ and thus $\tm\isub\var\gconst = \la\vartwo{(\tmthree\isub\var\gconst)}$ which is an abstraction.    
    
    \item ($\Rightarrow$): By induction on the inert structure of $\tm\isub\var\gconst$. Cases:
    \begin{itemize}
      \item \emph{Variable}, \ie $\tm\isub\var\gconst = \vartwo$, possibly with $\var = \vartwo$. Then $\tm = \var$ or $\tm = \vartwo$, and in both cases $\tm$ is inert.
      \item \emph{Compound Inert}, \ie $\tm\isub\var\gconst = \gconsttwo \fire$. If $\tm$ is a variable then it is inert. Otherwise it is an application $\tm = \tmtwo \tmthree$, and so $\tmtwo\isub\var\gconst = \gconsttwo$ and $\tmthree\isub\var\gconst = \fire$. By \ih, $\tmtwo$ is an inert term. Consider $\fire$. Two cases:
      
      \begin{enumerate}
	\item $\fire$ is an abstraction. Then by \refpoint{inert-anti-sub-bis-abs} $\tmthree$ is an abstraction.
	\item $\fire$ is an inert term. Then by \ih $\tmthree$ is an inert term.
      \end{enumerate}
      
      In both cases $\tmthree$ is a fireball, and so $\tm = \tmtwo \tmthree$ is an inert term.      
      \end{itemize}

      ($\Leftarrow$): By induction on the inert structure of $\tm$. Cases:
    \begin{itemize}
      \item \emph{Variable}, \ie either $\tm = \var$ or $\tm = \vartwo$: in the first case  $\tm\isub\var\gconst = \gconst$, in the second case  $\tm\isub\var\gconst = \vartwo$; in both cases $\tm\isub\var\gconst$ is  an inert term.
      \item \emph{Compound Inert}, \ie $\tm = \gconsttwo \fire$. Then $\tm\isub\var\gconst = \gconsttwo\isub\var\gconst \fire\isub\var\gconst$.
      By \ih, $\gconsttwo\isub\var\gconst$ is an inert term. Concerning $\fire$, there are two cases:
      \begin{enumerate}
	\item $\fire$ is an abstraction. Then by \refpoint{inert-anti-sub-bis-abs} $\fire\isub\var\gconst$ is an abstraction.
	\item $\fire$ is an inert term. Then by \ih $\fire\isub\var\gconst$ is an inert term.
      \end{enumerate}
      In both cases $\fire\isub\var\gconst$ is a fireball, and hence $\tm\isub\var\gconst = \gconsttwo\isub\var\gconst\fire\isub\var\gconst$ is an inert term.      
      \end{itemize}
    
    \item Immediate consequence of \reflemmasps{inert-anti-sub-bis}{abs}{inert}, since every fireball is either an abstraction or an inert term.
    \qedhere
  \end{enumerate}
\end{proof}

\begin{lemma}[Substitution of \Quiet Terms Does Not Create $\betaf$-Redexes]
\label{l:inert-anti-red-bis} % \reflemmap{inert-anti-red}{value}
  Let $\tm, \tmtwo$~be terms and $\gconst$ be an inert term.
  There is a term $\tmthree$ such that:
  \begin{varenumerate}
    \item \label{p:inert-anti-red-bis-betav} if $\tm\isub\var\gconst \tobabs \tmtwo$ then 
    $\tm \tobabs \tmthree$ and $\tmthree\isub\var\gconst = \tmtwo$;

    \item \label{p:inert-anti-red-bis-inert} if $\tm\isub\var\gconst \toin \tmtwo$ then   $\tm \toin \tmthree$ and $\tmthree\isub\var\gconst = \tmtwo$.
  \end{varenumerate}
\end{lemma}

\begin{proof}
  We prove the two points by induction on the evaluation context closing the root redex. Cases:  
    \begin{itemize}
      \item \emph{Step at the root}:
      \begin{enumerate}
      \item \emph{Abstraction Step}, \ie $\tm\isub\var\gconst \defeq (\la\vartwo\tmfour\isub\var\gconst) \tmfive\isub\var\gconst \allowbreak\rtobabs\allowbreak \tmfour\isub\var\gconst\isub\vartwo{\tmfive\isub\var\gconst} \eqdef \tmtwo$.
      By \reflemmap{inert-anti-sub-bis}{abs}, $\tmfive$ is an abstraction, since $\tmfive\isub\var\gconst$ is an abstraction by hypothesis. Then $\tm = (\la\vartwo\tmfour) \tmfive \rtobabs \tmfour\isub\vartwo{\tmfive}$. Then $\tmthree \defeq \tmfour\isub\var{\tmfive}$ verifies the statement, as $\tmthree\isub\var\gconst = (\tmfour\isub\vartwo{\tmfive})\isub\var\gconst = \tmfour\isub\var\gconst\isub\vartwo{\tmfive\isub\var\gconst} = \tmtwo$.     
      
      \item \emph{Inert Step}, identical to the abstraction subcase, just replace \emph{abstraction} with \emph{inert term} and the use of  \reflemmap{inert-anti-sub-bis}{abs} with the use of \reflemmap{inert-anti-sub-bis}{inert}.
      \end{enumerate}      
      
      \item \emph{Application Left}, \ie $\tm = \tmfour\tmfive$ and reduction takes place in $\tmfour$:
      \begin{enumerate}
      \item \emph{Abstraction Step}, \ie $\tm\isub\var\gconst \defeq \tmfour\isub\var\gconst \tmfive\isub\var\gconst \tobabs \tmsix \tmfive\isub\var\gconst \eqdef \tmtwo$. 
      By \ih there is a term $\tmthreep$ such that $\tmsix = \tmthreep\isub\var\gconst$ and $\tmfour \tobabs \tmthreep$. Then $\tmthree \defeq \tmthreep \tmfive$ satisfies the statement, as $\tmthree\isub\var\gconst = (\tmthreep \tmfive)\isub\var\gconst = \tmthreep\isub\var\gconst \tmfive\isub\var\gconst = \tmtwo$.
      
      \item \emph{Inert Step}, identical to the abstraction subcase.
      \end{enumerate}

      \item \emph{Application Right}, \ie $\tm = \tmfour\tmfive$ and reduction takes place in $\tmfive$. Identical to the \emph{application left} case, just switch left and right.
      \qedhere
    \end{itemize}
      
\end{proof}

\setcounter{lemmaAppendix}{\value{l:inerts-and-creations}}
\begin{lemmaAppendix}[Inert Substitutions Can Be Avoided]
\label{lappendix:inerts-and-creations}
  Let 
  \NoteState{l:inerts-and-creations}
  $\tm, \tmtwo$ be terms and $\gconst$ be an inert term. Then, $\tm \tof \tmtwo$ iff $\tm \isub\var\gconst \tof \tmtwo \isub\var\gconst$.
\end{lemmaAppendix}

\begin{proof}
  The right-to-left direction is \reflemma{inert-anti-red-bis}, since $\tof \,=\, \tobabs\cup \toin$.
  The left-to-right direction is proved by induction on the definition of $\tm \tof \tmtwo$. Cases:
  \begin{itemize}
    \item \emph{Step at the root}:
    \begin{enumerate}
      \item \emph{Abstraction Step}, \ie $\tm = (\la\vartwo\tmthree)\tmfour \rtobabs \tmthree\isub\vartwo\tmfour = \tmtwo$ where $\tmfour$ is an abstraction.
      We can suppose without loss of generality that $\vartwo \in \fv{\gconst} \cup \{\var\}$.
      Note that $\tmfour\isub\var\gconst$ is an abstraction, according to \reflemmap{inert-anti-sub-bis}{abs}.
      Then, $\tm\isub\var\gconst = (\la\vartwo\tmthree\isub\var\gconst)\tmfour\isub\var\gconst \rtobabs\! \tmthree\isub\var\gconst\isub\vartwo{\tmfour\isub\var\gconst} = \tmthree\isub\vartwo\tmfour\isub\var\gconst = \tmtwo\isub\var\gconst$.
      \item \emph{Inert Step}, \ie $\tm = (\la\vartwo\tmthree)\gconsttwo \rtoin \tmthree\isub\vartwo\gconsttwo = \tmtwo$ where $\gconsttwo$ is an inert term.
      We can suppose without loss of generality that $\vartwo \in \fv{\gconst} \cup \{\var\}$.
      According to \reflemmap{inert-anti-sub-bis}{inert}, $\tmfour\isub\var\gconst$ is inert.
      Then, $\tm\isub\var\gconst = (\la\vartwo\tmthree\isub\var\gconst)\gconsttwo\isub\var\gconst \rtoin \tmthree\isub\var\gconst\isub\vartwo{\gconsttwo\isub\var\gconst} = \tmthree\isub\vartwo\gconsttwo\isub\var\gconst = \tmtwo\isub\var\gconst$.
    \end{enumerate}
    \item \emph{Application Left}, \ie $\tm = \tmthree\tmfour \tof \tmthreep\tmfour = \tmtwo$ with $\tmthree \tof \tmthreep$.
    By \ih, $\tmthree\isub\var\gconst \tof \tmthreep\isub\var\gconst$, hence $\tm\isub\var\gconst = \tmthree\isub\var\gconst\tmfour\isub\var\gconst \tof \tmthreep\isub\var\gconst\tmfour\isub\var\gconst = \tmtwo\isub\var\gconst$.
    \item \emph{Application Right}, \ie $\tm = \tmfour\tmthree \tof \tmfour\tmthreep = \tmtwo$ with $\tmthree \tof \tmthreep$.
    By \ih, $\tmthree\isub\var\gconst \tof \tmthreep\isub\var\gconst$, hence $\tm\isub\var\gconst = \tmfour\isub\var\gconst\tmthree\isub\var\gconst \tof \tmfour\isub\var\gconst\tmthreep\isub\var\gconst = \tmtwo\isub\var\gconst$.
    \qedhere
  \end{itemize}

\end{proof}

\subsection{Proofs of Section~\ref{sect:machines-intro} (Preliminaries on Abstract Machines, Implementations, and Complexity Analysis)}

% !TEX root = main.tex
Here we provide the abstract proof of \refthm{abs-impl}, stating that the conditions required to an implementation system (\refdef{implementation}) indeed imply that the machine implements the strategy via the decoding (in the sense of \refdef{implem}).

The \emph{executions-to-derivations} part of the implementation theorem is easy to prove, essentially \emph{$\beta$-projection} and \emph{overhead transparency} allow to project a single transition, and the projection of executions is obtained as a simple induction.

The \emph{derivations-to-executions} part is a bit more delicate, instead, because the simulation of $\beta$-steps has to be done \emph{up to} overhead transitions. The following lemma shows how the conditions for implementation systems allow to do that. 
Interestingly all five conditions of \refdef{implementation} are used in the proof.

\begin{lemma}[One-Step Simulation]
  \label{l:one-step-simulation}
  Let $\mach$, $\tostrat$, and $\decode\cdot$ be a machine, a strategy, and a decoding forming an %\emph
  {implementation system}. %, \ie such that:
% \begin{varenumerate}
% 		\item\label{p:beta-projection} \emph{$\beta$-Projection}: $\state \tomachhole\beta \statetwo$ implies $\decode\state \tostrat \decode\statetwo$;
% 		\item\label{p:overhead-transparency} \emph{Overhead Transparency}: $\state \tomacho \statetwo$ implies $\decode\state = \decode\statetwo$;
% 		\item\label{p:overhead-terminate} \emph{Overhead Transitions Terminate}:  $\tomacho$ terminates;
% 
% 	\item\label{p:determinism} \emph{Determinism}: both $\mach$ and $\tostrat$ are deterministic;
% 
% 	\item\label{p:progress} \emph{Progress}: $\mach$ final states decode to $\tostrat$-normal terms.
% \end{varenumerate}
% 
%   Let $\state$ be a state of $\mach$.
  For any state $\state$ of $\mach$, if $\decode\state \to \tmtwo$ then there is a state $\statetwo$ of $\mach$ such that $\state \tomacho^*\tomachb \statetwo$.
\end{lemma}

\begin{proof}
  According to \refdef{implementation}, since $(\mach,\to,\decode\cdot)$ is an implementation system, the following conditions hold:
  \begin{varenumerate}
		\item\label{p:beta-projection} \emph{$\beta$-Projection}: $\state \tomachhole\beta \statetwo$ implies $\decode\state \tostrat \decode\statetwo$;
		\item\label{p:overhead-transparency} \emph{Overhead Transparency}: $\state \tomacho \statetwo$ implies $\decode\state = \decode\statetwo$;
		\item\label{p:overhead-terminate} \emph{Overhead Transitions Terminate}:  $\tomacho$ terminates;

	\item\label{p:determinism} \emph{Determinism}: both $\mach$ and $\tostrat$ are deterministic;

	\item\label{p:progress} \emph{Progress}: $\mach$ final states decode to $\tostrat$-normal terms.
  \end{varenumerate}

  For any state $\state$ of $\mach$, let $\nfo{\state}$ be the state that is the normal form of $\state$ with respect to $\tomacho$: such a state exists and is unique because overhead transitions terminate (\refpoint{overhead-terminate}) and $\mach$ is deterministic (\refpoint{determinism}).
  Since $\tomacho$ is mapped on identities (\refpoint{overhead-transparency}), one has $\decode{\nfo{\state}} = \decode\state$.
  As $\decode\state$ is not $\to$-normal by hypothesis, the progress property (\refpoint{progress}) entails that $\nfo{\state}$ is not final, therefore $\state \tomacho^* \nfo{\state} \tomachb \statetwo$ for some state $\statetwo$, and thus $\decode\state = \decode{\nfo{\state}} \to \decode{\statetwo}$ by $\beta$-projection (\refpoint{beta-projection}).
  According to the determinism of $\to$ (\refpoint{determinism}), one obtains $\decode{\statetwo} = \tmtwo$.
\end{proof}

Now, the one-step simulation can be extended to the a simulation of derivations by an easy induction on the length of the derivation.

\setcounter{theoremAppendix}{\value{thm:abs-impl}}
\begin{theoremAppendix}[Sufficient Condition for Implementations]
\label{thmappendix:abs-impl}
Let 
\NoteState{thm:abs-impl}
% $\mach$, $\tostrat$, and $\decode\cdot$ be a machine, a strategy, and a decoding forming an \emph{implementation system}.
$(\mach, \tostrat, \decode\cdot)$ be an \emph{implementation system}. %, \ie such that:
% \begin{varenumerate}
% 	  \item \emph{$\beta$-Projection}: $\state \tomachhole\beta \statetwo$ implies $\decode\state \tostrat \decode\statetwo$;
% 	  \item \emph{Overhead Transparency}: $\state \tomacho \statetwo$ implies $\decode\state = \decode\statetwo$;
% 	  \item \emph{Overhead Transitions Terminate}:  $\tomacho$ terminates;
% 	\item \emph{Determinism}: both $\mach$ and $\tostrat$ are deterministic;
% 	\item \emph{Progress}: $\mach$ final states decode to $\tostrat$-normal terms.
% \end{varenumerate}
% 
 Then, $\mach$ implements $\tostrat$ via $\decode\cdot$.
\end{theoremAppendix}

\begin{proof}
  According to \refdef{implem}, given a $\lambda$-term $\tm$, we have to show that:
  \begin{enumerate}[(i)]
    \item\label{p:exec-to-deriv} \emph{Executions to Derivations with $\beta$-matching}: for any $\mach$-execution $\exec: \compil\tm \tomachine^* \state$ there exists a $\tostrat$-derivation $\deriv: \tm \tostrat^* \decode\state$ such that $\size\deriv = \sizebeta\exec$.

    \item\label{p:deriv-to-exec} \emph{Derivations to Executions with $\beta$-matching}: for every $\tostrat$-derivation $\deriv: \tm \tostrat^* \tmtwo$ there exists a $\mach$-execution $\exec: \compil\tm \tomachine^* \state$ such that $\decode\state = \tmtwo$ and $\size\deriv = \sizebeta\exec$.
  \end{enumerate}

  \paragraph{Proof of \refeqpoint{exec-to-deriv}:}  by induction on $\sizebeta\exec \in \nat$.
  
  If $\sizebeta\exec = 0$ then $\exec \colon \compil\tm \tomacho^* \state$ and hence $\decode{\compil\tm} = \decode\state$ by overhead transparency (\refpoint{def-overhead-transparency} of \refdef{implementation}).
  Moreover, $\tm = \decode{\compil\tm}$ since decoding is the inverse of compilation on initial states, therefore we are done by taking the empty (\ie without steps) derivation $\deriv$ with starting (and end) term $\tm$.
  
  Suppose $\sizebeta\exec > 0$: then, $\exec \colon \compil{\tm} \tomachine^* \state$ is the concatenation of an execution $\execp \colon \compil{\tm} \tomachine^* \statetwo$ followed by an execution $\execpp \colon \statetwo \tomachb \statethree \tomacho^* \state$.
  By \ih applied to $\execp$, there exists a derivation $\derivp \colon \tm \to^* \decode\statetwo$ with $\sizebeta\execp = \size\derivp$.
  By $\beta$-projection (\refpoint{def-beta-projection} of \refdef{implementation}) and overhead transparency (\refpoint{def-overhead-transparency} of \refdef{implementation}) applied to $\execpp$, one has $\derivpp \colon \decode\statetwo \to \decode\statethree = \decode\state$.
  Therefore, the derivation %$\deriv \colon \tm \to^* \decode\statetwo \to \decode\state$
  $\deriv$ defined as the concatenation of $\derivp$ and $\derivpp$ is such that $\deriv \colon \tm \to^* \decode\state$ and $\size\deriv = \size\derivp + \size\derivpp = \sizebeta\execp + 1 = \sizebeta\exec$.
 
  \paragraph{Proof of \refeqpoint{deriv-to-exec}:}  by induction on $\size\deriv \in \nat$.

  If $\size\deriv = 0$ then $\tm = \tmtwo$.
  Since decoding is the inverse of compilation on initial states, one has $\decode{\compil{\tm}} = \tm$.
  We are done by taking the empty (\ie without transitions) execution $\exec$ with initial (and final) state $\compil{\tm}$.
  
  Suppose $\size\deriv > 0$: so, $\deriv\colon \tm \to^* \tmtwo$ is the concatenation of a derivation $\derivp \colon \tm \to^* \tmtwop$ followed by the step $\tmtwop \to \tmtwo$.
  By \ih, there exists a $\mach$-execution $\execp\colon \compil\tm \tomachine^* \statetwo$ such that $\decode\statetwo = \tmtwop$ and $\size\derivp = \sizebeta\execp$.
  According to the one-step simulation (\reflemma{one-step-simulation}, since $\decode{\statetwo} \to \tmtwo$ and $(\mach,\to,\decode\cdot)$ is an implementation system), there is a state $\state$ of $\mach$ such that $\statetwo \tomacho^*\tomachb \state$ and $\decode\state = \tmtwo$.
  Therefore, the execution $\exec \colon \compil\tm \tomachine^*\statetwo \tomacho^*\tomachb \state$ is such that $\sizebeta\exec = \sizebeta\execp +1 = \size\derivp + 1 = \size\deriv$.
\end{proof}

\subsection{Proofs of Section~\ref{sect:eglamour} (\eglamour)}

% !TEX root = main.tex
First we prove the invariants of the \eglamour, and then we use them to prove that it forms an implementation system with respect to right-to-left evaluation $\torf$ in the fireball calculus (via the decoding).

\setcounter{lemmaAppendix}{\value{l:eglamour-invariants}}
\begin{lemmaAppendix}[\eglamour Invariants]
	\label{lappendix:eglamour-invariants} % \reflemmap{eglamour-invariants}{subterm}
	Let
\NoteState{l:eglamour-invariants}
	$\state = (\dump,\code,\stack,\genv)$ be a reachable state. Then:
	\begin{varenumerate}
		\item \emph{Name:} \label{pappendix:eglamour-invariants-name}
		\begin{varenumerate}
			\item \emph{Explicit Substitution}: \label{pappendix:eglamour-invariants-name-es}
			if $\genv = \genvtwo  \esub\var\codetwo \genvthree$  then $\var$ is fresh wrt $\codetwo$ and $\genvthree$;

			\item \emph{Abstraction}: \label{pappendix:eglamour-invariants-name-abs}
			if $\la\var\codetwo$ is a subterm of $\dump$, $\code$, $\stack$, or $\genv$ then $\var$ may occur only in $\codetwo$;
		\end{varenumerate}
		
		\item \label{pappendix:eglamour-invariants-fireball-stack}\emph{Fireball Item}: $\decode\stackitem$ and $\relunf{\decode\stackitem}\genv$ are inert terms if $\stackitem = \pair\var\stacktwo$ and abstractions otherwise, for every item $\stackitem$ in $\stack$, in $\genv$, and in every stack in $\dump$;

		\item \label{pappendix:eglamour-invariants-ev-ctx}\emph{Contextual Decoding}: $\stctx\state = \relunf{\decdumpp\decstack}\genv$ is a right context.

		%\item \label{p:eglamour-invariants-env-size}\emph{Environment Size}: the length of the global environment $\genv$ is bound by $\sizem\exec$. Moreover, the number of entries in $\genv$ that contain an abstraction is exactlty $\sizemv\exec$.
		
	\end{varenumerate}
\end{lemmaAppendix}

\begin{proof}
By induction on the length of the execution leading to the reachable state. In an initial state all the invariants trivially hold. For a non-empty execution the proof for every invariant is by case analysis on the last transition, using the \ih. 
\begin{enumerate}
 \item \emph{Name.} Cases:
\begin{enumerate}
 \item $\statetwo = \glamst\dump {\code\codetwo} \stack \genv 
	\tomachcone
	\glamst { \dump\cons\dentry\code\stack } \codetwo \stempty \genv = \state$. Both points follow immediately from the \ih

  \item $\statetwo = 
	  \glamst{ \dump\cons\dentry\code\stack }{ \la\var\codetwo } \stempty  \genv
	  \tomachctwo 
	  \glamst \dump \code {\pair{\la\var\codetwo}{\stempty} \cons\stack} \genv	
	 = \state$. Both points follow immediately from the \ih

%   \item $\statetwo = 
% 	  \glamst {\dump\cons\dentry\code\stack} \var \stacktwo {\genv_1\esub\var{\pair{\vartwo}{\stackthree}}\genv_2}
% 	  \tomachcthree
% 	  \glamst \dump \code {\pair{\var}{\stacktwo}\cons\stack} {\genv_1\esub\var{\pair{\vartwo}{\stackthree}}\genv_2}
% 	 = \state$. Both points follow immediately from the \ih
% 
%   \item $\statetwo = 
% 	  \glamst {\dump\cons\dentry\code\stack} \var \stacktwo \genv
% 	  \tomachcfour
% 	  \glamst \dump \code {\pair{\var}{\stacktwo}\cons\stack} \genv 
% 	 = \state$ with $\genv(\var)=\bot$. Both points follow immediately from the \ih

  \item $\statetwo = 
	  \glamst {\dump\cons\dentry\code\stack} \var \stacktwo \genv
	  \tomachcthree
	  \glamst \dump \code {\pair{\var}{\stacktwo}\cons\stack} \genv 
	 = \state$ with $\genv(\var)=\bot$ or $\genv(\var) = \pair{\vartwo}{\stackthree}$. Both points follow immediately from the \ih

  \item $\statetwo = 
	  \glamst \dump {\la\var\code} {\stackitem\cons\stack} \genv
	  \tomachsm
	  \glamst \dump \code \stack {\esub\var\stackitem\genv}
	 = \state$. \refpoint{eglamour-invariants-name-es} for the new entry in the environment follows from the \ih for \refpoint{eglamour-invariants-name-abs}, for the other entries from the \ih for \refpoint{eglamour-invariants-name-es}. \refpoint{eglamour-invariants-name-abs} follows from its \ih
	  
  \item \begin{equation*}
          \begin{split}
	 \statetwo &= 
	  \glamst \dump \var \stack {\genv_1\esub\var{\pair{\la\vartwo\codetwo}\stempty}\genv_2} \\
	  &\tomachse
	  \glamst \dump {\rename{\la\vartwo\codetwo}} \stack {\genv_1\esub\var{\pair{\la\vartwo\codetwo}\stempty}\genv_2}
	   = \state.
	 \end{split}
       \end{equation*} \refpoint{eglamour-invariants-name-es} follows from its \ih. \refpoint{eglamour-invariants-name-abs} for the new code is guaranteed by the $\alpha$-renaming operation $\rename{\la\vartwo\codetwo}$, the rest follows from its \ih
\end{enumerate} 
 
 \item \emph{Fireball Item}. Cases:
\begin{enumerate}
 \item $\statetwo = \glamst\dump {\code\codetwo} \stack \genv 
	\tomachcone
	\glamst { \dump\cons\dentry\code\stack } \codetwo \stempty \genv = \state$. It follows from the \ih

  \item $\statetwo = 
	  \glamst{ \dump\cons\dentry\code\stack }{ \la\var\codetwo } \stempty  \genv
	  \tomachctwo 
	  \glamst \dump \code {\pair{\la\var\codetwo}{\stempty} \cons\stack} \genv	
	 = \state$.  For $\pair{\la\var\codetwo}{\stempty}$ we have that $\decode{\pair{\la\var\codetwo}{\stempty}}$ and $\relunf{\decode{\pair{\la\var\codetwo}{\stempty}}}\genv = \relunf{(\la\var\codetwo)}\genv = \la\var\relunf{\codetwo}\genv$ are abstractions, and hence fireballs.	 
	 For all other items the invariant follows from the \ih

  \item $\statetwo = 
	  \glamst {\dump\cons\dentry\code\stack} \var \stacktwo {\genv}
	  \tomachcthree
	  \glamst \dump \code {\pair{\var}{\stacktwo}\cons\stack} {\genv}
	 \allowbreak= \state$ with $\genv(\var) = \bot$ or $\genv(\var) = {\pair\vartwo\stackthree}$. 
	 For $\pair\var\stacktwo$, we have that $\decode{\pair\var\stacktwo} = \ctxholep{\var}\decode\stacktwo$ and $\relunf{\decode{\pair\var\stacktwo}}\genv = \ctxholep{\relunf{\var}\genv}(\relunf{\decode\stacktwo}\genv)$.
	 By \ih, $\decode\stackitemtwo$ is a fireball for every item $\stackitemtwo$ in $\stacktwo$. 
	 Therefore, $\decode{\pair\var\stacktwo}$ is an inert term.
	 Concerning $\relunf{\decode{\pair\var\stacktwo}}\genv$, there are two subcases:
	 \begin{enumerate}
	   \item $\genv(\var) = \pair{\vartwo}{\stackthree}$ \ie $\genv \defeq \genv_1\esub\var{\pair\vartwo\stackthree}\genv_2$. 
	   By \reflemmapp{eglamour-invariants}{name}{es}, every ES in $\genv$ binds a different variable, so $\relunf{\var}\genv = \relunf{\var}{\genv_1\esub\var{\pair\vartwo\stackthree}\genv_2} = \relunf{\var}{\genv_1}\relunf{\isub\var{\decode{\pair{\vartwo}{\stackthree}}}}{\genv_2} = \relunf{\decode{\pair{\vartwo}{\stackthree}}}{\genv_2} = \relunf{\decode{\pair{\vartwo}{\stackthree}}}{\genv}$, that by \ih is an inert term. 
	 Moreover, the \ih also gives that $\relunf{\stackitemtwo}\genv$ is a fireball for every item $\stackitemtwo$ in $\stacktwo$. 
	 Therefore  $\relunf{\decode{\pair\var\stacktwo}}\genv = \ctxholep{\relunf{\var}\genv}(\relunf{\decode\stacktwo}\genv)$ is an inert term. 

	   \item $\genv(\var)=\bot$. Similar to the previous case. 
	 By hypothesis, we have $\relunf{\var}\genv = \var$. 
	 As before, by \ih $\relunf{\decode\stackitemtwo}\genv$ is a  fireball for every item $\stackitemtwo$ in $\stacktwo$. 
	 So,  $\relunf{\decode{\pair\var\stacktwo}}\genv = \ctxholep{\relunf{\var}\genv}(\relunf{\decode\stacktwo}\genv)  = \ctxholep{\var}(\relunf{\decode\stacktwo}\genv)$ is an inert term. 
	 \end{enumerate}
	 For all other items in $\state$ the invariant follows from the \ih

  \item $\statetwo = 
	  \glamst \dump {\la\var\code} {\stackitem\cons\stack} \genv
	  \tomachsm
	  \glamst \dump \code \stack {\esub\var\stackitem\genv}
	 = \state$. By \reflemmapp{eglamour-invariants}{name}{abs} $\var$ may occur only in $\code$. 
	 Thus the substitution $\relunf{}{\esub\var\stackitem\genv}$ acts exactly as $\relunf{}\genv$ on every item in $\state$. Then the invariant follows from the \ih
	  
  \item \begin{equation*}
	 \begin{split}
	 \statetwo &= 
	  \glamst \dump \var \stack {\genv_1\esub\var{\pair{\la\vartwo\codetwo}\stempty}\genv_2} \\
	  &\tomachse
	  \glamst \dump {\rename{\la\vartwo\codetwo}} \stack {\genv_1\esub\var{\pair{\la\vartwo\codetwo}\stempty}\genv_2}
	 = \state.
         \end{split}
       \end{equation*}
 It follows from the \ih

\end{enumerate} 

\item \emph{Contextual Decoding}. Cases:
\begin{enumerate}
 \item $\statetwo = \glamst\dump {\code\codetwo} \stack \genv 
	\tomachcone
	\glamst { \dump\cons\dentry\code\stack } \codetwo \stempty \genv = \state$. By \ih $\stctx\statetwo = \relunf{\decdumpp\decstack}\genv$ is a right context, as well as $\relunf{\codetwo}\genv \ctxhole$. Then their composition $(\relunf{\decdumpp\decstack}\genv)\ctxholep{\relunf{\codetwo}\genv \ctxhole} \allowbreak= \relunf{\decdumpp{\decstackp{\codetwo \ctxhole}}}\genv = \stctx\state$ is a right context.

  \item $\statetwo = 
	  \glamst{ \dump\cons\dentry\code\stack }{ \la\var\codetwo } \stempty  \genv
	  \tomachctwo 
	  \glamst \dump \code {\pair{\la\var\codetwo}\stempty \cons\stack} \genv	
	 \allowbreak= \state$. 
	 By \ih $\stctx\statetwo = \relunf{\decode{\dump\cons\dentry\code\stack}}\genv = \relunf{\decdumpp{\decodeinv\stack{\code\ctxhole}}}\genv$ is a right context, that implies that $\relunf{\decdumpp{\decstack}}\genv$ is one such context as well. So, $\stctx\statetwo \!= \relunf{\decdumpp{\decode{\pair{\la\var\codetwo}\stempty \cons\stack}}}\genv = \relunf{\decdumpp{\decodeinv\stack{\ctxhole\la\var\codetwo}}}\genv \allowbreak= (\relunf{\decdumpp{\decstack}}\genv)\ctxholep{\ctxhole\relunf{\la\var\codetwo}\genv}$ is a right context, because it is the composition of right context, given that $\relunf{\la\var\codetwo}\genv$ is a fireball.

%   \item $\statetwo = 
% 	  \glamst {\dump\cons\dentry\code\stack} \var \stacktwo {\genv_1\esub\var{\pair{\vartwo}{\stackthree}}\genv_2}
% 	  \tomachcthree
% 	  \glamst \dump \code {\pair\var\stacktwo \!\cons\stack} {\genv_1\esub\var{\pair{\vartwo}{\stackthree}}\genv_2}
% 	 \allowbreak= \state$. Let $\genv \defeq \genv_1\esub\var{\pair{\vartwo}{\stackthree}}\genv_2$.
% 	 By \ih $\stctx\statetwo = \relunf{\decodep{\dump\cons\dentry\code\stack}{\decode\stacktwo}}\genv = \relunf{\decdumpp{\decodeinv\stack{\code(\decode\stacktwo)}}}\genv$ is a right context, that implies that $\relunf{\decdumpp{\decstack}}\genv$ is one such context as well. Then 
% 	 $\stctx\state = \relunf{\decdumpp{\decode{\pair\var\stacktwo\cons\stack}}}\genv = 
% 	 \relunf{\decdumpp{\decodeinv\stack{\ctxhole\decode{\pair\var\stacktwo}}}}\genv = 
% 	 (\relunf{\decdumpp{\decstack}}\genv)\ctxholep{\relunf{\ctxhole\decode{\pair\var\stacktwo}}\genv}$ is a right context, because it is the composition of right context, given that $\decode{\pair\var\stacktwo}\indsub\genv$ is a fireball by \refpoint{eglamour-invariants-fireball-stack}.
% 
%   \item $\statetwo = 
% 	  \glamst {\dump\cons\dentry\code\stack} \var \stacktwo \genv
% 	  \tomachcfour
% 	  \glamst \dump \code {\pair{\var}{\stacktwo}\cons\stack} \genv 
% 	 = \state$ with $\genv(\var)=\bot$. Exactly as the previous case.
  \item $\statetwo = 
	  \glamst {\dump\cons\dentry\code\stack} \var \stacktwo {\genv}
	  \tomachcthree
	  \glamst \dump \code {\pair\var\stacktwo \!\cons\stack} {\genv}
	 \allowbreak= \state$ with $\genv(\var)=\bot$ or with $\genv(\var)= \pair{\vartwo}{\stackthree}$. 
	 By \ih $\stctx\statetwo = \relunf{\decodep{\dump\cons\dentry\code\stack}{\decode\stacktwo}}\genv = \relunf{\decdumpp{\decodeinv\stack{\code(\decode\stacktwo)}}}\genv$ is a right context, that implies that $\relunf{\decdumpp{\decstack}}\genv$ is one such context as well. Then 
	 $\stctx\state = \relunf{\decdumpp{\decode{\pair\var\stacktwo\cons\stack}}}\genv = 
	 \relunf{\decdumpp{\decodeinv\stack{\ctxhole\decode{\pair\var\stacktwo}}}}\genv = 
	 (\relunf{\decdumpp{\decstack}}\genv)\ctxholep{\relunf{\ctxhole\decode{\pair\var\stacktwo}}\genv}$ is a right context, because it is the composition of right context, given that $\relunf{\decode{\pair\var\stacktwo}}\genv$ is a fireball by \reflemmap{eglamour-invariants}{fireball-stack}.

  \item $\statetwo = 
	  \glamst \dump {\la\var\code} {\stackitem\cons\stack} \genv
	  \tomachsm
	  \glamst \dump \code \stack {\esub\var\stackitem\genv}
	 = \state$. 
	  By \ih $\stctx\statetwo = \relunf{\decdumpp{\decode{\stackitem\cons\stack}}}\genv$ is a right context, that implies that $\relunf{\decdumpp{\decstack}}\genv$ is one such context as well. 
	  Now, note that $\stctx\state = \relunf{\decdumpp{\decstack}}{\esub\var\stackitem\genv} = \relunf{\decdumpp{\decstack}}\genv$ because by \reflemmapp{eglamour-invariants}{name}{abs} $\var$ may occur only in $\code$, and so the substitution $\relunf{}{\esub\var\stackitem\genv}$ acts on every code in $\dump$ and $\stack$ exactly as $\relunf{}\genv$. 
	  
  \item \begin{equation*}
          \begin{split}
	 \statetwo &= 
	  \glamst \dump \var \stack {\genv_1\esub\var{\pair{\la\vartwo\codetwo}\stempty}\genv_2} \\
	  &\tomachse 
	  \glamst \dump {\rename{\la\vartwo\codetwo}} \stack {\genv_1\esub\var{\pair{\la\vartwo\codetwo}\stempty}\genv_2}
	 = \state.
	  \end{split}
        \end{equation*} 
	 It follows by the \ih because $\stctx\statetwo = \stctx\state$, as the only component that changes is the code.
  \qedhere
\end{enumerate}
\end{enumerate}
\end{proof}

\begin{note}\label{note:union-transitions}
  Given a machine $\mach$, a transition is a binary relation on the set of states of $\mach$. 
  Given two transitions $\tomachhole{\rsym_1}$ and $\tomachhole{\rsym_2}$, we set $\tomachhole{\rsym_1,\rsym_2} \ \defeq \ \tomachhole{\rsym_1}  \!\cup \tomachhole{\rsym_2}$ (also denoted by $\tomachhole{\rsym_{1,2}}$ or simply $\tomachhole{\rsym}$).
\end{note}

\paragraph{Conditions for an Implementation System.} We now prove that the \eglamour satisfies the conditions for an implementation system with respect to $\torf$. First, we deal with the two conditions about the projection of transitions on the calculus.
\begin{lemma}[\eglamour $\beta$-Projection and Overhead Transparency]
  \label{l:eglamour-trans-projection} % \reflemmap{eglamour-trans-projection}{exp-com}
	Let %
% \NoteState{l:eglamour-trans-projection}
	$\state$ be a reachable state.
	\begin{varenumerate}
		\item \label{p:eglamour-trans-projection-exp-com} \emph{Overhead Transparency}: 
		if $\state\tomachhole{\ssym,\csym_{1,2,3}}\statetwo$ (see Note \ref{note:union-transitions} for the meaning of $\tomachhole{\ssym,\csym_{1,2,3}}$) then $\decode\state=\decode\statetwo$;
		\item \label{p:eglamour-trans-projection-mult}\emph{$\beta$-Projection}: if $\state\tomachsm\statetwo$ then $\decode\state\torf\decode\statetwo$.	  
		
	\end{varenumerate}
\end{lemma}

\begin{proof}
Transitions:
\begin{enumerate}
 \item $\state = \glamst\dump {\code\codetwo} \stack \genv 
	\tomachcone
	\glamst { \dump\cons\dentry\code\stack } \codetwo \stempty \genv = \statetwo$.
	Then
	\[\begin{array}{rclclcl}
	   \decode\state 
	   & = & \relunf{\decodep{\dump}{\decodeinv\stack{\code\codetwo}}}\genv\\
	   & = & \relunf{\decodep{\dump\cons\dentry\code\stack}{\codetwo}}\genv \\
 	   & = & \relunf{\decodep{\dump\cons\dentry\code\stack}{\decodeinv\stempty\codetwo}}\genv 
 	   & = & \decode\statetwo
	  \end{array}\]

  \item $\state = 
	  \glamst{ \dump\cons\dentry\code\stack }{ \la\var\codetwo } \stempty  \genv
	  \tomachctwo 
	  \glamst \dump \code {\pair{\la\var\codetwo}{\stempty} \cons\stack} \genv	
	 = \statetwo$.
	 Then
	\[\begin{array}{rclclcl}
	   \decode\state 
	   & = & \relunf{\decodep{\dump\cons\dentry\code\stack}{\decodeinv\stempty{\la\var\codetwo}}}\genv\\	   
	   & = & \relunf{\decodep{\dump}{\decodeinv{\stack}{\code(\decodeinv\stempty{\la\var\codetwo})}}}\genv \\
	   & = & \relunf{\decodep{\dump}{\decodeinv{\pair{\la\var\codetwo}{\stempty}\cons\stack}\code}}\genv 
 	   & = & \decode\statetwo
	  \end{array}\]

%   \item $\state = 
% 	  \glamst {\dump\cons\dentry\code\stack} \var \stacktwo {\genv_1\esub\var{\pair{\vartwo}{\stackthree}}\genv_2}
% 	  \tomachcthree
% 	  \glamst \dump \code {\pair{\var}{\stacktwo}\cons\stack} {\genv_1\esub\var{\pair{\vartwo}{\stackthree}}\genv_2}
% 	 = \statetwo$. Let $\genvtwo \defeq \genv_1\esub\var{\pair{\vartwo}{\stackthree}}\genv_2$
% 	  Then
% 	\[\begin{array}{rclclcl}
% 	   \decode\state 
% 	   & = & \decodep{\dump\cons\dentry\code\stack}{\decodeinv\stacktwo\var}\indsub\genvtwo\\
% 	   & = & \decodep{\dump}{\decodeinv\stack{\code(\decodeinv\stacktwo\var)}}\indsub\genvtwo\\
% 	   & = & \decodep{\dump}{\decodeinv{\pair\var\stacktwo \cons\stack}\code}\indsub\genvtwo
%  	   & = & \decode\statetwo
% 	  \end{array}\]
% 
%   \item $\state = 
% 	  \glamst {\dump\cons\dentry\code\stack} \var \stacktwo \genv
% 	  \tomachcfour
% 	  \glamst \dump \code {\pair{\var}{\stacktwo}\cons\stack} \genv 
% 	 = \statetwo$ with $\genv(\var)=\bot$.
% 	  Then
% 	\[\begin{array}{rclclcl}
% 	   \decode\state 
% 	   & = & \decodep{\dump\cons\dentry\code\stack}{\decodeinv\stacktwo\var}\indsub{\genv}\\
% 	   & = & \decodep{\dump}{\decodeinv\stack{\code(\decodeinv\stacktwo\var)}}\indsub{\genv}\\
% 	   & = & \decodep{\dump}{\decodeinv{\pair\var\stacktwo \cons\stack}\code}\indsub{\genv}
%  	   & = & \decode\statetwo
% 	  \end{array}\]
  \item $\state = 
	  \glamst {\dump\cons\dentry\code\stack} \var \stacktwo {\genv}
	  \tomachcthree
	  \glamst \dump \code {\pair{\var}{\stacktwo}\cons\stack} {\genv}
	 = \statetwo$ with $\genv(\var) = \bot$ or $\genv(\var) = {\pair{\vartwo}{\stackthree}}$.
	  Then
	\[\begin{array}{rclclcl}
	   \decode\state 
	   & = & \relunf{\decodep{\dump\cons\dentry\code\stack}{\decodeinv\stacktwo\var}}\genv\\
	   & = & \relunf{\decodep{\dump}{\decodeinv\stack{\code(\decodeinv\stacktwo\var)}}}\genv\\
	   & = & \relunf{\decodep{\dump}{\decodeinv{\pair\var\stacktwo \cons\stack}\code}}\genv
 	   & = & \decode\statetwo
	  \end{array}\]

  \item $\state = 
	  \glamst \dump {\la\var\code} {\stackitem\cons\stack} \genv
	  \tomachsm
	  \glamst \dump \code \stack {\esub\var\stackitem\genv}
	 = \statetwo$.
	  Then
	\[\begin{array}{rclclcl}
	   \decode\state 
	   & = & \relunf{\decodep{\dump}{\decodeinv{\stackitem\cons\stack}{\la\var\code}}}{\genv}\\
	   & = & \relunf{\decodep{\dump}{\decodeinv{\stack}{(\la\var\code)\decode\stackitem}}}{\genv}\\
	   & \torf & \relunf{\decodep{\dump}{\decodeinv{\stack}{\code \isub\var{\decode\stackitem}}}}{\genv}\\
	   & = & \relunf{\decodep{\dump}{\decodeinv{\stack}\code}\isub\var{\decode\stackitem}}{\genv}\\
	   & = & \relunf{\decodep{\dump}{\decodeinv{\stack}\code}}{\esub\var\stackitem\genv}
 	   & = & \decode\statetwo
	  \end{array}\]
	  where the rewriting step takes place because
	  \begin{enumerate}
	   \item $\relunf{\decodep{\dump}{\decode\stack}}{\genv}$ is a right context by \reflemmap{eglamour-invariants}{ev-ctx};
	   \item $\decode\stackitem$ is a fireball by \reflemmap{eglamour-invariants}{fireball-stack}.
	  \end{enumerate}
	  Moreover, the meta-level substitution $\isub\var{\decode\stackitem}$ can be extruded (in the equality step after the rewriting) without renaming $\var$, because by \reflemmapp{eglamour-invariants}{name}{abs} $\var$ does not occur in $\dump$ nor $\stack$.
	  
  \item \begin{equation*}
          \begin{split}
            \state &= 
	  \glamst \dump \var \stack {\genv_1\esub\var{\pair{\la\vartwo\codetwo}\stempty}\genv_2} \\
	  &\tomachse
	  \glamst \dump {\rename{\la\vartwo\codetwo}} \stack {\genv_1\esub\var{\pair{\la\vartwo\codetwo}\stempty}\genv_2}
	 = \statetwo.
          \end{split}
        \end{equation*}
         Let $\genvtwo \defeq \genv_1\esub\var{\pair{\la\vartwo\codetwo}\stempty}\genv_2$.
	   Then
	\[\begin{array}{rclclcl}
	   \decode\state 
	   & = & \relunf{\decodep{\dump}{\decodeinv\stack\var}}{\genvtwo}\\
	   & = & \decodep{\relunf{\dump}{\genvtwo}}{\decodeinv{\relunf{\stack}{\genvtwo}}{\relunf{\var}{\genvtwo}}}\\
	   & = & \decodep{\relunf{\dump}{\genvtwo}}{\decodeinv{\relunf{\stack}{\genvtwo}}{\relunf{{\la\vartwo\codetwo}}{\genvtwo}}}\\
	   & = & \relunf{\decodep{\dump}{\decodeinv\stack{\la\vartwo\codetwo}}}{\genvtwo}
 	   & = & \decode\statetwo
	  \end{array}\]
\qedhere
\end{enumerate} 
\end{proof}

We also need a lemma for the progress condition.

\begin{lemma}[\eglamour Progress]
\label{l:eglamour-progress}
  Let %
% \NoteState{l:eglamour-progress}
  $\state$ be a reachable final state.
  Then $\decode\state$ is fireball, \ie it is $\betaf$-normal. 
\end{lemma}

\begin{proof}
 An immediate inspection of the transitions shows that in a final state the code cannot be an application and the dump is necessarily empty. In fact, final states have one of the following two shapes:
 \begin{varenumerate}
  \item \emph{Top-Level Unapplied Abstraction}, \ie $\state = \glamst\stempty{\la\var\code}\stempty\genv$. Then $\decode\state = \relunf{(\la\var\code)}\genv = \la\var\relunf{\code}\genv$ that is a fireball.
  
  \item \emph{Top-Level Free Variable or Inert Term with Free Head}, \ie $\state = \glamst\stempty\var\stack\genv$ with $\genv(\var) = \undef$. Then $\decode\state = \relunf{(\decodeinv\stack\var)}\genv = \ctxholep{\relunf{\var}\genv}(\relunf{\decode\stack}\genv) = \ctxholep\var(\relunf{\decode\stack}\genv)$. Now, by the fireball item invariant (\reflemmap{eglamour-invariants}{fireball-stack}) every element of $\relunf{\decode\stack}\genv$ is a fireball, and so $\ctxholep\var(\relunf{\decode\stack}\genv)$ is an inert term, \ie a fireball.\qedhere
  \end{varenumerate}
\end{proof}

% % \setcounter{theoremAppendix}{\value{thm:weak-bis}}
% \begin{theorem}[Weak Bisimulation]
% \label{thm:weak-bis} % \refthm{weak-bis}{rev-sim}
%   Let% 
% % \NoteState{thm:weak-bis}
%   $\state$ be an initial state of code $\code$. 
% %   Then
%  \begin{enumerate}
%   \item \emph{Simulation:} For every execution $\exec \colon \state\tomach^*\statetwo$ there exists a derivation $\deriv \colon \decode\state\torf^*\decode\statetwo$ such that $\sizef{\deriv} = \sizem{\exec}$;
%   
%   \item \label{p:weak-bis-rev-sim} \emph{Reverse Simulation:} For every derivation $\deriv \colon \code\torf^*\tmtwo$ there is an execution $\exec:\state\tomach^*\statetwo$ such that $\decode\statetwo = \tmtwo$ and $\sizef{\deriv} = \sizem\exec$.
%   
%  \end{enumerate}
% \end{theorem}

Finally, we obtain the implementation theorem.

\setcounter{theoremAppendix}{\value{thm:eglamour-implementation}}
\begin{theoremAppendix}[\eglamour Implementation]
  \label{thmappendix:eglamour-implementation}
  The 
  \NoteState{thm:eglamour-implementation}
  \eglamour implements right-to-left evaluation $\torf$ in $\firecalc$ (via the decoding $\decode\cdot$).
\end{theoremAppendix}

\begin{proof}
  According to \refthm{abs-impl}, it is enough to show that the \eglamour and the right-to-left evaluation $\torf$ and the decoding $\decode\cdot$ form an implementation system, \ie that the %following 
  five conditions in \refdef{implementation} hold%:
  .
%   \begin{varenumerate}
% 	\item \emph{$\beta$-Projection}: $\state \tomachb \statetwo$ implies $\decode\state \tostrat \decode\statetwo$;
% 	\item \emph{Overhead Transparency}: $\state \tomachhole{\ssym, \csym_{1,2,3}} \statetwo$ implies $\decode\state = \decode\statetwo$;
% 	\item 	\emph{Overhead Transitions Terminate}:  $\tomachhole{\ssym, \csym_{1,2,3}}$ terminates;
% 	\item \emph{Determinism}: both \eglamour and $\torf$ are deterministic;
% 	\item \emph{Progress}: \eglamour final states decode to $\torf$-normal terms.
%   \end{varenumerate}
% 
  Note that substitution ($\tomachse$) and commutative ($\tomachhole{\csym_{1,2,3}}$) transitions are considered as overhead transitions.  
%   We shall prove each point above separately.
  \begin{varenumerate}
	\item \emph{$\beta$-Projection}: $\state \tomachb \statetwo$ implies $\decode\state \tostrat \decode\statetwo$ by \reflemmap{eglamour-trans-projection}{mult}.
	\item \emph{Overhead Transparency}: $\state \tomachhole{\ssym, \csym_{1,2,3}} \statetwo$ implies $\decode\state = \decode\statetwo$ by \reflemmap{eglamour-trans-projection}{exp-com} (recall that $\tomachhole{\ssym, \csym_{1,2,3}} \ = \ \tomachhole{\ssym} \cup \tomachhole{\csym_1} \cup \tomachhole{\csym_2} \cup \tomachhole{\csym_3}$ according to Note~\ref{note:union-transitions}).
	\item 	\emph{Overhead Transitions Terminate}: Termination of  $\tomachhole{\ssym, \csym_{1,2,3}}$ is given by forthcoming \reflemma{comm-bound} and \refcoro{exp-bilinear}, which are postponed because they actually give precise complexity bounds, not just termination.
	\item \emph{Determinism}: The \eglamour machine is deterministic, as it can be seen by an easy inspection of the transitions (see \reffig{eglamour}).
% 	(recall that a code is a term and then it is either a variable, or an abstraction, or an application). 
	\reflemmap{prop-of-torf}{determ} proves that $\torf$ is deterministic.
	\item \emph{Progress}: Let $\state$ be an \eglamour final state. By \reflemma{eglamour-progress}, $\decode\state$ is a $\betaf$-normal term, in particular it is $\torf$-normal because $\torf \, \subseteq \, \tof$.
	\qedhere
  \end{varenumerate}

\end{proof}

\subsection{Proofs of Section~\ref{sect:eg-compl-anal} (Complexity Analysis of the \eglamour)}

\setcounter{lemmaAppendix}{\value{l:subterm-invariant}}
\begin{lemmaAppendix}[Subterm Invariant]
\label{lappendix:subterm-invariant} % \reflemmap{eglamour-invariants}{subterm}
	Let
\NoteState{l:subterm-invariant}
	$\exec: \compil{\tm_0} \tomach^* (\dump,\code,\stack,\genv)$ be a \eglamour execution. If $\la\var\codetwo$ is a subterm of $\dump$, $\code$, $\stack$, or $\genv$ then it is a subterm of $\tm_0$.
\end{lemmaAppendix}

\begin{proof}
First of all, let us be precise about \emph{subterms}: for us, $\codetwo$ is a subterm of $\tm_0$ if it does so up to variable names, both free and bound (and so the distinction between terms and codes is irrelevant). More precisely: define $\tm^-$ as $\tm$ in which all variables (including those appearing in binders) are replaced by a fixed symbol $\ast$. Then, we will consider $\tmtwo$ to be a subterm of $\tm$ whenever $\tmtwo^-$ is a subterm of $\tm^-$ in the usual sense. The key property ensured by this definition is that the size $\size\codetwo$ of $\codetwo$ is bounded by $\size\code$.

Now, the proof is by induction on the length of the execution leading to the reachable state. In an initial state the invariant trivially holds. For a non-empty execution the proof is by a straightforward case analysis on the last transition, always relying on the \ih\qedhere
\end{proof}

\setcounter{lemmaAppendix}{\value{l:free-occ-inv}}
\begin{lemmaAppendix}[Free Occurrences Invariant]
\label{lappendix:free-occ-inv}
  Let%
\NoteState{l:free-occ-inv} 
  $\exec: \compil{\tm_0} \tomach^* \state$ be a \eglamour execution.
  Then $\sizefree\state \leq \sizefree{\tm_0} + \size{\tm_0} \cdot \sizem\exec - \sizee \exec$.
\end{lemmaAppendix}

\begin{proof}
  By induction on $\size\exec$. Case $\size\exec=0$ is obvious, since $\compil{\tm_0} = \state$. Otherwise
  $\exectwo:\compil{\tm_0} \tomach^* \statetwo$ and $\exec$ extends $\exectwo$ with $\statetwo \tomach \state$.
  By \ih, $\sizefree\statetwo \leq \sizefree{\tm_0} + \size{\tm_0} \cdot \sizem\exectwo - \sizee \exectwo$. Cases (the notation refers to the transitions of the machine, in \reffig{eglamour}):

  \begin{itemize}
    \item\emph{the last transition is a substitution transition.} We have to show $\sizefree\state \leq \sizefree\code + \size{\tm_0} \cdot \sizem\exec - \sizee \exec$. It follows from the \ih and 
    \begin{itemize}
      \item $\sizefree\state = \sizefree\statetwo -1$ because dump and stack do not change and the code changes from a variable (of measure 1) to an abstraction (of measure 0);
      \item $\sizem\exec = \sizem\exectwo$;
      \item $\sizee\exec = \sizee\exectwo + 1$;
    \end{itemize}
    
    \item \emph{the last transition is a $\beta$-transition}. For $\tomachsm$: 

    \begin{align*}
    \sizefree \exec & = \sizefree\dump + \sizefree\stack + \sizefree\code\\
		    & \leq \sizefree\dump + \sizefree{\fire\cons\stack} + \sizefree\code &&\!\mbox{($\sizefree\fire \geq 0$)}\\
		    & = \sizefree\dump + \sizefree{\la\var\code} + \sizefree{{\la\vartwo\codetwo}\cons\stack} + \sizefree\code &&\!\mbox{($\sizefree{\la\var\code} = 0$)}\\
		    & = \sizefree\statetwo + \sizefree\code &&\!\mbox{(def. of $\sizefree\statetwo$)}\\
		    & = \sizefree\statetwo + \size{\tm_0} &&\!\mbox{(\reflemma{subterm-invariant})}\\
		    & \leq \sizefree{\tm_0} + \size{\tm_0}\cdot\sizem\exectwo - \sizee \exectwo + \size{\tm_0} &&\!\mbox{(\ih)}\\
		    & = \sizefree{\tm_0} + \size{\tm_0}\cdot(\sizem\exectwo +1)- \sizee \exectwo \\
		    & = \sizefree{\tm_0} + \size{\tm_0}\cdot\sizem\exec- \sizee \exec 
    \end{align*}

    \item \emph{the last transition is a commutative transition}. Note that (sub)terms and stacks are moved around but never erased, never duplicated, and never modified. Moreover no new pieces of code are introduced, so that the measure never changes. Since also $\sizem\exec$ and $\sizee \exec$ do not change, the statement follows from the \ih \qedhere
  \end{itemize}
\end{proof}

\setcounter{corollaryAppendix}{\value{coro:exp-bilinear}}
\begin{corollaryAppendix}[Bilinear Number of Substitution Transitions]
\label{coroappendix:exp-bilinear}
  Let%
\NoteState{coro:exp-bilinear}
  $\exec: \compil{\tm_0} \tomach^* \state$ be a \eglamour execution.
  Then $\sizee \exec \leq  (1 + \sizem\exec) \cdot \size{\tm_0}$.
\end{corollaryAppendix}

\begin{proof}
  By \reflemma{free-occ-inv}, $\sizee \exec\leq \sizefree{\tm_0} + \size{\tm_0}\cdot\sizem\exec - \sizefree\state$, that implies $\sizee \exec\leq \sizefree{\tm_0} + \size{\tm_0}\cdot\sizem\exec$. 
  The statement follows from the fact that $\sizefree{\tm_0} \leq \size{\tm_0}$.\qedhere
\end{proof}

\setcounter{lemmaAppendix}{\value{l:comm-bound}}
\begin{lemmaAppendix}[Number of Commutative Transitions]
\label{lappendix:comm-bound} % \refcoro{comm-bound}
%   If 
% \NoteState{l:comm-bound}
%   $\state$ is a state reached by an execution $\exec$ of initial code $\code$, then
  For%
\NoteState{l:comm-bound}
  $\exec: \compil{\tm_0} \tomach^* \state$ be an \eglamour execution. Then 
  $\sizecom\exec \leq  \sizecom\exec + \sizecom\state \leq (1+\sizee\exec)\cdot \size{\tm_0} \in O((1+\sizem\exec) \cdot \size{\tm_0}^2)$.
\end{lemmaAppendix}
\begin{proof}
% We prove $\sizecom\exec + \sizecom\state \leq (1+\sizee{\exec})\cdot \size{\tm_0}$ by induction over the length of the derivation. 
% Since $\sizecom\state$ is non-negative, % it immediately follows $\sizecom\exec \leq  \sizecom\exec + \sizecom\state$ and thus the statement. 
First, note that $\sizecom\exec \leq  \sizecom\exec + \sizecom\state$ since $\sizecom\state \geq 0$.
We prove that $\sizecom\exec + \sizecom\state \leq (1+\sizee{\exec})\cdot \size{\tm_0}$ by induction on the length of the execution $\exec$. 

\emph{Base case} (empty execution): then, $\compil{\tm_0} = \state$ and $\sizecom{\exec} = 0 = \sizee{\exec}$, thus the property collapses on the tautology $\size{\code_0} \leq \size{\tm_0}$. 

\emph{Inductive case}: let $\statetwo\tomach \state$ be the
last transition of $\exec$ and let $\exectwo$ be the prefix of $\exec$ ending on $\statetwo$. The statement holds for $\statetwo$ by the \ih, \ie $\sizecom\exectwo + \sizecom\statetwo \leq (1+\sizee\exectwo) \cdot \size{\tm_0}$. We now show that the statement hold by analyzing the various cases of $\statetwo\tomach \state$ and showing that the inequality holds also after the transition:
\begin{itemize}
 \item \emph{Commutative Transitions} $\tomachcone$: the rule splits the code $\code\codetwo$ between the dump and the code. Therefore, $\sizecom\state = \sizecom\statetwo - 1$ while clearly $\sizecom\exec = \sizecom\exectwo + 1$, that is the lhs does not change. The rhs does not change either, and so the inequality is preserved. 
 
 \item \emph{Commutative Transitions} $\tomachcp{_{2,3}}$: these rules consume the current code, so $\sizecom\state \leq \sizecom\statetwo - 1$. Since clearly $\sizecom\exec = \sizecom\exectwo + 1$, it follows that the lhs either decreases or stays the same. The rhs does not change either, and so the inequality is preserved. 
 
 \item \emph{$\beta$-Transition $\tomachm$}: trivial, as the lhs decreases of 1 (because the $\l$ of the abstraction is consumed) and the rhs does not change. 
 
 \item \emph{Substitution Transition $\tomachse$}: it modifies the current code by replacing a variable (of size 1) with an abstraction coming from the environment. Because of the subterm invariant (\reflemma{subterm-invariant}), the abstraction is a subterm of $\tm_0$ and so the increment of the lhs is bounded by $\size{\tm_0}$. We have $\sizee\exec = \sizee\exectwo +1$  and so the rhs increases of $\size{\tm_0}$, that is, the inequality still holds.

\end{itemize}

This ends the proof of $\sizecom\exec + \sizecom\state \leq (1+\sizee{\exec})\cdot \size{\tm_0}$.
Now, substituting the bound given by \refcoro{exp-bilinear} into $\sizecom\exec + \sizecom\state \leq (1+\sizee{\exec})\cdot \size{\tm_0}$ we obtain
    \begin{equation*}
      \sizecom\exec + \sizecom\state  \leq\! (1+\sizee{\exec})\cdot \size{\tm_0} \leq (1+(1+\sizem\exec) \cdot \size{\tm_0}) \cdot \size{\tm_0} = (1+\sizem{\exec}\!)\cdot{\size{\tm_0}}^2 + \size{\tm_0}
    \end{equation*}

Then $\sizecom\exec + \sizecom\state$, and thus $\sizecom\exec$, is bound by $O((1+\sizem{\exec})\cdot{\size{\tm_0}}^2)$.\qedhere
\end{proof}

\setcounter{theoremAppendix}{\value{thm:eglamour-overhead-bound}}
\begin{theoremAppendix}[\eglamour Overhead Bound]
\label{thmappendix:eglamour-overhead-bound}
  Let% 
\NoteState{thm:eglamour-overhead-bound}
$\exec: \compil{\tm_0} \tomach^* \state$ be a \eglamour execution. Then $\exec$ is implementable on RAM in $O((1+\sizem{\exec})\cdot \size{\tm_0}^2)$, \ie linear in the number of $\beta$-transitions (aka the length of the derivation $\deriv$ implemented by $\exec$) and quadratic in the size of the initial term $\tm_0$.
\end{theoremAppendix}

\begin{proof}
  The cost of implementing $\exec$ is the sum of the costs of implementing the $\beta$, substitution, and commutative transitions:
  \begin{enumerate}
    \item \emph{$\beta$-Transition $\tomachm$}: each one costs $O(1)$ and so all together they cost $O(\sizem{\exec})$.
    \item \emph{Substitution Transition $\tomachse$}: by \refcoro{exp-bilinear} we have $\sizee \exec \leq  (1+\sizem\exec)\cdot \size{\tm_0}$, \ie the number of substitution transitions is bilinear. By the subterm invariant (\reflemma{subterm-invariant}), each substitution step costs at most $O(\size{\tm_0})$, and so their full cost is $O((1+\sizem{\exec})\cdot{\size{\tm_0}}^2)$.
    \item \emph{Commutative Transitions $\tomachc$}: by \reflemma{comm-bound} we have $\sizecom{\exec} \leq (1+\sizee{\exec})\cdot \size{\tm_0}$. Now, substituting the bound given by \refcoro{exp-bilinear} we obtain
    $$\sizecom{\exec} \leq\! (1+\sizee{\exec})\cdot \size{\tm_0} \leq (1+(1+\sizem\exec) \cdot \size{\tm_0}) \cdot \size{\tm_0} = (1+\sizem{\exec}\!)\cdot{\size{\tm_0}}^2 + \size{\tm_0}$$
      Since every commutative transition evidently takes constant time, the whole cost of the commutative transitions is bound by $O((1+\sizem{\exec})\cdot{\size{\tm_0}}^2)$.
  \end{enumerate}
  Then the cost of implementing $\exec$ is $O((1+\sizem{\exec})\cdot{\size{\tm_0}}^2)$.\qedhere
\end{proof}

\subsection{Proofs of Section~\ref{sect:fglamour} (\fglamour)}
\label{subsect:fglamour-proofs}
% !TEX root = main.tex
% \setcounter{lemmaAppendix}{\value{l:eglamour-invariants}}
For the \fglamour we proceed like for the \eglamour: first we prove the invariants, and then we use them to prove that it forms an implementation system with respect to right-to-left evaluation $\torf$ in the fireball calculus (via the decoding). The differences are minimal, but we include detailed proofs for the sake of completeness.

\begin{lemma}[\fglamour Invariants]
	\label{l:fglamour-invariants} % \reflemmap{eglamour-invariants}{subterm}
	Let $\state = (\dump,\code,\stack,\genv)$ be a reachable state. Then:
% \NoteState{l:eglamour-invariants}
	\begin{varenumerate}
		\item \emph{Name:} \label{p:fglamour-invariants-name}
		\begin{enumerate}
			\item \emph{Explicit Substitutions}: \label{p:fglamour-invariants-name-es}
			if $\genv = \genvtwo \esub\var\codetwo \genvthree$  then $\var$ is fresh wrt $\codetwo$ and $\genvthree$;

			\item \emph{Abstractions}: \label{p:fglamour-invariants-name-abs}
			if $\la\var\codetwo$ is a subterm of $\dump$, $\codetwo$, $\stack$, or $\genv$ then $\var$ may occur only in $\codetwo$;
		\end{enumerate}
		
		\item \label{p:fglamour-invariants-fireball-stack}\emph{Fireball Item}: $\decode\stackitem$ and $\relunf{\decode\stackitem}\genv$ are:
		\begin{itemize}
		  \item inert terms if $\stackitem = \pair\var\stacktwo$ and either $\genv(\var) = \bot$ or $\genv(\var) = \pair\vartwo\stackthree$,
		  \item abstractions otherwise,
		\end{itemize}
		for every item $\stackitem$ in $\stack$, in $\genv$, and in every stack in $\dump$;
		\item \label{p:fglamour-invariants-ev-ctx}\emph{Contextual Decoding}: $\stctx\state = \relunf{\decdumpp\decstack}\genv$ is a right context;

		%\item \label{p:eglamour-invariants-env-size}\emph{Environment Size}: the length of the global environment $\genv$ is bound by $\sizem\exec$. Moreover, the number of entries in $\genv$ that contain an abstraction is exactlty $\sizemv\exec$.
		
	\end{varenumerate}
\end{lemma}

\begin{proof}
By induction on the length of the execution leading to the reachable state. In an initial state all the invariants trivially hold. For a non-empty execution the proof for every invariant is by case analysis on the last transition, using the \ih. 
\begin{enumerate}
 \item \emph{Name.} Cases:
\begin{enumerate}
 \item $\statetwo = \glamst\dump {\code\codetwo} \stack \genv 
	\tomachcone
	\glamst { \dump\cons\dentry\code\stack } \codetwo \stempty \genv = \state$. Both points follow immediately from the \ih

  \item $\statetwo = 
	  \glamst{ \dump\cons\dentry\code\stack }{ \la\var\codetwo } \stempty  \genv
	  \tomachctwo 
	  \glamst \dump \code {\pair{\la\var\codetwo}{\stempty} \cons\stack} \genv	
	 = \state$. Both points follow immediately from the \ih

  \item $\statetwo = 
	  \glamst {\dump\cons\dentry\code\stack} \var \stacktwo \genv
	  \tomachcthree
	  \glamst \dump \code {\pair{\var}{\stacktwo}\!\cons\!\stack} \genv 
	 = \state$ with $\genv(\var)=\bot$ or $\genv(\var) = \pair{\vartwo}{\stackthree}$ or ($\genv(\var) = \pair{\la\vartwo\codetwo}{\stempty}$ and $\stacktwo = \stempty$). 
	 Both points follow immediately from the \ih

  \item $\statetwo = 
	  \glamst \dump {\la\var\code} {\pair{\vartwo}{\stempty}\cons\stack} \genv
	  \tomachmone
	  \glamst \dump {\code\isub{\var}\vartwo} \stack \genv
	 = \state$. 
	 Both points follow immediately from the \ih
  \item $\statetwo = 
	  \glamst \dump {\la\var\code} {\stackitem\cons\stack} \genv
	  \tomachmtwo
	  \glamst \dump \code \stack {\esub\var\stackitem\genv}
	 = \state$ with $\stackitem \neq \pair{\vartwo}{\stempty}$. \refpoint{fglamour-invariants-name-es} for the new entry in the environment follows from the \ih for \refpoint{fglamour-invariants-name-abs}, for the other entries from the \ih for \refpoint{fglamour-invariants-name-es}. \refpoint{fglamour-invariants-name-abs} follows from its \ih
	  
  \item \begin{equation*}
          \begin{split}
	 \statetwo &= 
	  \glamst \dump \var {\stackitem\cons\stack} {\genv_1\esub\var{\pair{\la\vartwo\codetwo}\stempty}\genv_2} \\
	  &\tomachse
	  \glamst \dump {\rename{\la\vartwo\codetwo}} {\stackitem\cons\stack} {\genv_1\esub\var{\pair{\la\vartwo\codetwo}\stempty}\genv_2}
	   = \state.
	 \end{split}
       \end{equation*} \refpoint{fglamour-invariants-name-es} follows from its \ih. \refpoint{fglamour-invariants-name-abs} for the new code is guaranteed by the $\alpha$-renaming operation $\rename{\la\vartwo\codetwo}$, the rest follows from its \ih
\end{enumerate} 
 
 \item \emph{Fireball Item}. Cases:
\begin{enumerate}
 \item $\statetwo = \glamst\dump {\code\codetwo} \stack \genv 
	\tomachcone
	\glamst { \dump\cons\dentry\code\stack } \codetwo \stempty \genv = \state$. It follows from the \ih

  \item $\statetwo = 
	  \glamst{ \dump\cons\dentry\code\stack }{ \la\var\codetwo } \stempty  \genv
	  \tomachctwo 
	  \glamst \dump \code {\pair{\la\var\codetwo}{\stempty} \cons\stack} \genv	
	 = \state$.  For $\pair{\la\var\codetwo}{\stempty}$ we have that $\decode{\pair{\la\var\codetwo}{\stempty}}$ and $\relunf{\decode{\pair{\la\var\codetwo}{\stempty}}}\genv = \relunf{(\la\var\codetwo)}\genv = \la\var\relunf{\codetwo}\genv$ are abstractions, and hence fireballs.	 
	 For all other items the invariant follows from the \ih

  \item $\statetwo = 
	  \glamst {\dump\cons\dentry\code\stack} \var \stacktwo {\genv}
	  \tomachcthree
	  \glamst \dump \code {\pair{\var}{\stacktwo}\cons\stack} {\genv}
	 \allowbreak= \state$ with $\genv(\var) = \bot$ or $\genv(\var) = {\pair\vartwo\stackthree}$ or ($\genv(\var) = \pair{\la\vartwo\codetwo}{\stempty}$ and $\stacktwo = \stempty$). 
	 For $\pair\var\stacktwo$, we have that $\decode{\pair\var\stacktwo} = \ctxholep{\var}\decode\stacktwo$ and $\relunf{\decode{\pair\var\stacktwo}}\genv = \ctxholep{\relunf{\var}\genv}(\relunf{\decode\stacktwo}\genv)$.
	 By \ih, $\decode\stackitemtwo$ is a fireball for every item $\stackitemtwo$ in $\stacktwo$. 
	 Therefore, $\decode{\pair\var\stacktwo}$ is an inert term.
	 Concerning $\relunf{\decode{\pair\var\stacktwo}}\genv$, there are three subcases:
	 \begin{enumerate}
	   \item $\genv(\var) = \pair{\vartwo}{\stackthree}$ \ie $\genv \defeq \genv_1\esub\var{\pair\vartwo\stackthree}\genv_2$. 
	   By \reflemmapp{fglamour-invariants}{name}{es}, every ES in $\genv$ binds a different variable, so $\relunf{\var}\genv = \relunf{\var}{\genv_1\esub\var{\pair\vartwo\stackthree}\genv_2} = \relunf{\var}{\genv_1}\relunf{\isub\var{\decode{\pair{\vartwo}{\stackthree}}}}{\genv_2} = \relunf{\decode{\pair{\vartwo}{\stackthree}}}{\genv_2} = \relunf{\decode{\pair{\vartwo}{\stackthree}}}{\genv}$, that by \ih is an inert term. 
	   Moreover, the \ih also gives that $\relunf{\stackitemtwo}\genv$ is a fireball for every item $\stackitemtwo$ in $\stacktwo$. 
	   Therefore  $\relunf{\decode{\pair\var\stacktwo}}\genv = \ctxholep{\relunf{\var}\genv}(\relunf{\decode\stacktwo}\genv)$ is an inert term. 

	   \item $\genv(\var)=\bot$. Similar to the previous case. 
	   By hypothesis, we have $\relunf{\var}\genv = \var$. 
	   As before, by \ih $\relunf{\decode\stackitemtwo}\genv$ is a  fireball for every item $\stackitemtwo$ in $\stacktwo$. 
	   So,  $\relunf{\decode{\pair\var\stacktwo}}\genv = \ctxholep{\relunf{\var}\genv}(\relunf{\decode\stacktwo}\genv)  = \ctxholep{\var}(\relunf{\decode\stacktwo}\genv)$ is an inert term. 
	   
	   \item $\genv(\var) = \pair{\la\vartwo\codetwo}\stempty$ (\ie $\genv = \genv_1\esub\var{\pair{\la\vartwo\codetwo}{\stempty}}\genv_2$) and $\stacktwo = \stempty$.
	   Then $\decode{\pair{\var}{\stacktwo}} = \var$.
	   By \reflemmapp{fglamour-invariants}{name}{es}, every ES in $\genv$ binds a different variable, so $\relunf{\var}\genv = \relunf{\var}{\genv_1\esub\var{\pair{\la\vartwo\codetwo}\stempty}\genv_2} = \relunf{\var}{\genv_1}\relunf{\isub\var{\decode{\pair{\la\vartwo\codetwo}\stempty}}}{\genv_2} = \relunf{\decode{\pair{\la\vartwo\codetwo}\stempty}}{\genv_2} = \la\vartwo\relunf{\codetwo}{\genv}$. 
	   Therefore  $\relunf{\decode{\pair\var\stacktwo}}\genv = \relunf{\var}\genv = \la\vartwo\relunf{\codetwo}{\genv}$ is an abstraction.

	 \end{enumerate}
	 For all other items in $\state$ the invariant follows from the \ih

  \item $\statetwo = 
	  \glamst \dump {\la\var\code} {\pair{\vartwo}\stempty\cons\stack} \genv
	  \tomachmone
	  \glamst \dump {\code\isub{\var}\vartwo} \stack {\genv}
	 = \state$. Then the invariant follows immediately from the \ih
  \item $\statetwo = 
	  \glamst \dump {\la\var\code} {\stackitem\cons\stack} \genv
	  \tomachmtwo
	  \glamst \dump \code \stack {\esub\var\stackitem\genv}
	 = \state$ with $\stackitem \neq \pair{\vartwo}{\stempty}$. By \reflemmapp{fglamour-invariants}{name}{abs} $\var$ may occur only in $\code$. 
	 Thus the substitution $\relunf{}{\esub\var\stackitem\genv}$ acts exactly as $\relunf{}\genv$ on every item in $\state$. Then the invariant follows from the \ih
	  
  \item \begin{equation*}
	 \begin{split}
	 \statetwo &= 
	  \glamst \dump \var {\stackitem\cons\stack} {\genv_1\esub\var{\pair{\la\vartwo\codetwo}\stempty}\genv_2} \\
	  &\tomachse
	  \glamst \dump {\rename{\la\vartwo\codetwo}} {\stackitem\cons\stack} {\genv_1\esub\var{\pair{\la\vartwo\codetwo}\stempty}\genv_2}
	 = \state.
         \end{split}
       \end{equation*}
 It follows from the \ih

\end{enumerate} 

\item \emph{Contextual Decoding}. Cases:
\begin{enumerate}
 \item $\statetwo = \glamst\dump {\code\codetwo} \stack \genv 
	\tomachcone
	\glamst { \dump\cons\dentry\code\stack } \codetwo \stempty \genv = \state$. By \ih $\stctx\statetwo = \relunf{\decdumpp\decstack}\genv$ is a right context, as well as $\relunf{\codetwo}\genv \ctxhole$. Then their composition $(\relunf{\decdumpp\decstack}\genv)\ctxholep{\relunf{\codetwo}\genv \ctxhole} \allowbreak= \relunf{\decdumpp{\decstackp{\codetwo \ctxhole}}}\genv = \stctx\state$ is a right context.

  \item $\statetwo = 
	  \glamst{ \dump\cons\dentry\code\stack }{ \la\var\codetwo } \stempty  \genv
	  \tomachctwo 
	  \glamst \dump \code {\pair{\la\var\codetwo}\stempty \cons\stack} \genv	
	 \allowbreak= \state$. 
	 By \ih $\stctx\statetwo = \relunf{\decode{\dump\cons\dentry\code\stack}}\genv = \relunf{\decdumpp{\decodeinv\stack{\code\ctxhole}}}\genv$ is a right context, that implies that $\relunf{\decdumpp{\decstack}}\genv$ is one such context as well. So, $\stctx\statetwo \!= \relunf{\decdumpp{\decode{\pair{\la\var\codetwo}\stempty \cons\stack}}}\genv = \relunf{\decdumpp{\decodeinv\stack{\ctxhole\la\var\codetwo}}}\genv \allowbreak= (\relunf{\decdumpp{\decstack}}\genv)\ctxholep{\ctxhole\relunf{\la\var\codetwo}\genv}$ is a right context, because it is the composition of right context, given that $\relunf{\la\var\codetwo}\genv$ is a fireball.

%   \item $\statetwo = 
% 	  \glamst {\dump\cons\dentry\code\stack} \var \stacktwo {\genv_1\esub\var{\pair{\vartwo}{\stackthree}}\genv_2}
% 	  \tomachcthree
% 	  \glamst \dump \code {\pair\var\stacktwo \!\cons\stack} {\genv_1\esub\var{\pair{\vartwo}{\stackthree}}\genv_2}
% 	 \allowbreak= \state$. Let $\genv \defeq \genv_1\esub\var{\pair{\vartwo}{\stackthree}}\genv_2$.
% 	 By \ih $\stctx\statetwo = \relunf{\decodep{\dump\cons\dentry\code\stack}{\decode\stacktwo}}\genv = \relunf{\decdumpp{\decodeinv\stack{\code(\decode\stacktwo)}}}\genv$ is a right context, that implies that $\relunf{\decdumpp{\decstack}}\genv$ is one such context as well. Then 
% 	 $\stctx\state = \relunf{\decdumpp{\decode{\pair\var\stacktwo\cons\stack}}}\genv = 
% 	 \relunf{\decdumpp{\decodeinv\stack{\ctxhole\decode{\pair\var\stacktwo}}}}\genv = 
% 	 (\relunf{\decdumpp{\decstack}}\genv)\ctxholep{\relunf{\ctxhole\decode{\pair\var\stacktwo}}\genv}$ is a right context, because it is the composition of right context, given that $\decode{\pair\var\stacktwo}\indsub\genv$ is a fireball by \refpoint{eglamour-invariants-fireball-stack}.
% 
%   \item $\statetwo = 
% 	  \glamst {\dump\cons\dentry\code\stack} \var \stacktwo \genv
% 	  \tomachcfour
% 	  \glamst \dump \code {\pair{\var}{\stacktwo}\cons\stack} \genv 
% 	 = \state$ with $\genv(\var)=\bot$. Exactly as the previous case.
  \item $\statetwo = 
	  \glamst {\dump\cons\dentry\code\stack} \var \stacktwo {\genv}
	  \tomachcthree
	  \glamst \dump \code {\pair\var\stacktwo \!\cons\!\stack} {\genv}
	 \allowbreak= \state$ with $\genv(\var)=\bot$ or $\genv(\var)= \pair{\vartwo}{\stackthree}$ or ($\genv(\var) = \pair{\la\vartwo\codetwo}\stempty$ and $\stacktwo = \stempty$). 
	 By \ih $\stctx\statetwo = \relunf{\decodep{\dump\cons\dentry\code\stack}{\decode\stacktwo}}\genv = \relunf{\decdumpp{\decodeinv\stack{\code(\decode\stacktwo)}}}\genv$ is a right context, that implies that $\relunf{\decdumpp{\decstack}}\genv$ is one such context as well. Then 
	 $\stctx\state = \relunf{\decdumpp{\decode{\pair\var\stacktwo\cons\stack}}}\genv = 
	 \relunf{\decdumpp{\decodeinv\stack{\ctxhole\decode{\pair\var\stacktwo}}}}\genv = 
	 (\relunf{\decdumpp{\decstack}}\genv)\ctxholep{\relunf{\ctxhole\decode{\pair\var\stacktwo}}\genv}$ is a right context, because it is the composition of right context, given that $\relunf{\decode{\pair\var\stacktwo}}\genv$ is a fireball by \reflemmap{fglamour-invariants}{fireball-stack}.

  \item $\statetwo = 
	  \glamst \dump {\la\var\code} {\pair{\vartwo}{\stempty}\!\cons\!\stack} \genv
	  \tomachmone
	  \glamst \dump {\code\isub{\var}{\vartwo}} \stack {\genv}
	 = \state$. 
	  By the \ih $\stctx\statetwo = \relunf{\decdumpp{\decode{\pair{\vartwo}{\stempty}\!\cons\!\stack}}}\genv %= \decstackp{\ctxhole\decode{\stackitem}}
	  $ is a right context, that implies that $\stctx\state = \relunf{\decdumpp{\decstack}}\genv$ is one such context as well. 
  \item $\statetwo = 
	  \glamst \dump {\la\var\code} {\stackitem\!\cons\!\stack} \genv
	  \tomachmtwo
	  \glamst \dump \code \stack {\esub\var\stackitem\genv}
	 = \state$ with $\stackitem \neq \pair\vartwo\stacktwo$. 
	  By the \ih $\stctx\statetwo = \relunf{\decdumpp{\decode{\stackitem\cons\stack}}}\genv% = \relunf{\decdumpp{\decstackp{\ctxhole\decode{\stackitem}}}}\genv
	  $ is a right context, that implies that $\relunf{\decdumpp{\decstack}}\genv$ is one such context as well. Now, note that $\stctx\state = \relunf{\decdumpp{\decstack}}{\esub\var\stackitem\genv} = \relunf{\decdumpp{\decstack}}\genv$ because by \reflemmapp{fglamour-invariants}{name}{abs} $\var$ may occur only in $\code$, and so the substitution $\relunf{}{\esub\var\stackitem\genv}$ acts on every code in $\dump$ and $\stack$ exactly as $\relunf{}\genv$. 
	  
  \item \begin{equation*}
          \begin{split}
	 \statetwo &= 
	  \glamst \dump \var {\stackitem\cons\stack} {\genv_1\esub\var{\pair{\la\vartwo\codetwo}\stempty}\genv_2} \\
	  &\tomachse 
	  \glamst \dump {\rename{\la\vartwo\codetwo}} {\stackitem\cons\stack} {\genv_1\esub\var{\pair{\la\vartwo\codetwo}\stempty}\genv_2}
	 = \state.
	  \end{split}
        \end{equation*} 
	 It follows by the \ih because $\stctx\statetwo = \stctx\state$, as the only component that changes is the code.
  \qedhere
\end{enumerate}
\end{enumerate}
\end{proof}

\begin{lemma}[\fglamour $\beta$-Projection and Overhead Transparency]
  \label{l:fglamour-trans-projection} % \reflemmap{eglamour-trans-projection}{exp-com}
	Let 
% \NoteState{l:fglamour-trans-projection}
	$\state$ be a reachable state.
	\begin{varenumerate}
		\item \label{p:fglamour-trans-projection-exp-com} \emph{Overhead Transparency}: 
		if $\state\tomachhole{\ssym,\csym_{1,2,3}}\statetwo$ then $\decode\state=\decode\statetwo$;
		\item \label{p:fglamour-trans-projection-mult}\emph{$\beta$-Projection}: if $\state\tomachhole{\beta_{1,2}} \statetwo$ then $\decode\state\torf\decode\statetwo$.	  
		
	\end{varenumerate}
\end{lemma}

\begin{proof}
Transitions:
\begin{enumerate}
 \item $\state = \glamst\dump {\code\codetwo} \stack \genv 
	\tomachcone
	\glamst { \dump\cons\dentry\code\stack } \codetwo \stempty \genv = \statetwo$.
	Then
	\[\begin{array}{rclclcl}
	   \decode\state 
	   & = & \relunf{\decodep{\dump}{\decodeinv\stack{\code\codetwo}} }\genv\\
	   & = & \relunf{\decodep{\dump\cons\dentry\code\stack}{\codetwo} }\genv \\
 	   & = & \relunf{\decodep{\dump\cons\dentry\code\stack}{\decodeinv\stempty\codetwo} }\genv 
 	   & = & \decode\statetwo
	  \end{array}\]

  \item $\state = 
	  \glamst{ \dump\cons\dentry\code\stack }{ \la\var\codetwo } \stempty  \genv
	  \tomachctwo 
	  \glamst \dump \code {\pair{\la\var\codetwo}{\stempty} \cons\stack} \genv	
	 = \statetwo$.
	 Then
	\[\begin{array}{rclclcl}
	   \decode\state 
	   & = & \relunf{\decodep{\dump\cons\dentry\code\stack}{\decodeinv\stempty{\la\var\codetwo}}}\genv\\	   
	   & = & \relunf{\decodep{\dump}{\decodeinv{\stack}{\code(\decodeinv\stempty{\la\var\codetwo})}}}\genv \\
	   & = & \relunf{\decodep{\dump}{\decodeinv{\pair{\la\var\codetwo}{\stempty}\cons\stack}\code}}\genv
 	   & = & \decode\statetwo
	  \end{array}\]

%   \item $\state = 
% 	  \glamst {\dump\cons\dentry\code\stack} \var \stacktwo {\genv_1\esub\var{\pair{\vartwo}{\stackthree}}\genv_2}
% 	  \tomachcthree
% 	  \glamst \dump \code {\pair{\var}{\stacktwo}\cons\stack} {\genv_1\esub\var{\pair{\vartwo}{\stackthree}}\genv_2}
% 	 = \statetwo$. Let $\genvtwo \defeq \genv_1\esub\var{\pair{\vartwo}{\stackthree}}\genv_2$
% 	  Then
% 	\[\begin{array}{rclclcl}
% 	   \decode\state 
% 	   & = & \decodep{\dump\cons\dentry\code\stack}{\decodeinv\stacktwo\var}\indsub\genvtwo\\
% 	   & = & \decodep{\dump}{\decodeinv\stack{\code(\decodeinv\stacktwo\var)}}\indsub\genvtwo\\
% 	   & = & \decodep{\dump}{\decodeinv{\pair\var\stacktwo \cons\stack}\code}\indsub\genvtwo
%  	   & = & \decode\statetwo
% 	  \end{array}\]
% 
%   \item $\state = 
% 	  \glamst {\dump\cons\dentry\code\stack} \var \stacktwo \genv
% 	  \tomachcfour
% 	  \glamst \dump \code {\pair{\var}{\stacktwo}\cons\stack} \genv 
% 	 = \statetwo$ with $\genv(\var)=\bot$.
% 	  Then
% 	\[\begin{array}{rclclcl}
% 	   \decode\state 
% 	   & = & \decodep{\dump\cons\dentry\code\stack}{\decodeinv\stacktwo\var}\indsub{\genv}\\
% 	   & = & \decodep{\dump}{\decodeinv\stack{\code(\decodeinv\stacktwo\var)}}\indsub{\genv}\\
% 	   & = & \decodep{\dump}{\decodeinv{\pair\var\stacktwo \cons\stack}\code}\indsub{\genv}
%  	   & = & \decode\statetwo
% 	  \end{array}\]
  \item $\state = 
	  \glamst {\dump\cons\dentry\code\stack} \var \stacktwo {\genv}
	  \tomachcthree
	  \glamst \dump \code {\pair{\var}{\stacktwo}\cons\stack} {\genv}
	 = \statetwo$ with $\genv(\var) = \bot$ or $\genv(\var) = {\pair{\vartwo}{\stackthree}}$ or ($\genv(\var) = \pair{\la\vartwo\codetwo}{\stempty}$ and $\stacktwo = \stempty$).
	  Then
	\[\begin{array}{rclclcl}
	   \decode\state 
	   & = & \relunf{\decodep{\dump\cons\dentry\code\stack}{\decodeinv\stacktwo\var}}\genv\\
	   & = & \relunf{\decodep{\dump}{\decodeinv\stack{\code(\decodeinv\stacktwo\var)}}}\genv\\
	   & = & \relunf{\decodep{\dump}{\decodeinv{\pair\var\stacktwo \cons\stack}\code}}\genv
 	   & = & \decode\statetwo
	  \end{array}\]

  \item $\state = 
	  \glamst \dump {\la\var\code} {\pair{\vartwo}{\stempty}\cons\stack} \genv
	  \tomachmone
	  \glamst \dump {\code\isub{\var}{\vartwo}} \stack {\genv}
	 = \statetwo$.
	  Then
	\[\begin{array}{rclclcl}
	   \decode\state 
	   & = & \relunf{\decodep{\dump}{\decodeinv{\pair{\vartwo}{\stempty}\cons\stack}{\la\var\code}}}{\genv}\\
	   & = & \relunf{\decodep{\dump}{\decodeinv{\stack}{(\la\var\code)\vartwo}}}{\genv}\\
	   & \torf & \relunf{\decodep{\dump}{\decodeinv{\stack}{\code \isub\var{\vartwo}}}}{\genv}
 	   & = & \decode\statetwo
	  \end{array}\]
	  where the rewriting step takes place because $\relunf{\decodep{\dump}{\decode\stack}}{\genv}$ is a right context by \reflemmap{eglamour-invariants}{ev-ctx}.

  \item $\state = 
	  \glamst \dump {\la\var\code} {\stackitem\cons\stack} \genv
	  \tomachmtwo
	  \glamst \dump \code \stack {\esub\var\stackitem\genv}
	 = \statetwo$ with $\stackitem \neq \pair{\vartwo}{\stempty}$.
	  Then
	\[\begin{array}{rclclcl}
	   \decode\state 
	   & = & \relunf{\decodep{\dump}{\decodeinv{\stackitem\cons\stack}{\la\var\code}}}{\genv}\\
	   & = & \relunf{\decodep{\dump}{\decodeinv{\stack}{(\la\var\code)\decode{\stackitem}}}}{\genv}\\
	   & \torf & \relunf{\decodep{\dump}{\decodeinv{\stack}{\code \isub\var{\decode\stackitem}}}}{\genv}\\
	   & = & \relunf{\decodep{\dump}{\decodeinv{\stack}\code}\isub\var{\decode\stackitem}}{\genv}\\
	   & = & \relunf{{\dump}{\decodeinv{\stack}\code}}{\esub\var\stackitem\genv}
 	   & = & \decode\statetwo
	  \end{array}\]
	  where the rewriting step takes place because
	  \begin{enumerate}
	   \item $\relunf{\decodep{\dump}{\decode\stack}}{\genv}$ is a right context by \reflemmap{eglamour-invariants}{ev-ctx};
	   \item $\decode\stackitem$ is a fireball by \reflemmap{eglamour-invariants}{fireball-stack}.
	  \end{enumerate}
	  Moreover, the meta-level substitution $\isub\var{\decode\stackitem}$ can be extruded (in the equality step after the rewriting) without renaming $\var$, because by \reflemmapp{eglamour-invariants}{name}{abs} $\var$ does not occur in $\dump$ nor $\stack$.
	  
  \item $\state =  \glamst \dump \var {\stackitem\!\cons\!\stack} \genv \tomachse
	  \glamst \dump {\rename{\la\vartwo\codetwo}} {\stackitem\!\cons\!\stack} \genv
	 = \statetwo$ with $\genv = {\genv_1\esub\var{\pair{\la\vartwo\codetwo}\stempty}\genv_2}$.
	   Then
	\[\begin{array}{rclclcl}
	   \decode\state 
	   & = & \relunf{\decodep{\dump}{\decodeinv{\stackitem\cons\stack}\var}}{\genv}\\
	   & = & \decodep{\relunf{\dump}{\genv}}{\decodeinv{\relunf{\stackitem\cons\stack}{\genv}}{\relunf{\var}{\genv}}}\\
	   & = & \decodep{\relunf{\dump}{\genv}}{\decodeinv{\relunf{{\stackitem\cons\stack}}{\genv}}{\relunf{{\la\vartwo\codetwo}}{\genv}}}\\
	   & = & \relunf{\decodep{\dump}{\decodeinv{\stackitem\cons\stack}{\la\vartwo\codetwo}}}{\genv}
 	   & = & \decode\statetwo
	  \end{array}\]
\qedhere
\end{enumerate} 
\end{proof}

\begin{lemma}[\fglamour Progress]
\label{l:fglamour-progress}
  Let
%   \NoteState{l:fglamour-progress}
  $\state$ be a reachable final state. Then $\decode\state$ is a fireball, \ie it is $\betaf$-normal. 
\end{lemma}

\begin{proof}
 An immediate inspection of the transitions shows that in a final state the code cannot be an application and the dump is necessarily empty. In fact, final states have one of the following two shapes:
 \begin{enumerate}
  \item \emph{Top-Level Unapplied Abstraction}, \ie $\state = \glamst\stempty{\la\var\code}\stempty\genv$. Then $\decode\state = \relunf{(\la\var\code)}\genv = \la\var\relunf{\code}\genv$ that is a fireball.
  
  \item \emph{Top-Level Free Variable}, \ie $\state = \glamst\stempty\var\stempty\genv$ (note that, differently from what happens in the \eglamour Machine, it might be that $\genv(\var) \neq \bot$).
  We claim that $\decode\state = \relunf{\var}\genv$ is a fireball. 
  Indeed, according to the fireball invariant item (\reflemmap{fglamour-invariants}{fireball-stack}), $\decode\stackitem$ is a fireball for any item $\stackitem$ in $\genv$; thus, the only possibility to have $\relunf{\var}\genv$ different from a fireball is that there is an item $\stackitem$ in $\genv$ such that $\relunf{\decode\stackitem}\genv$ is not a fireball, but this is impossible by \reflemmap{fglamour-invariants}{fireball-stack}.
  \item \emph{Top-Level %Free Variable or 
  Compound Inert Term}, \ie $\state = \glamst\stempty\var{\stackitem \!\cons\! \stack}\genv$ with $\genv(\var) \neq \pair{\la\vartwo\code}\stempty$. Subcases:
  \begin{enumerate}
    \item $\genv(\var) = \undef$.   
    Then $\decode\state = \relunf{(\decodeinv\stack\var)}\genv = \ctxholep{\relunf{\var}\genv}(\relunf{\decode\stack}\genv) = \ctxholep\var(\relunf{\decode\stack}\genv)$. 
    Now, by the fireball item invariant (\reflemmap{fglamour-invariants}{fireball-stack}) every element of $\relunf{\decode\stack}\genv$ is a fireball, so $\ctxholep\var(\relunf{\decode\stack}\genv)$ is an inert term, and hence a fireball.
    \item $\genv(\var) = \pair{\vartwo}\stacktwo$.   Then $\decode\state = \relunf{(\decodeinv\stack\var)}\genv = \ctxholep{\relunf{\var}\genv}(\relunf{\decode\stack}\genv) = \ctxholep{\ctxholep\vartwo(\relunf{\decode\stacktwo}\genv)}(\relunf{\decode\stack}\genv)$. 
    By the fireball item invariant (\reflemmap{fglamour-invariants}{fireball-stack}), any element of $\relunf{\decode\stack}\genv$ and $\relunf{\decode\stacktwo}\genv$ is a fireball, so $\ctxholep{\ctxholep\vartwo(\relunf{\decode\stacktwo}\genv)}(\relunf{\decode\stack}\genv)$ is an inert term, and hence a fireball.
  \qedhere
  \end{enumerate}

  \end{enumerate}
\end{proof}

\setcounter{theoremAppendix}{\value{thm:fglamour-implementation}}
\begin{theoremAppendix}[\fglamour Implementation]
  \label{thmappendix:fglamour-implementation}
  The 
  \NoteState{thm:fglamour-implementation}
  \fglamour implements right-to-left evaluation $\torf$ in $\firecalc$ (via the decoding $\decode\cdot$).
\end{theoremAppendix}

\begin{proof}
  According to \refthm{abs-impl}, it is enough to show that the \fglamour and the right-to-left evaluation $\torf$ and the decoding $\decode\cdot$ form an implementation system, \ie that the %following
  five conditions in \refdef{implementation} hold%:
  .
%   \begin{varenumerate}
% 	\item \emph{$\beta$-Projection}: $\state \tomachhole{\beta_{1,2}} \statetwo$ implies $\decode\state \tostrat \decode\statetwo$;
% 	\item \emph{Overhead Transparency}: $\state \tomachhole{\ssym, \csym_{1,2,3}} \statetwo$ implies $\decode\state = \decode\statetwo$;
% 	\item 	\emph{Overhead Transitions Terminate}:  $\tomachhole{\ssym, \csym_{1,2,3}}$ terminates;
% 	\item \emph{Determinism}: both \fglamour and $\torf$ are deterministic;
% 	\item \emph{Progress}: \fglamour final states decode to $\torf$-normal terms.
%   \end{varenumerate}
% 
  Note that substitution ($\tomachse$) and commutative ($\tomachhole{\csym_{1,2,3}}$) transitions are considered as overhead transitions, whereas the $\beta$ transitions are $\tomachhole{\beta_1}$ and $\tomachhole{\beta_2}$.  
%   We shall prove each point above separately.
  \begin{varenumerate}
	\item \emph{$\beta$-Projection}: $\state \tomachhole{\beta_{1,2}} \statetwo$ implies $\decode\state \tostrat \decode\statetwo$ by \reflemmap{fglamour-trans-projection}{mult} 
	(recall that $\tomachhole{\beta_{1,2}} \ = \ \tomachhole{\beta_1} \cup \tomachhole{\beta_2}$ according to Note~\ref{note:union-transitions}).
	\item \emph{Overhead Transparency}: $\state \tomachhole{\ssym, \csym_{1,2,3}} \statetwo$ implies $\decode\state = \decode\statetwo$ by \reflemmap{fglamour-trans-projection}{exp-com} 
	(recall that $\tomachhole{\ssym, \csym_{1,2,3}} \ = \ \tomachhole{\ssym} \cup \tomachhole{\csym_1} \cup \tomachhole{\csym_2} \cup \tomachhole{\csym_3}$ according to Note~\ref{note:union-transitions}).
	\item 	\emph{Overhead Transitions Terminate}: Termination of  $\tomachhole{\ssym, \csym_{1,2,3}}$ is given by forthcoming \reflemma{exp-bilinear}, which is postponed because they actually give precise complexity bounds, not just termination.
	\item \emph{Determinism}: The \fglamour machine is deterministic, as it can be seen by an easy inspection of the transitions (see \reffig{fglamour}).
% 	(recall that a code is a term and then it is either a variable, or an abstraction, or an application). 
	\reflemmap{prop-of-torf}{determ} proves that $\torf$ is deterministic.
	\item \emph{Progress}: Let $\state$ be a \fglamour final states. By \reflemma{fglamour-progress}, $\decode\state$ is a $\betaf$-normal term, in particular it is $\torf$-normal because $\torf \, \subseteq \, \tof$.
	\qedhere
  \end{varenumerate}
\end{proof}

\paragraph{Complexity Analysis of the \fglamour.} As explained in the paper, the complexity analysis of the \fglamour is essentially trivial. Here we add a few details to convince the skeptical reader.

\setcounter{lemmaAppendix}{\value{l:exp-linear}}
\begin{lemmaAppendix}[Number of Overhead Transitions]
\label{lappendix:exp-bilinear}
  Let%
\NoteState{l:exp-bilinear}
	$\exec: \compil{\tm_0} \tomach^* \state$ be a \fglamour execution.Then 
  \begin{varenumerate}
  \item \emph{Substitution $vs$ $\beta$ Transitions}: $\sizee \exec \leq  \sizem\exec$.
  \item \emph{Commutative $vs$ Substitution Transitions}: $\sizecom\exec \leq (1+\sizee{\exec})\cdot \size{\tm_0} \leq (1+\sizem{\exec})\cdot \size{\tm_0}$.
  \end{varenumerate}
\end{lemmaAppendix}

\begin{proof}
  \begin{varenumerate}
  \item \emph{Substitution $vs$ $\beta$ Transitions}:   since abstractions are substituted on-demand, every substitution transition is followed by a $\beta$-transition. Therefore, in an execution $\exec$ there can be at most one substitution transition not followed by a $\beta$-transition, and so   $\sizee \exec \leq  \sizem\exec + 1$. Now, note that executions start on initial states, \ie on states with empty environments where substitution transitions are not possible. So in $\exec$ there must be a $\beta$-transition before any other substitution transition. The $+1$ can then be removed, obtaining  $\sizee \exec \leq  \sizem\exec$.
  
  \item \emph{Commutative $vs$ Substitution Transitions}: the bound $\sizecom\exec \leq (1+\sizee{\exec})\cdot \size{\tm_0}$, that is the same as in the \eglamour, is obtained in exactly the same way, by using the commutative size measure defined in \refsect{eg-compl-anal}. The differences in the proof are minimal:
  \begin{itemize}
   \item {Transitions $\tomachcone$ and $\tomachctwo$}: no difference, because they are exactly the same transitions of the \eglamour.
   \item {Transition $\tomachcthree$}: the transition has a side-condition more than the same transition of the \eglamour. Than it is a sub-case, and so the bound obviously hold.
   \item {Transition $\tomachmone$}: the novelty of the transition is the renaming of the code, but it lets the size, and thus the measure, unchanged.
    \item {Transition $\tomachmtwo$}: a special case of $\tomachm$ of the \eglamour.
    \item {Transition $\tomachse$}: a special case of $\tomachse$  of the \eglamour.
  \end{itemize}
Last, the inequality $(1+\sizee{\exec})\cdot \size{\tm_0} \leq (1+\sizem{\exec})\cdot \size{\tm_0}$ is obtained by applying the first point of the lemma.
  \qedhere
  \end{varenumerate}

\end{proof}

\setcounter{theoremAppendix}{\value{thm:fglamour-overhead-bound}}
\begin{theoremAppendix}[\fglamour Bilinear Overhead]
\label{thmappendix:fglamour-overhead-bound}
  Let% 
\NoteState{thm:fglamour-overhead-bound}
  $\exec: \compil{\tm_0} \tomach^* \state$ be a \fglamour execution. Then $\exec$ is implementable on RAM in $O((1+\sizem{\exec})\cdot \size{\tm_0})$, \ie linear in the number of $\beta$-transitions and the size of the initial term.
\end{theoremAppendix}

\begin{proof}
    The cost of implementing $\exec$ is the sum of the costs of implementing the $\beta$, substitution, and commutative transitions:
  \begin{varenumerate}
    \item \emph{$\beta$-Transitions $\tomachmone$ and $\tomachmtwo$}: $\tomachmone$  costs $O(\size{\tm_0})$ because the code has to be renamed and by the subterm invariant the size of the code is bound by $\size{\tm_0}$. Transition $\tomachmone$ instead takes constant time. In the worst case all together they cost $O(\sizem{\exec}\cdot\size{\tm_0})$.
    \item \emph{Substitution Transition $\tomachse$}: by \reflemma{exp-bilinear} we have $\sizee \exec \leq  \sizem\exec$. By the subterm invariant (\reflemma{subterm-invariant}), each substitution step costs at most $O(\size{\tm_0})$, and so their full cost is $O(\sizem{\exec}\cdot{\size{\tm_0}})$.
    \item \emph{Commutative Transitions $\tomachc$}: by \reflemma{exp-bilinear} $\sizecom\exec \leq (1+\sizem{\exec})\cdot \size{\tm_0}$.
      Since every commutative transition evidently takes constant time, the whole cost of the commutative transitions is bound by $O((1+\sizem{\exec})\cdot{\size{\tm_0}})$.
  \end{varenumerate}
  Then, the cost of implementing $\exec$ is $O((1+\sizem{\exec})\cdot{\size{\tm_0}})$.\qedhere

\end{proof}

\subsection{Proofs of Section~\ref{sect:conclusions} (Conclusions)}

Here we give the proof of size explosion for the family given at the end of the paper. Let us recall the definition of the family. 

Let the identity combinator be $I \defeq \la\varthree \varthree$ (it can in fact be replaced by any closed abstraction). Define
\begin{align*}
\tmthree_1 & \defeq \la{\var}\la{\vartwo}(\vartwo \var \var)  & \tmfour_0 & \defeq I \\
\tmthree_{n+1} & \defeq \la{\var}(\tmthree_n (\la{\vartwo}(\vartwo \var \var))) & \tmfour_{n+1} & \defeq \la{\vartwo}(\vartwo \tmfour_{n} \tmfour_{n})	
\end{align*}

The size exploding family is $\{\tmthree_n I\}_{n\in\nat}$, \ie it is obtained by applying $\tmthree_n$ to the identity $I = \tmfour_0$. The statement we are going to prove is in fact more general than the one given in the paper, it is about $\tmthree_n \tmfour_m$ instead of just $\tmthree_n I$, in order to obtain a simple inductive proof.
\setcounter{propositionAppendix}{\value{prop:abs-size-explosion}}
\begin{propositionAppendix}{Abstraction Size Explosion}
\label{propappendix:abs-size-explosion}
  Let 
  \NoteState{prop:abs-size-explosion}
  $n \!>\! 0$. Then $\tmthree_n \tmfour_m \tobabs^n  \tmfour_{n+m}$, and in particular $\tmthree_n I \tobabs^n  \tmfour_n$. Moreover, $\size{\tmthree_n I} = O(n)$, $\size{\tmfour_n} = \Omega(2^n)$, $\tmthree_n I$ is closed, and $\tmfour_n$ is normal.
\end{propositionAppendix}

\begin{proof}
  By induction on $n > 0$. 

  The base case: $\tmthree_1 \tmfour_m = \la{\var}\la{\vartwo}(\vartwo \var \var) \tmfour_m \tobabs (\la{\vartwo}(\vartwo \tmfour_m \tmfour_m)) = \tmfour_{m+1}$.
  The inductive case: $\tmthree_{n+1} \tmfour_m = \la{\var}(\tmthree_n (\la{\vartwo}(\vartwo \var \var))) \tmfour_m  \tobabs \tmthree_n (\la{\vartwo}(\vartwo \tmfour_m \tmfour_m)) = \tmthree_n \tmfour_{m+1} \tobabs^n \tmfour_{n+m+1}$, where the second sequence is obtained by the \ih The rest of the statement is immediate.
\end{proof}

The family $\{\tmthree_n I\}_{n\in\nat}$ is interesting because no matter how one looks at it, it always explodes: if evaluation is weak (\ie it does not go under abstraction) there is only one possible derivation to normal form and if it is strong (\ie unrestricted) all derivations have the same length (and are permutatively equivalent). Last, note that it is an example of size explosion also for \ccbv, because the steps are weak and the term is closed. 
Note why machines for \ccbv and \ocbv are not concerned with the question of substituting abstractions on-demand: the exponential number of substitutions of abstractions required by the evaluation of the family are all substitutions under abstraction, and so closed and open machines never do them anyway.

}

\end{document}